\documentclass[aps,pra,reprint,showpacs,superscriptaddress,nofootinbib,onecolumn]{revtex4-2}
\usepackage[utf8]{inputenc}
\usepackage[T1]{fontenc}
\usepackage{amsmath, amsthm, amssymb, mathtools, dsfont} 
\usepackage{fixmath}
\usepackage{times}
\usepackage{graphicx}  
\usepackage{bm,bbm}
\usepackage[normalem]{ulem}
\usepackage[usenames,dvipsnames]{xcolor}
\usepackage{esint}
\usepackage{paralist}
\usepackage[normalem]{ulem}
\usepackage{mathrsfs}

\usepackage{geometry}
\definecolor{myurlcolor}{rgb}{0,0,0.7}
\definecolor{myrefcolor}{rgb}{0.8,0,0}
\usepackage[unicode=true,pdfusetitle,  bookmarks=true,bookmarksnumbered=false,bookmarksopen=false, breaklinks=false,pdfborder={0 0 0},backref=false,colorlinks=true, linkcolor=myrefcolor,citecolor=myurlcolor,urlcolor=myurlcolor]{hyperref}

\usepackage{inconsolata}

\DeclareGraphicsExtensions{.png,.pdf,.eps}
\graphicspath{ {./figures/}} 
\usepackage{afterpage}
\usepackage{epstopdf}
\usepackage{cleveref}
\usepackage[caption=false]{subfig}
\usepackage{ragged2e}
\usepackage{dblfloatfix}
\usepackage{float}

\makeatletter
\let\newfloat\newfloat@ltx
\makeatother

\usepackage{stackengine,scalerel}
\stackMath

\usepackage{mathtools}

\makeatletter

 \theoremstyle{plain}
 
 \theoremstyle{plain}
 \newtheorem{lem}{Lemma}
 \theoremstyle{plain}
 \newtheorem{thm}{Theorem}
 \theoremstyle{plain}
  
 \theoremstyle{plain}
 \newtheorem{exa}{Example}
 \theoremstyle{plain}
 \newtheorem{corr}{Corollary}
 \theoremstyle{plain}
 
 \theoremstyle{remark}
 \newtheorem*{rem*}{Remark}
 \theoremstyle{plain}
  \newtheorem{rem}{Remark}

\theoremstyle{plain}
 \newtheorem*{conj*}{Conjecture}
 \theoremstyle{plain}

\newcommand{\ot}{\otimes} 
\renewcommand{\exp}{\mathrm{exp}} 
\DeclareMathOperator{\tr}{tr}   
\newcommand{\HS}{\mathrm{HS}} 

\newcommand{\CC}{\mathbb{C}} 

\renewcommand{\dim}{{d}} 

\newcommand{\Herm}{\mathrm{Herm}} 
\newcommand{\states}{\mathrm{D}}

\newcommand{\pstates}{\mathcal{S}} 
\newcommand{\povms}{\mathrm{P}}

\newcommand{\channels}{\mathrm{CPTP}} 

\newcommand{\Hsym}[1]{\mathcal{H}_\mathrm{sym}^{(#1)}} 
\newcommand{\av}{\mathrm{av}} 
\newcommand{\s}{\mathrm{s}} 
\newcommand{\m}{\mathrm{m}} 
\newcommand{\ch}{\mathrm{ch}} 

\newcommand{\lstates}{\ell}
\newcommand{\lpovms}{\lstates}
\newcommand{\lchannels}{\lstates^{\ch}}

\newcommand{\ustates}{u}
\newcommand{\upovms}{\ustates}
\newcommand{\uchannels}{\ustates^{\ch}}

\newcommand{\cstates}{a}
\newcommand{\cpovms}{\cstates}
\newcommand{\cchannels}{\cstates^{\ch}}

\newcommand{\Cstates}{A}
\newcommand{\Cpovms}{\Cstates}
\newcommand{\Cchannels}{\Cstates^{\ch}}

\newcommand{\I}{\mathbb{I}} 
\renewcommand{\S}{\mathbb{S}} 

\newcommand{\idenC}{\mathcal{I}}

\renewcommand{\H}{\mathcal{H}} 

\renewcommand{\P}{\mathsf{P}} 
\newcommand{\M}{\mathsf{M}} 
\newcommand{\N}{\mathsf{N}} 

\newcommand{\clc}{\mathrm{T}} 

\newcommand{\rbracket}[1]{\left(#1\right)} 
\newcommand{\sbracket}[1]{\left[#1\right]} 
\newcommand{\cbracket}[1]{\left\{#1\right\}} 

\newcommand{\Tr}[1]{\tr\rbracket{#1}} 

\newcommand{\defeq}{\coloneqq}
\renewcommand{\H}{\mathcal{H}} 
\newcommand{\T}{\mathcal{T}} 
\newcommand{\J}{\mathcal{J}} 

\newcommand{\Choi}{\J} 
\newcommand{\iden}{\I}

\newcommand{\p}{\mathrm{\mathbf{p}}} 
\newcommand{\q}{\mathrm{\mathbf{q}}} 

\newcommand{\dtv}{\mathrm{TV}} 
\newcommand{\dtr}{\mathrm{d}_{\mathrm{tr}}} 
\newcommand{\dop}{\mathrm{d}_{\mathrm{op}}} 
\newcommand{\ddiam}{\mathrm{d}_\diamond} 

\newcommand{\expect}[1]{\underset{#1}{\mathbb{E}}} 

\newcommand{\dav}{\mathrm{d}_{\mathrm{av}}} 

\newcommand{\davS}{\dav^{\s}}
\newcommand{\davM}{\dav^{\m}}
\newcommand{\davC}{\dav^{\ch}}

\newcommand{\uniform}{\text{uniform}}

\newcommand{\Psym}[1]{\mathbb{P}_{\mathrm{sym}}^{(#1)}} 

\definecolor{dukeblue}{rgb}{0.0, 0.0, 0.61}
\definecolor{cadmiumgreen}{rgb}{0.0, 0.42, 0.24}

\global\long\global\long\global\long\def\bra#1{\mbox{\ensuremath{\langle#1|}}}
\global\long\global\long\global\long\def\ket#1{\mbox{\ensuremath{|#1\rangle}}}

\global\long\global\long\global\long\def\kb#1#2{\mbox{\ensuremath{\ensuremath{\ensuremath{|#1\rangle\!\langle#2|}}}}}

\makeatother
\renewcommand{\ket}[1]{\left| #1 \right>} 
\renewcommand{\bra}[1]{\left< #1 \right|} 
\newcommand{\ketbra}[2]{\left| #1 \rangle\langle #2 \right|} 

\newcommand\numberthis{\addtocounter{equation}{1}\tag{\theequation}}

\begin{document} 	
\title{Exploring Quantum Average-Case Distances:\\Proofs, properties, and examples}

\author{Filip B. Maciejewski}\thanks{fmaciejewski@usra.edu}
\affiliation{ 
Center for Theoretical Physics, Polish Academy of Sciences,\\ Al. Lotnik\'ow 32/46, 02-668
Warszawa, Poland}
\affiliation{ 
USRA Research Institute for Advanced Computer Science (RIACS), Mountain View, CA, 94043}

\author{Zbigniew Pucha{\l}a}
\affiliation{Institute of Theoretical and Applied Informatics, Polish Academy of Sciences, 44-100 Gliwice, Poland}
\affiliation{Faculty of Physics, Astronomy and Applied Computer Science, Jagiellonian University, 30-348 Krak\'{o}w, Poland}

\author{Micha\l\ Oszmaniec}\thanks{oszmaniec@cft.edu.pl}
\affiliation{ 
Center for Theoretical Physics, Polish Academy of Sciences,\\ Al. Lotnik\'ow 32/46, 02-668
Warszawa, Poland}	
\email{oszmaniec@cft.edu.pl}

\begin{abstract}
In this work, we present an in-depth study of average-case quantum distances introduced in \cite{summaryVERSION}. 
The average-case distances approximate, up to the relative error, the average Total-Variation (TV) distance between measurement outputs of two quantum processes, in which quantum objects of interest (states, measurements, or channels) are intertwined with random quantum circuits.  
Contrary to conventional distances, such as trace distance or diamond norm, they quantify \textit{average-case} statistical distinguishability via random quantum circuits.  

We prove that once a family of random circuits forms an $\delta$-approximate $4$-design, with $\delta=o(d^{-8})$, then the average-case distances can be approximated by simple explicit functions that can be expressed via simple degree two polynomials in objects of interest.
For systems of moderate dimension, they can be easily explicitly computed -- no optimization is needed as opposed to diamond norm distance between channels or operational distance between measurements.
We prove that those functions, which we call quantum average-case distances, have a plethora of desirable properties, such as subadditivity w.r.t. tensor products, joint convexity, and (restricted) data-processing inequalities.
Notably, all distances utilize the Hilbert-Schmidt (HS) norm, which provides this norm with a new operational interpretation. We also provide upper bounds on the maximal ratio between worst-case and average-case distances, and for each of them, we provide an example that saturates the bound.
Specifically, we show that for each dimension $d$ this ratio is at most $d^{\frac{1}{2}},\ d, \ d^{\frac{3}{2}}$  for states, measurements, and channels, respectively. To support the practical usefulness of our findings, we study multiple examples in which average-case quantum distances can be calculated analytically.
\end{abstract}
\maketitle

\section{Introduction}

\subsection*{Motivation}

The question of how far away are two quantum objects (states, measurements, or channels) is of both fundamental and practical importance. 
That question is often phrased in terms of the statistical distinguishability of probability distributions corresponding to two objects in question (which is a problem of classical hypothesis testing \cite{mike&ike}).
Indeed, the most common distances, such as trace distance or diamond norm distance, are based on optimal protocols for such statistical discrimination. 
However, those protocols have limitations. 
In general, they might require a lot of resources (e.g., high-depth circuits) \cite{ComplexityGrowthModels}, thus they are not necessarily practical.
Another perspective on the limitations of common distance measures comes from a study of noise on quantum devices.
Specifically, a distance between a theoretical (ideal) model of an object in question, and a model for its experimental (noisy) implementation, can be used to study the potential effects of experimental imperfections on the protocol one wishes to implement.
When that is the case, the distances based on optimal state/measurement/channel discrimination in fact inform about the worst-case performance of a protocol.
However, in practice, one does not necessarily expect the worst-case to be representative of a typical device's performance 
(an in-depth study of the effects of noise on protocols involving random circuits that exploits average-case distances is presented in an accompanying manuscript \cite{summaryVERSION}).

With this motivation in mind, we propose distance measures of distance based on \textit{average statistical distinguishability} using random quantum circuits. 
Operationally, if the average-case distance between a pair of quantum objects is significant, this implies that they can be (statistically) distinguished almost perfectly using just a few implementations of random circuits.
This provides a natural interpretation analogous to conventional distances, but we consider averages over random circuits instead of optimal scenarios.
Such quantifiers can be more suitable for studying the performance of NISQ devices' performance than the above-mentioned conventional distances quantifying worst-case performance.
In particular, one of the most promising near-term applications of quantum computing are hybrid quantum-classical variational algorithms \cite{Cerezo2021VQA}, such as Quantum Approximate Optimization Algorithm (QAOA) \cite{farhi2014qaoa,farhi2019quantum,Harrigan2021QAOA} and Variational Quantum Eigensolver (VQE) \cite{peruzzo2014vqe,Kandala2017VQE,Parrish2019VQE}.
Since NISQ devices are expected to suffer from a significant amount of noise, it is instrumental to understand how it can affect such algorithms (see, e.g., \cite{Xue2021QAOA,Marshall2020QAOA,Maciejewski2021,StilckFranca2021VQA}).
Our distance measures might prove particularly useful in this context because, as explained later, the random circuits we consider form unitary designs.
Recently it was realized that circuits appearing in variational algorithms are expected to have, on average, design-like properties. Thus we expect average-case quantum distances to be a good metric to quantify the average performance of such algorithms \cite{Barren2018}.

The manuscript is accompanied by a shorter paper that summarizes all the main results and contains a discussion of practical applications of average-case distances, as well as exhaustive numerical studies (involving random circuits with QAOA-like structure) \cite{summaryVERSION}.

\subsection*{Summary of results}

In this work, we study the average Total-Variation (TV) distance between measurement outputs (statistics) of two quantum processes, in which quantum objects of interest are intertwined with random quantum circuits.  
TV distance is well known to quantify the statistical distinguishability of two probability distributions.
In general, as TV distance is not a polynomial function of underlying probability distributions, the relevant averages are hard to calculate. 
However, we derive lower and upper bounds for average TV distance and show that both bounds differ only by dimension-independent \emph{constants}.
The derivation of upper bounds requires the calculation of 2nd moments of quantities of interest, and lower bounds are derived using 2nd and 4th moments.
Formally, this means that to get both upper and lower bounds, the random circuits must form an approximate $4$-design.
Importantly, our results are valid also for any (approximate) $k$-design with $k\geq4$. 
The particular choice of 4-designs is of purely technical origin -- as remarked above, our proof techniques require 4th-degree polynomials to get lower bounds on average TV distance, while for upper bounds already $2$-designs suffice.
The above implies that for a broad family of random quantum circuits, the average TV distance is \textit{approximated}, up to the known relative error, by a simple explicit function of the objects that we wish to compare (states, measurements, or channels).
These functions, which we call average-case quantum distances, define bona fide distance measures with multiple desired properties, such as subadditivity w.r.p. to tensor products, joint convexity, or (restricted) data processing inequalities. 
Importantly, all of the proposed distances (between states, measurements, and channels) can be expressed via simple degree two polynomials in objects in question and can be easily explicitly computed for systems of moderate dimension. 
No optimization is needed as opposed to diamond norm distance between channels \cite{watrous2009semidefinite} or operational distance between measurements \cite{Puchala2018optimal}.
Notably, all of the distances utilize the Hilbert-Schmidt (HS) norm in some way. 
This gives the HS norm an operational interpretation that it did not possess before
(especially for quantum states for which average-case distance is proportional to HS distance).

Finally, so-defined average-case quantum distances have sound operational interpretation.
Namely, if a TV distance is bounded from below by a constant $c$ (here proportional to average-case quantum distance), then there exists a strategy that uses random circuits which distinguishes between two objects with probability at least $\frac{1}{2}(1+c)$ in \textit{single-shot} scenario.
Thus from Hoeffding bound, it follows that having access to multiple copies (samples) allows one to exponentially quickly approach the success probability of discrimination equal to 1 using a simple majority vote.

\subsection*{Related works}

Let us now comment on some of the commonly used distances.
The study of similarity measures between quantum objects has a long history \cite{mike&ike, Fuchs1999crypto}, and thus there are a lot of different metrics currently used in the field. Some of the most popular distances are based on the \emph{optimal} statistical distinguishability of quantum objects -- this includes trace distance between states \cite{bengtsson_zyczkowski_2006}, the operational distance between measurements \cite{Puchala2018optimal}, as well as diamond norm distance between channels \cite{bengtsson_zyczkowski_2006}.
While in those distances the optimization is done over all possible operations, there has been an interest also in distinguishability under restricted sets of operations -- such as local POVMs for discrimination of quantum states \cite{Matthews2009restricted,Lancien2013local}.
Recently, a quantum Wasserstein distance of order 1 was proposed as a measure of distance between quantum states. It generalizes a classical Wasserstein distance based on the Hamming weight and captures the notion of similarity of quantum states based on differences between their marginals  \cite{DePalma2021Wasserstein}.

For quantum states, the other very common similarity measure is quantum fidelity, which induces distance between states known as Bures distance  \cite{Luo2004distance,Acin2001unitaries}.
When one wants to compare unitary channel (quantum gate) with a general channel (noisy implementation of a gate), the relevant notions are worst-case \cite{mike&ike} and average-case gate fidelity \cite{bowdrey2002fidelity,Nielsen2002fidelity,Schumacher1996fidelity,horodecki1999fidelity,yoshi2021}.
In both cases, the relevant optimization/averaging is over all quantum states. 
For distance measures between measurements, one of the natural choices is to treat measurement as a quantum-classical channel and compute diamond norm distance \cite{Puchala2018optimal,Puchala2021multiple}.
In the context of detector tomography sometimes fidelities between theoretical and experimental POVM's elements were considered \cite{Lundeen2008,Zhang2012recursive,Endo2021qdt}.
When the target measurement is a computational basis, it is customary to use single-qubit error probabilities as a simplified quantifier of measurement's quality \cite{Google2019}.
See \cite{Fuchs1999crypto} for an extensive overview of distinguishability measures between quantum objects.

The distance measures introduced by us rely on random quantum circuits which have many applications in the context of practical quantum computing.
A notable example is shadow tomography, where random circuits are exploited to estimate multiple properties of quantum states with relatively low sample complexity \cite{aaronson2018shadow,Huang2020predicting,hadfield2020shadow,Chen2021shadow,hadfield2021shadow}. 
Another example are generalizations of the classical randomized-benchmarking scheme \cite{Emerson2005RB,Easwar2010RB,Magesan2012RB,Gambetta2012RB}  that use random circuits to estimate averaged quality metrics of quantum gates  \cite{helsen2019RB,flammia2021ACES,jonas2021}.

In Ref.~\cite{AmbainisEmerson2007} the authors prove that two states distant in Hilbert-Schmidt norm can be distinguished by a POVM constructed from approximate 4-design.
Our proofs concerning average TV distances for quantum states and measurements were inspired by the proofs therein.
In Ref.~\cite{Radhakrishnan2009hidden} the authors derived lower bounds (also containing HS distance) for TV distance in the same scenario for Haar-random POVMs, investigating applications for hidden subgroup problems.
In Ref.~\cite{Lee2003operational} the "total operational distance" between states was introduced.
It is based on the differences in obtained statistics when one performs mutually complementary projective measurements (see Ref.~\cite{Lee2003operational} for the notion of complementarity that is used).
Importantly, the authors show that such distance is equivalent to HS distance between states of interest.

\subsection*{Structure of the paper}

Let us now outline the structure of the paper.
We start by introducing necessary theoretical concepts in Section~\ref{sec:theoretical_background}.
This includes a discussion of common distance measures based on optimal statistical distinguishability, exact and approximate unitary $k$-designs, as well as stating several auxiliary Lemmas.
From those, Lemma~\ref{lem:curiousInequality2} is one of the important technical results of the work.
In Section~\ref{sec:general_methodology} we define the average Total-Variation distance between two states, measurements, and channels.
We also outline the general methodology of the proofs presented in the main section of our work -- Section~\ref{sec:main_section}.
In that section, we prove the main results of our work.
Namely, that the average Total-Variation distances between quantum objects can be approximated by explicit functions of the objects in question -- quantum states in Theorem~\ref{th:STATESav}, quantum measurements in Theorem~\ref{th:MEASav}, and quantum channels in Theorem~\ref{th:CHANNELSav}.
Those functions are what we call average-case quantum distances.
The main section is followed by Section~\ref{sec:distances_properties} where we prove that average-case quantum distances possess a variety of desired properties, such as subadditivity, joint convexity, and restricted data-processing inequalities -- summarized in Table~\ref{tab:properties_states} for states, Table~\ref{tab:properties_measurements} for measurements, and Table~\ref{tab:properties_channels} for channels.
In this section, we also prove asymptotic separations between average-case and worst-case distances, together with examples that saturate derived bounds.
In Section~\ref{sec:further_examples} we study exemplary scenarios where average-case quantum distances can be calculated analytically.
We also show that average-case distances can be used to study the average convergence of noisy distribution to uniform (trivial) distribution.
We conclude the paper with Section~\ref{sec:open_problems} where we discuss possible future research directions.

\section{Theoretical background}\label{sec:theoretical_background}

In this section, we give theoretical background for our main results.
We start by introducing basic concepts and notation. 
Then we discuss in detail common distance measures based on optimal statistical distinguishability, and we recall notions of exact and approximate unitary $k$-designs.
Finally, we state a number of auxiliary lemmas that will prove useful in later parts of the work.

\subsection{Notation and basic concepts}

We start by recalling basic quantum-mechanical concepts used throughout the paper.
We will be interested in $\dim$-dimensional Hilbert space $\H_{d} \approx \CC^\dim$. 
We will omit subscript "d" if the dimension is not of importance.
By $\Herm(\H_\dim)$ we denote a space of Hermitian operators on $\H_\dim$.
A quantum state $\rho$ is a positive-semi-definite operator with a trace equal to $1$.
We denote set of all quantum states on $\H_{\dim}$ as $\states(\H_\dim)$, and subset of pure states as $\pstates(\H_\dim)$.
An $n$-outcome POVM \cite{Peres2002} (Positive Operator-Valued Measure, or simply a quantum measurement) $\M$ is a tuple of $n$ operators (called \textit{effects}) $\M = \rbracket{M_1, \dots, M_n}$ that fulfill $M_i\geq 0$ and $\sum_{i=1}^{n}M_i = \iden_{\dim}$, where $\iden_{\dim}$ is an identity on $\H_{\dim}$.
The set of all $n$-outcome POVMs on $\H_{\dim}$ will be denoted as $\povms(\H_\dim,n)$.
We will omit the symbol "n" if the number of outcomes is not of importance.
An important example of a measurement that will be useful later is a computational-basis measurement defined as $\M^{\text{comp}}=\rbracket{\ketbra{1}{1},\dots,\ketbra{\dim}{\dim}}$.
A quantum channel $\Lambda$ is a linear CTPT (Completely-Positive Trace-Preserving) map \cite{mike&ike}.  
Trace-preserving condition means that for any quantum state $\rho , \Lambda\left(\rho\right) \in \states(\H_{\dim})$, we have $ \Tr{\Lambda\rbracket{\rho}} = \Tr{\rho}$.
Complete-positivity means that $(\Lambda \otimes \idenC_{\dim'})\ \tilde{\rho} \geq 0$ for any $\dim'$ and any $\tilde{\rho} \in \states(\H_{\dim \dim'})$, where $\idenC_{\dim'}$ denotes identity channel on $\H_{\dim'}$.
Quantum channel $\Lambda$ is described via corresponding Choi-Jamio\l{}kowski state defined as $\Choi_{\Lambda} \coloneqq (\idenC_{\dim} \otimes \Lambda) (\ketbra{\Phi^{+}}{\Phi^{+}})$, where we extend Hilbert space by its copy and act with channel $\Lambda$ on a half of the maximally entangled state $\ket{\Phi^{+}}\coloneqq\frac{1}{\sqrt{d}}\sum_{i=1}^{d}\ket{ii}$.
We denote the set of all quantum channels from $\mathcal{H}$ to itself as $\channels(\H_{\dim})$.
A \emph{unital} quantum channels is a channel $\Phi \in \channels(\H_{\dim}) $ such that $\Phi(\tau_\dim)=\tau_\dim$, where $\tau_\dim$ is the maximally mixed state in $\H_\dim$.
When quantum state $\rho \in \states(\H_{\dim})$ undergoes process $\Lambda \in \channels(\H_{\dim})$ followed by measurement described by POVM $\M \in \povms(\H_{\dim},n)$, the probability of outcome labeled as "$i$" is given by Born's rule $p_i\rbracket{i|\rho,\Lambda,\M} = \Tr{\Lambda(\rho) M_i}$.

In the next subsection, we define distances induced by the following norms.
Denote by $L(\H_{\dim})$ a space of linear operators on $\H_{\dim}$.
Then for $A \in L(\H_{\dim})$, the trace norm is defined as
\begin{align}
  ||A||_{1} = \tr{\left(\sqrt{AA^{\dagger}}\right)}\ .
\end{align}
For a channel $\Lambda \in \channels(\H_{\dim})$, the diamond norm is defined through optimization of a trace norm as
\begin{align}
  ||\Lambda||_{\diamond} = \max_{A\in L(\H_{\dim^2}),\ ||A||_1 \leq 1} ||\left(\Lambda \otimes \idenC_{\dim} \right) A||_{1}  \ .
\end{align}

\begin{figure}[!t]
\begin{center}
\captionsetup[subfigure]{format=default,singlelinecheck=on,justification=RaggedRight}
\subfloat[\label{fig:diagram_worst_distance_states}Trace distance between quantum states.]
        {\includegraphics[width=0.31\textwidth]{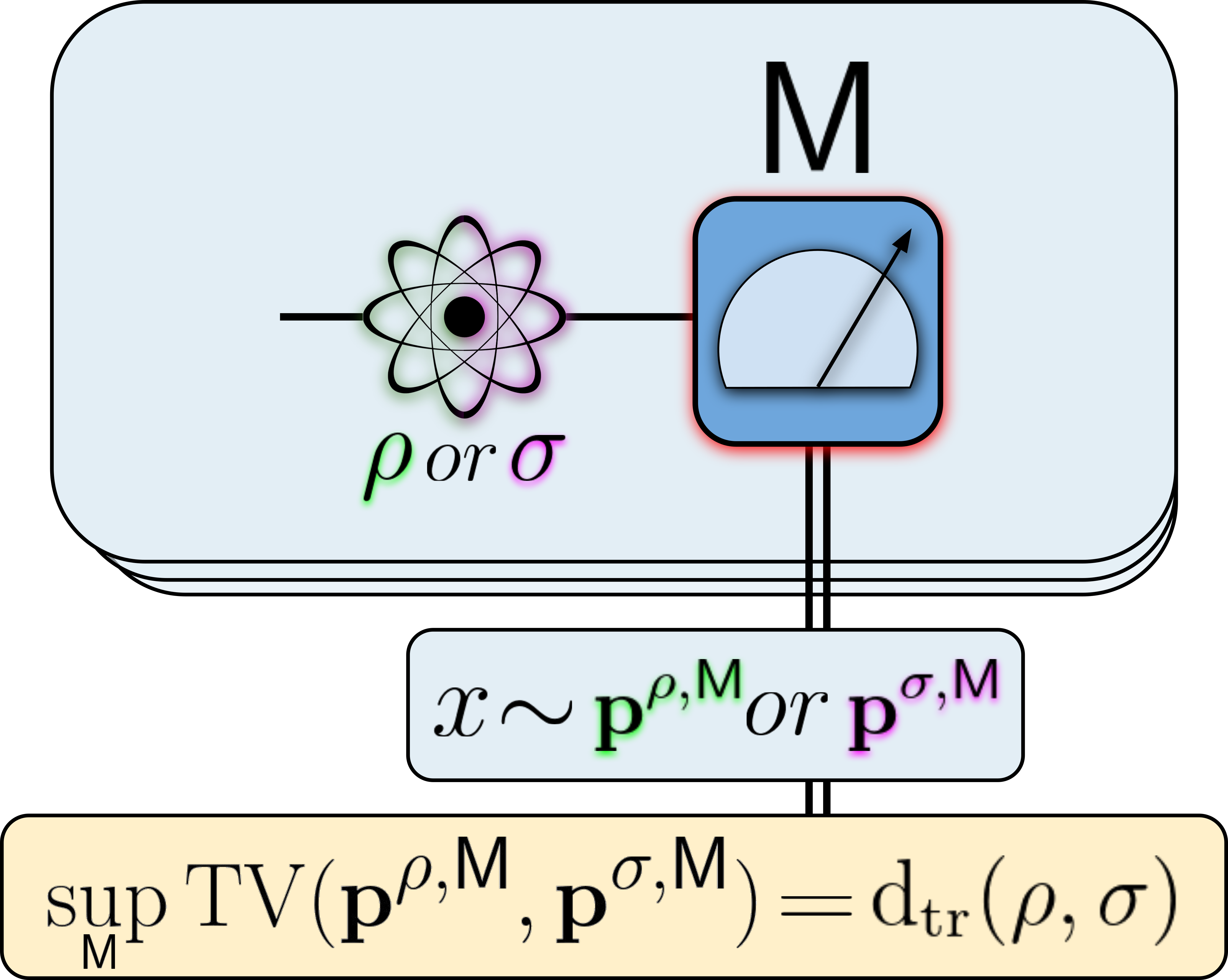}} $\quad   $
\subfloat[\label{fig:diagram_worst_distance_measurements}Operational distance between POVMs.]
        {\includegraphics[width=0.31\textwidth]{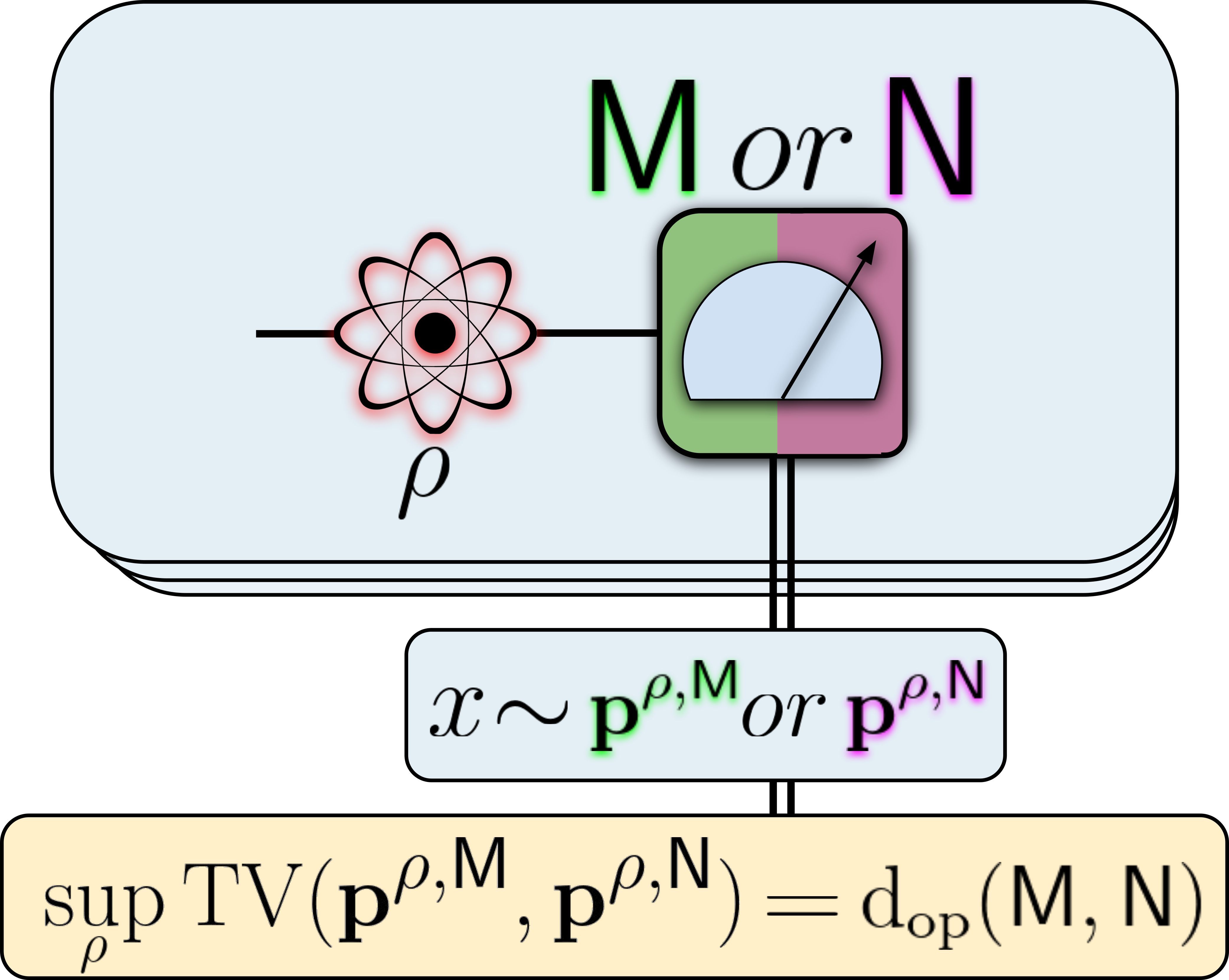}}$\quad   $
\subfloat[\label{fig:diagram_worst_distance_channels}Diamond distance between quantum channels.]
        {\includegraphics[width=0.31\textwidth]{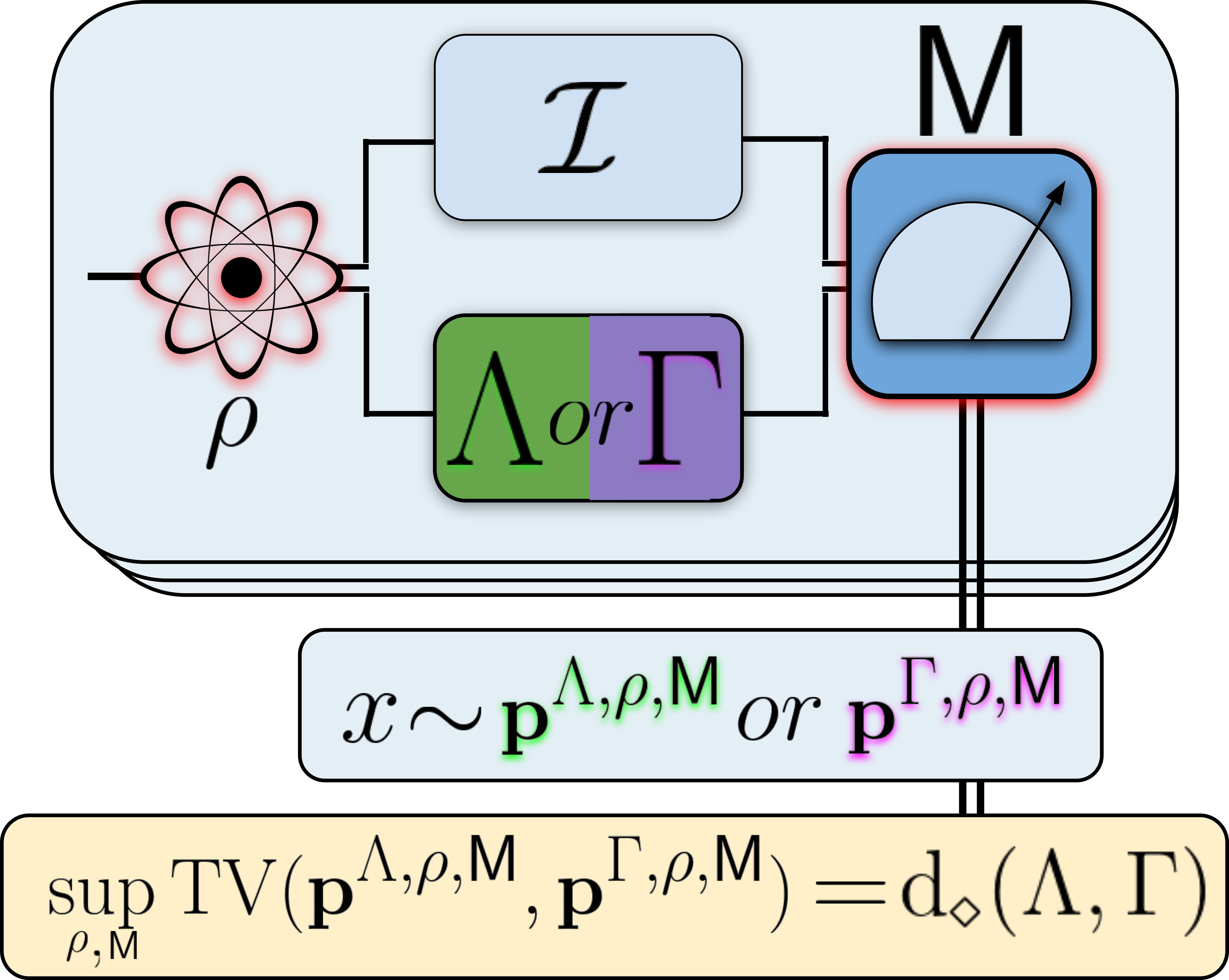}}
    \caption{\label{fig:diagram_worst_case_distances}
Depiction of measures of distance between quantum objects based on \textit{optimal} statistical distinguishability -- which can be also interpreted as "worst-case" distance. 
For quantum states (\ref{fig:diagram_worst_distance_states}), we optimize over all POVMs, while for measurements (\ref{fig:diagram_worst_distance_measurements}) we optimize over all states. 
For quantum channels (\ref{fig:diagram_worst_distance_channels}) we optimize over both states and measurements on the extended Hilbert space.
}
 \end{center}
\end{figure}

\subsection{Worst-case distance measures}\label{sec:distance_measures}
\textit{Total-Variation Distance} between two probability distributions $\p=\cbracket{p_i}_{i=1}^{n}$ and $\q = \cbracket{q_i}_{i=1}^{n}$ is defined by
\begin{align}
    \dtv\rbracket{\p,\q} = \frac{1}{2} \sum_{i=1}^{n} | p_i-q_i | \ .
\end{align}
The TV distance  quantifies the maximal statistical distinguishability of $\p$ and $\q$.
Specifically, in a task when we are asked to decide whether the provided samples come from $\mathbf{p}$ or $\mathbf{q}$ (where both are promised to be given with probability $\frac{1}{2}$), the optimal success probability (i.e., probability of correctly guessing using the best possible strategy) is $\frac{1}{2}\rbracket{1+\dtv\rbracket{\mathbf{p},\mathbf{q}}}$ \cite{mike&ike}. In quantum mechanics, the analogous task is to distinguish between two quantum objects, which can be either states, measurements, or channels (and, again, both are promised to be given with probability $\frac{1}{2}$), provided samples from the probability distributions that the objects of interest generate (via Born's rule).
In all cases, the optimal success probability of performing this task is related to the optimal (maximized) TV distance between relevant probability distributions.
This success probability is given by similar formula $\frac{1}{2}\rbracket{1+\text{d}\rbracket{\alpha_1,\alpha_2}}$, where $\alpha_1$ and $\alpha_2$ denote two objects to be distinguished, and the distance $\text{d}\rbracket{.\ ,\ .}$ depends on the scenario.
In Fig.~\ref{fig:diagram_worst_case_distances}, we pictorially present the most important distances based on optimal statistical distinguishability.

In the task where we want to distinguish between quantum states $\rho$ and $\sigma$, we optimize over measurements (POVMs) performed on them, and the relevant distance is \textit{trace distance} defined as \cite{mike&ike}
\begin{align}\label{eq:trace_distance}
    \dtr\rbracket{\rho,\sigma} =  \sup_{\M \in \povms(\H)}
    \ \dtv{\rbracket{\mathbf{p}^{\rho, \M},\mathbf{p}^{\sigma, \M} }} = \frac{1}{2}||\rho-\sigma||_{1} \ ,
\end{align}
where by $\mathbf{p}^{\rho, \M}$ we denote probability distribution obtained via Born's rule when measurement $\M$ is performed on state $\rho$. 
In this case, the optimal measurement, known as Helstrom's measurement,
is projective with 2 outcomes \cite{helstrom1976}.

In the case of quantum measurements, we want to decide whether the measurement performed is a POVM $\M$ or $\N$, and we are optimizing over input states. 
The relevant distance is so called \textit{operational distance} defined as \cite{Navascues2014energy,Puchala2018optimal,Puchala2021multiple}
\begin{align}\label{eq:dop}
    \dop\rbracket{\M,\N} =  \sup_{\rho \in \states(\H)}
    \ \dtv{\rbracket{\mathbf{p}^{\rho, \M},\mathbf{p}^{\sigma, \N}}}
    \ .
\end{align}

Finally, for distinguishing between two quantum channels $\Lambda$ and $\Gamma$, we are optimizing over both input states (with ancillae) and measurements. 
In this case, the relevant distance is known as \textit{diamond distance} defined as \cite{mike&ike}
\begin{align}\label{eq:diamond_distance}
\ddiam{\rbracket{\Lambda,\Gamma}} = \sup_{\rho \in \mathrm{D}\left(\mathcal{\H}^{\otimes 2}\right),\ \M \in \mathrm{P}\left(\mathcal{\H}^{\otimes 2}\right)} \dtv{\rbracket{\mathbf{p}^{\rho, \Lambda, \M},\mathbf{p}^{\sigma, \Gamma, \M}}}
\ ,
\end{align}
where we extended channel $\Lambda\otimes\idenC_{\dim}$ via identity channel $\idenC_{\dim}$ acting on ancillary system. 
While for the above distance, we do not have a simple expression as a function of underlying objects, its calculation can be formulated as an SDP program that can be efficiently computed for moderate system sizes \cite{watrous2009semidefinite}.

\subsection{Exact and approximate unitary $k$-designs }\label{sec:unitary_designs}

In our work, we will be interested in expected values (integrals) $\expect{\beta\sim \nu} f(\beta)= \int_\mathcal{\mathrm{U}(\H)} d\nu(\beta) f(\beta)$ of a random variable $f$ with respect to measure $\nu$ defined on unitary group $\mathrm{U}(\H)$.
The measure $\nu$ on unitary group induces measure on the set of pure quantum states in the following way -- choose arbitrary fixed state $\psi_0$ and apply to it unitary $U \sim \nu$ drawn from measure $\nu$, obtaining random state $\psi = U\psi_0U^{\dag}$. 
In short, we denote $\psi \sim \nu_{\pstates}$.
The unique left-and right-invariant probability measure on $\mathrm{U}(\H)$ is known as the Haar measure and it will be denoted as $\mu$.
Random states obtained from the induced measure on states are called Haar-random states and the corresponding measure will be denoted by $\mu_\pstates$.
Instrumental in our considerations, will be the notion of (approximate) unitary $4$-designs.
Unitary $k$-designs are, by definition, measures on $\mathrm{U}(\H)$ that reproduce averages of Haar measure $\mu$ on balanced polynomials of degree $k$ in entries of $U$ \cite{AmbainisEmerson2007}.
For approximate $k$-designs these averages agree only approximately, and the quantitative notion of approximation can be defined differently (see, e.g.\cite{Low2010}). 
Here we adapt the notion of approximation based on the diamond norm. 
We say that a measure $\nu$ on $\mathrm{U}(\H)$ is $\delta$-approximate $k$-design if
\begin{equation}
    \left\| \T_{k,\nu} - \T_{k,\mu} \right\|_\diamond \leq \delta\ ,
\end{equation}
where $\T_{k,\nu}$ is the quantum channel acting on $\H^{\ot k}$ defined as $\mathcal{T}_{k,\nu}(\rho)=\int_{\mathrm{U}(\H)} d\nu(U) U^{\ot k} \rho (U^{\dagger})^{\ot k}$. 
An important example of approximate $k$-designs are the 1D architecture random quantum circuits formed from \emph{arbitrary} universal gates that randomly couple neighboring qubits.
These easy-to-implement circuits approximate $k$-designs efficiently with the number of qubits $N$ \cite{LocalRandomCircuitsDesigns, designsNETS,ExplicitDesignsNickJonas2021}.
Specifically, $\delta$-approximate $4$-designs are generated by local random quantum circuits of depth $O\left(N\left(N+\log\frac{1}{\delta}\right)\right)$ and by the random brickwork architecture in depth $O(N+\log(1/\delta))$, with moderate numerical constants  (see Table 1 of \cite{ExplicitDesignsNickJonas2021} for the exact scaling).

\subsection{Auxiliary lemmas}

We will be interested in bounding from below and from above the expected values of some random variables. 
In bounding from above, we will use the following 
\begin{lem}\label{lem:jensens_inequality}(Jensen's inequality \cite{Jensen1906})
Let $f$ be a concave function, and $X$ be a random variable. 
Then we have
\begin{align*}\label{eq:jensens_inequality}
    f\rbracket{\expect{}X} \geq \expect{}f\rbracket{X} \ . 
    \numberthis
\end{align*}
\end{lem}
On the other hand, in bounding from below, we will use the following
\begin{lem}\label{lem:bergers_inequality}(Berger's inequality \cite{Berger1997})
Let $X$ be a random variable with well-defined second and fourth moments.
Then we have
\begin{align}\label{eq:bergers_inequality}
\frac{(\expect{} [X^2])^\frac{3}{2}}{(\expect{} [X^4])^\frac{1}{2} } \leq \expect{} |X|\ .
\end{align}
\end{lem}

We will also make use of the following auxiliary lemmas.

\begin{lem}[Auxiliary integral involving k-th moment {\cite[Prop. 6]{harrow2013church}}]\label{lem:any_moment}
Let $X\in\Herm((\H_\dim)^{\ot k})$ 
and $\mu$ be a Haar measure.
Then we have
\begin{equation}\label{eq:kMomentSimple}
      \expect{U\sim\mu} \tr\left(U^{\ot k} \kb{i}{i}^{\ot k} (U^\dag)^{\ot k} X\right) = \frac{1}{\binom{d+k-1}{k}} \tr\left(\Psym{k} X \right)\ ,
\end{equation}
where $\Psym{k}$ is the projector onto $k$-fold symmetric subspace $\Hsym{k}\subset \H_\dim^{\ot k}$.
\end{lem}

\begin{corr}[Auxiliary integral for 2nd moment]\label{corr:intAUX1}
Let $X\in\Herm(\H_\dim)$.
Then we have
\begin{equation}\label{eq:2MomentSimple}
      \expect{U\sim\mu} \tr(\kb{i}{i} U X U^\dagger)^2 = \frac{1}{d(d+1)} \left(\tr(X^2)+\tr(X)^2\right)\ .
\end{equation}
\begin{proof}
The above identity follows from Lemma~\ref{lem:any_moment}.
We use the identities  $\Psym{2}=\frac{1}{2}(\I\ot\I +\mathbb{S})$ and $\tr(\mathbb{S}\rho\otimes\rho)=\tr(\rho^2)$, where $\mathbb{S}$ denotes the swap operator acting on $\H^{\ot 2}$.
\end{proof}
\end{corr}

\begin{lem}[Lemma 2 from \cite{CliffordOrbitsDisting2016}]\label{lem:curiousInequality1}
Let $X\in\Herm(\H)$
.
Let $\Psym{k}$ denotes  orthogonal projector onto $k$-fold symmetrization of $\Hsym{k}\subset \H^{\ot k}$. We then have the following inequality
\begin{equation}\label{eq:projInequality1}
    \tr\left(X^{\ot4}\   \Psym{4}  \right)  \leq C \tr\left(X^{\ot2}\ \Psym{2}  \right)^2\ ,\ \text{where}\ C=\frac{10.1}{6}  \ .
\end{equation}
\end{lem}

Finally, the following Lemma~\ref{lem:curiousInequality2}, proved in Appendix~A of Supplementary Material (SM), generalizes Lemma~\ref{lem:curiousInequality1} and can be of independent interest. This result will be instrumental in proofs regarding average-case distances between quantum channels.
\begin{lem}[Inequality involving two operators and projections onto 2-fold symmetric subspaces]\label{lem:curiousInequality2}
Let $X,Y\in\Herm(\H)$.
Let $\Psym{k}$ denotes the orthogonal projector onto $k$-fold symmetrization of $\Hsym{k}\subset \H^{\ot k}$. We then have the following inequality
\begin{equation}\label{eq:projInequality2}
    \tr\left(X^{\ot2} \ot Y^{\ot2}\  \Psym{4}  \right)  \leq C \tr\left(X^{\ot2}\ \Psym{2}  \right)  \tr\left(Y^{\ot2}\ \Psym{2}  \right) \ ,\ \text{where}\ C=\frac{13}{6}   \ .
\end{equation}
\end{lem}
\begin{rem}
Note that the constant appearing on the right-hand side of \eqref{eq:projInequality2} is slightly worse than the one from \eqref{eq:projInequality1}.
\end{rem}


\section{Methodology}\label{sec:general_methodology}

The goal of this section is to provide an overview of the main results of this work.
We will describe the notion of average Total-Variation, and the general methodology for proofs given in Section~\ref{sec:main_section}.

\subsection{Average Total Variation distances}

\begin{figure}[!t]
\begin{center}
\captionsetup[subfigure]{format=default,singlelinecheck=on,justification=RaggedRight}
\subfloat[\label{fig:diagram_average_distance_states}Quantum states.]
        {\includegraphics[width=0.3\textwidth]{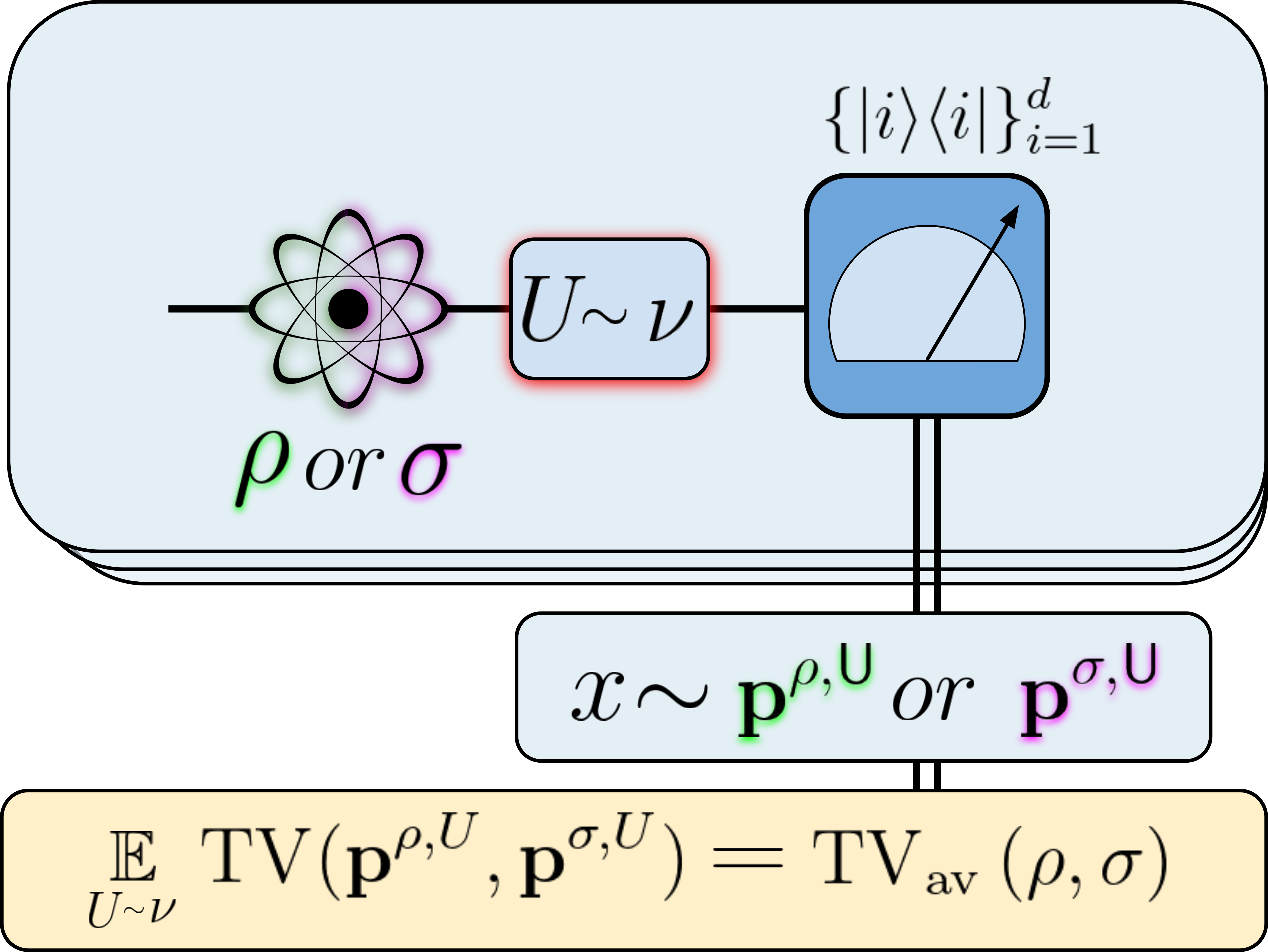}} $\qquad$
\subfloat[\label{fig:diagram_average_distance_measurements}Quantum measurements.]
        {\includegraphics[width=0.3\textwidth]{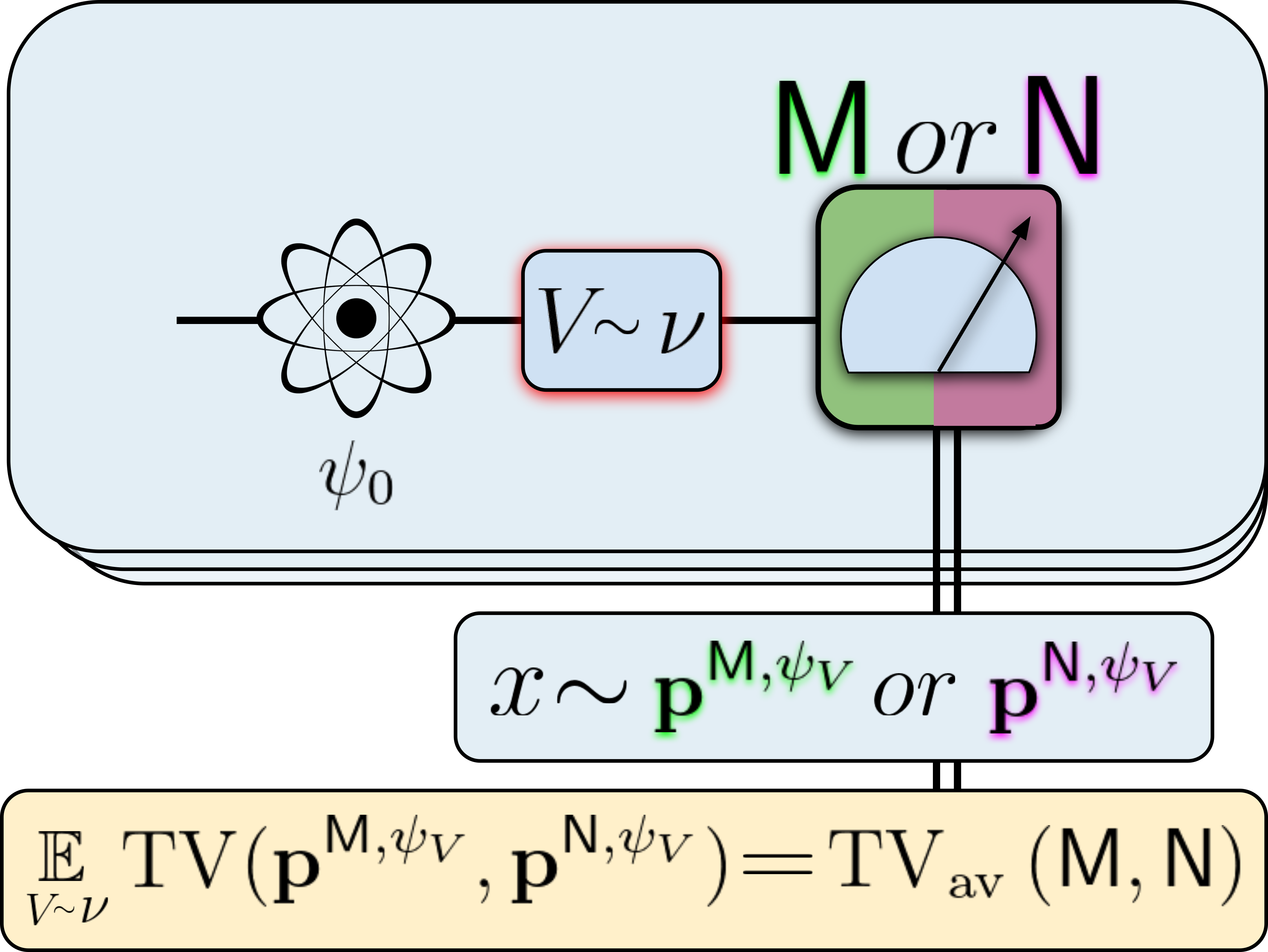}}$\qquad$
\subfloat[\label{fig:diagram_average_distance_channels}Quantum channels.]
        {\includegraphics[width=0.3\textwidth]{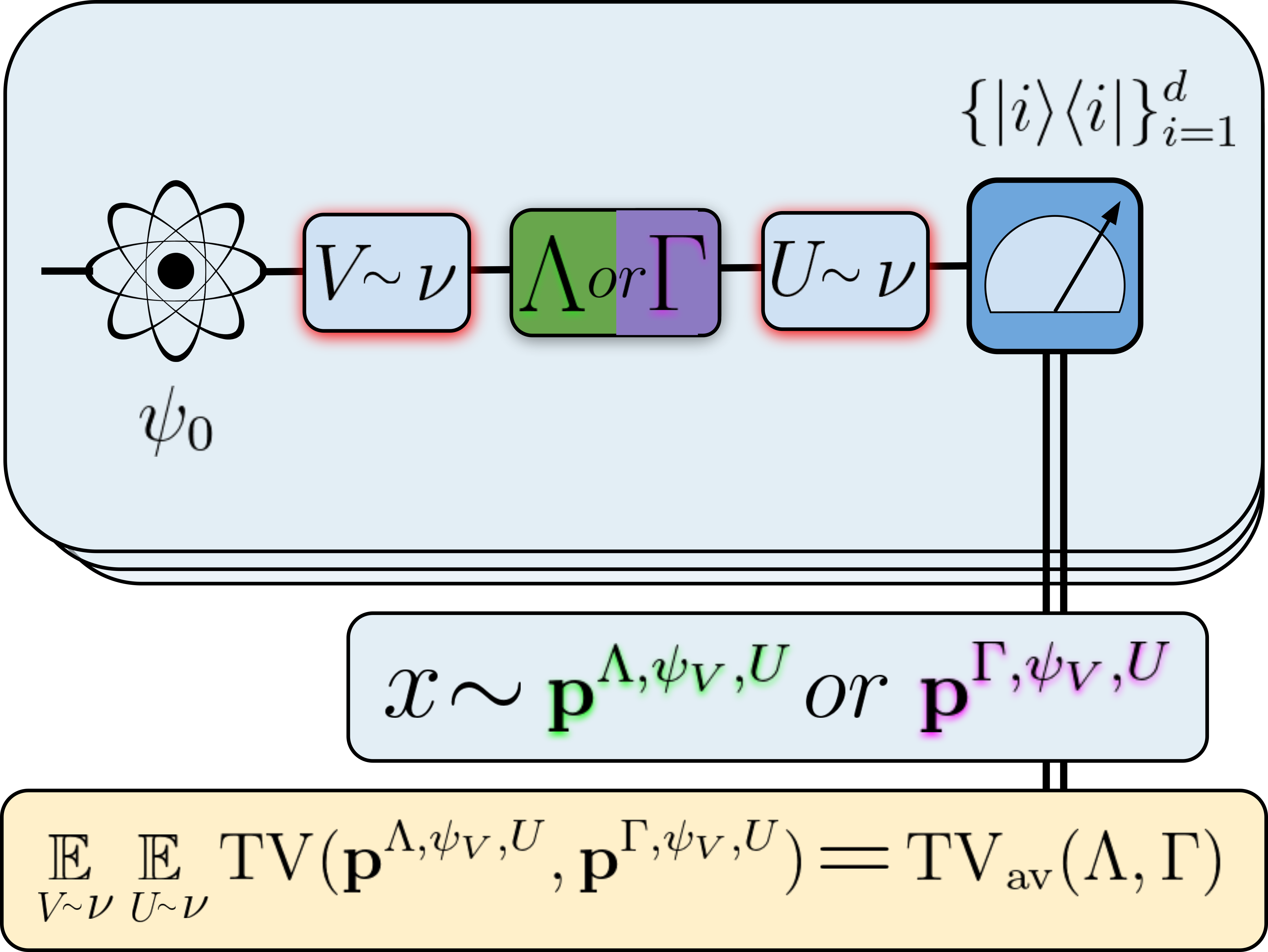}}
    \caption{\label{fig:diagram_average_case_distances}
Measures of distance between quantum objects based on \emph{average} statistical distinguishability. 
For quantum states (\ref{fig:diagram_average_distance_states}), we take the average over random unitaries applied to the state, followed by measurement in a standard basis.
For quantum measurements (\ref{fig:diagram_average_distance_states}), we take the average over random pure states measured on the detector.
Finally, for quantum channels (\ref{fig:diagram_worst_distance_channels}) we take the average over random input states, and random unitaries applied \emph{after} the action of the channel.
Note the difference with Fig.~\ref{fig:diagram_worst_case_distances}, where for common distance measures the \emph{optimal} protocol is chosen, while here we consider random protocols.
}
 \end{center}
\end{figure}

We will be interested in establishing bounds for \textit{average Total-Variation distance} between probability distributions generated  by two quantum objects (states, measurements, or general channels).
The average will be taken over an ensemble of  \textit{random} circuits.
These notions are represented pictorially in Fig.~\ref{fig:diagram_average_case_distances}, and we will now formally define them.

Consider a general quantum protocol that consists of a state preparation, an evolution of the system, and a quantum measurement.

Now we consider \textit{average Total-Variation distance} between two quantum objects:
\begin{enumerate}
    \item (\textit{States}) Two quantum states $\rho,\sigma \in \states(\H)$ are fixed, rotated by a random unitary, and measured in the computational basis.
    Let us denote by $\p^{\rho,U}$ probability distribution obtained in this process, i.e., $p_i^{\rho,U} = \tr\left(\ketbra{i}{i}U\rho U^{\dag}\right)$.
The average TV distance between $\rho$ and $\sigma$ is 
\begin{align}\label{eq:tvd_states}
   \dtv_{\av}\left(\rho,\sigma\right) \coloneqq \expect{U\sim\nu} \dtv\rbracket{\p^{\rho,U},\p^{\sigma,U}}
   \ .
\end{align}
   \item (\textit{Measurements}) Two $n$-outcome quantum measurements $\M, \N \in \povms(\H,n)$ are fixed, while states are taken to be random.
   Let us denote by $\p^{\M,\psi_{V}}$ probability distribution obtained in this process, i.e., $p_i^{\M,\psi_{V}} = \tr\left(M_i V \psi_0 V^{\dag}\right)$, where $\psi_0$ is a fixed pure state.
   The average TV distance between $\M$ and $\N$ is 
   \begin{align}\label{eq:tvd_measurements}
  \dtv_{\av}\left(\M,\N\right) \coloneqq   \expect{V\sim\nu} \dtv \left(\p^{\M,\psi_{V}}, \p^{\N,\psi_{V}}\right)
    \ .
\end{align}
       \item (\textit{Channels}) 
   Two quantum channels $\Lambda, \Gamma\in \channels(\H)$ are fixed.
   The input state is taken to be a random pure state $V\psi_0 V^{\dag}$ for fixed $\psi_0$.
   The output state is rotated by independent random unitary $U$ (hence we have random unitary rotations before and after the application of a channel), followed by measurement in a standard basis.
    Let us denote by $\p^{\Lambda,\psi_{V},U}$ probability distribution obtained in this process, i.e., $p_i^{\Lambda,\psi_{V},U} = \tr\left(\ketbra{i}{i} U\Lambda\left(V\psi_0 V^{\dag}\right)U^{\dag} \right)$
   The average TV distance between $\Lambda$ and $\Gamma$  is
   \begin{align}\label{eq:tvd_channels}
    \dtv_{\av}\left(\Lambda,\Gamma \right) \coloneqq 
    \expect{U\sim\nu} \expect{V\sim\nu} 
    \dtv\rbracket{\p^{\Lambda,\psi_{V},U},\p^{\Gamma,\psi_{V},U}}
    \ .
\end{align}
  
\end{enumerate}

\begin{rem}\label{rem:hoeffding}
If the average TV distance is bounded from below by a constant $c$, then there exists a strategy that uses random circuits which distinguishes between two objects with probability at least $\frac{1}{2}(1+c)$ in \textit{single-shot} scenario.
Thus, by the virtue of Hoeffding's inequality, having access to $s$ copies (samples) gives an error probability of the majority vote strategy dropping as $2\ \exp(-\frac{c^2}{2}s)$.
We note that while the lower bound implies \textit{existence} of such strategy, it does not tell what is the exact protocol for realizing this success probability. 
\end{rem}

\begin{rem}\label{rem:interpretation}
The value of the average $\dtv$ distance for quantum states can be reinterpreted as $\dtv$-distance of output statistics resulting from a measurement of a \emph{single} POVM with effects $M_{i,U_j}=\nu_j U^\dagger_j \kb{i}{i} U_j$, where $\nu_j$ is the probability of occurrence of circuit $U_j$ in the ensemble $\nu$ (for simplicity of presentation we assume that ensemble $\nu$ is discrete). 
This POVM can be interpreted as a convex combination \cite{Oszmaniec17} of projective measurements $\M^{U_j}$ with effects $M^{U_j}_i=  U^\dagger_j \kb{i}{i} U_j$. 
Analogous interpretation holds also for the average $\dtv$ distances for quantum measurements and channels -- they can be interpreted as $\dtv$ distances between output statistics of the corresponding randomized protocols \cite{AmbainisEmerson2007}.
Recall from Remark~\ref{rem:hoeffding} that a lower bound on average TV distance implies that such randomized protocol distinguishes between quantum states with high probability.
We note that it immediately follows that there also exists a deterministic (not randomized) \emph{optimal} distinguishability protocol that achieves the same success probability.
\end{rem}

\subsection{General methodology of proofs}

\begin{figure}[!t]
\begin{center}
\captionsetup[subfigure]{format=default,singlelinecheck=on,justification=RaggedRight}
\includegraphics[width=0.5\textwidth]{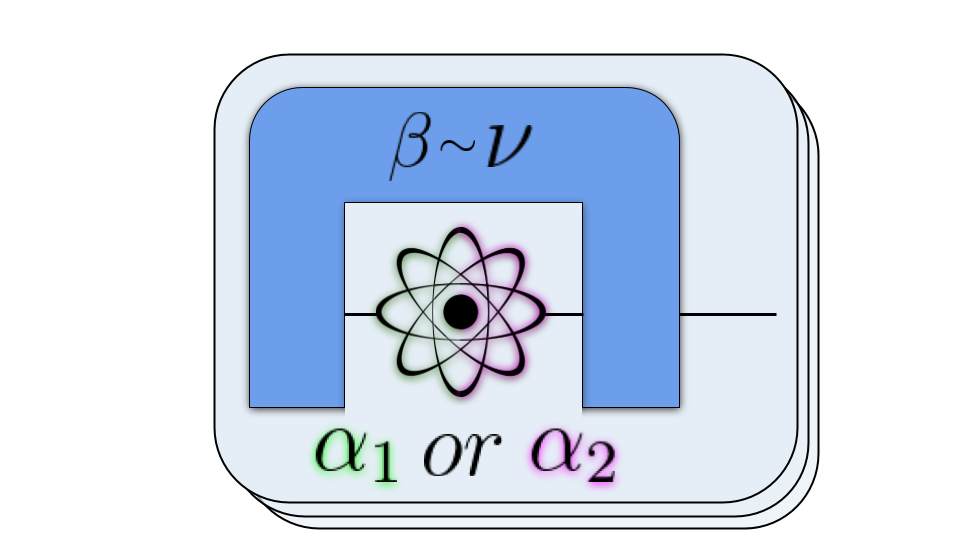}
\caption{\label{fig:diagram_general_setup}
Illustration of the general setup we consider in this work.
Two quantum objects $\alpha_1,\alpha_2$ that can be either quantum states, measurements, or channels, are surrounded by random circuits $\beta$ drawn from a probability measure $\nu$.
}
 \end{center}
\end{figure}

Consider a general quantum protocol that results in probability distribution $\mathbf{p}^{\alpha, \beta}$ where $\alpha$ denotes a fixed quantum object (state, measurement, or channel), and $\beta$ is a random variable (usually specifying quantum circuit) distributed according to a probability distribution $\nu$ (typically Haar measure, approximate $k$-design, or random instances of variational circuits).
See Fig.~\ref{fig:diagram_general_setup} for illustration.
We will be interested in bounding quantities of the type
\begin{equation}\label{eq:GenAVdefinition}
      \dtv_{\av}(\alpha_1,\alpha_2)\coloneqq \expect{\beta\sim \nu}\ \dtv(\p^{\alpha_1,\beta},\p^{\alpha_2,\beta})\ .
\end{equation}
For example, in the case of the distance between quantum states, $\alpha$ would correspond to two fixed quantum states that we want to calculate the distance between, and $\beta$ would correspond to random quantum measurements (as in Eq.~\eqref{eq:tvd_states}).

To estimate $\dtv_{\av}(\alpha_1,\alpha_2)$ from above we first expand 
\begin{equation}\label{eq:expantionAVtvd}
    \expect{\beta\sim \nu}\ \dtv(\p^{\alpha_1,\beta},\p^{\alpha_2,\beta})= \frac{1}{2}\sum_{i=1}^n\expect{\beta\sim \nu}|p^{\alpha_1,\beta}_i-p^{\alpha_2,\beta}_i |,
\end{equation}
and use Jensen's inequality (see Lemma~\ref{lem:jensens_inequality}) for the concave function $f(x)=\sqrt{x}$ to upper bound the average of each of the summands
\begin{equation}\label{eq:avUPPERbound}
   \dtv_{\av}(\alpha_1,\alpha_2) \leq \frac{1}{2} \sum_{i=1}^n \sqrt{\expect{\beta\sim \nu} (p^{\alpha_1,\beta}_i-p^{\alpha_2,\beta}_i)^2  }\ .
\end{equation}

To establish a lower bound for $\dtv_{\av}(\alpha_1,\alpha_2)$ we will apply Berger's inequality (see Lemma~\ref{lem:bergers_inequality}) to random variables $x_i=p^{\alpha_1,\beta}_i-p^{\alpha_2,\beta}_i$ and insert the obtained result to \eqref{eq:expantionAVtvd}.
Importantly, in Section~\ref{sec:main_section} it will turn out that for probabilities and measures involved in our considerations, we will have 
\begin{equation}
    \expect{\beta\sim \nu} (p^{\alpha_1,\beta}_i-p^{\alpha_2,\beta}_i)^4 \leq C \left[\expect{\beta\sim \nu} (p^{\alpha_1,\beta}_i-p^{\alpha_2,\beta}_i)^2\right]^2\ ,
\end{equation}
where $C>0$ is a constant independent of the dimension of the Hilbert space or the number of measurement outcomes. 
This fact, together with Berger's inequality (Eq.~\eqref{eq:bergers_inequality}, yields the bound
\begin{equation}\label{eq:avLOWERbound}
\frac{1}{C^{1/2}}
      \frac{1}{2}   \sum_{i=1}^n \sqrt{\expect{\beta\sim \nu} (p^{\alpha_1,\beta}_i-p^{\alpha_2,\beta}_i)^2} \leq \dtv_{\av}(\alpha_1,\alpha_2)\ .
\end{equation}
Therefore, we have
\begin{align}
            \frac{1}{C^{1/2}} \frac{1}{2}  \sum_{i=1}^n \sqrt{\expect{\beta\sim \nu} (p^{\alpha_1,\beta}_i-p^{\alpha_2,\beta}_i)^2} \leq \dtv_{\av}(\alpha_1,\alpha_2)\ \leq \frac{1}{2} \sum_{i=1}^n \sqrt{\expect{\beta\sim \nu} (p^{\alpha_1,\beta}_i-p^{\alpha_2,\beta}_i)^2  }\ ,
\end{align}
which makes it clear that to calculate both lower and upper bounds for average TV distance we will need to calculate 
$\expect{\beta\sim \nu} (p^{\alpha_1,\beta}_i-p^{\alpha_2,\beta}_i)^2$.
Importantly, since both bounds will differ only by a constant (independent of dimension), it will motivate the introduction of average-case quantum distances defined as
\begin{align}
    \dav(\alpha_1,\alpha_2)\coloneqq \frac{1}{2}  \sum_{i=1}^n \sqrt{\expect{\beta\sim \nu} (p^{\alpha_1,\beta}_i-p^{\alpha_2,\beta}_i)^2}  \ .
\end{align}
Fortunately, as will be shown in Section~\ref{sec:main_section}, such terms can be expressed via simple, explicit functions of the quantum objects that we want to calculate the distance between, provided that $\nu$ forms an approximate $4$-design.

\begin{rem}
We note that depending on the perspective one adopts, either the upper bound or lower bound on average TV distance might be of particular interest.
Namely, if one wishes to compare the ideal implementation of some protocol with its noisy version, then the upper bound might be satisfactory.
In such a scenario, the ensemble of random circuits suffices to be approximate 2-design, since only 2nd moments are needed for its calculation.
On the other hand, for statistical distinguishability, the lower bound is important (see Remark~\ref{rem:hoeffding}) and thus 4-design property is necessary.
\end{rem}

\section{Average-case quantum distances}\label{sec:main_section}

In this section, we present our main technical results.
We prove that if random circuits form approximate unitary $4$-designs, the average TV distances (see Section~\ref{sec:general_methodology}) can be approximated, up to the relative error, by simple functions that can be expressed by degree-2 polynomials in quantum objects in question.
We provide explicit expressions for those functions (which we call average-case quantum distances), as well as numerical constants for the relative errors. The proofs given in this section concern \emph{exact} unitary $4$-designs, while derivations for approximate designs are relegated to Appendix~B of SM.

\subsection{Quantum states}\label{sec:main_states}

Let $\p^{\rho,U}$ denote the probability distribution obtained when the state $\rho$ ($\sigma)$ undergoes a unitary transformation according to $U$ and is subsequently measured in the computational basis of $\H_d$. 
In other words $p_i^{\rho,U}=\tr\left(\kb{i}{i} U\rho U^\dag\right)$, where $\lbrace\ket{i}\rbrace_{i=1}^d$ is a computational basis of $\H$.

\begin{thm}[Average-case distinguishability of quantum sates]\label{th:STATESav}
Let $\rho,\sigma\in \states(\H_d)$ be states on $\H_d$ and let $U$ be a random unitary in $\H_\dim$ drawn from measure $\nu$ that forms a $\delta$-approximate $4$-design, with $\delta \coloneqq \frac{\delta'}{2d^4}$,  $\delta' \in \sbracket{0,\frac{1}{3}}$.
We then have the following inequalities 
\begin{equation}\label{eq:statesBOUNDS}
    \lstates(\delta')\ \cstates\ \davS(\rho,\sigma)\leq \expect{\mathrm{U}\sim\nu} \dtv(\p^{\rho,U},\p^{\sigma,U}) \leq \ustates(\delta')\ \Cstates\ \davS(\rho,\sigma) \ ,
\end{equation}
where we define the average-case quantum distance between states
\begin{equation}\label{eq:avDistance}
    \davS(\rho,\sigma)= \frac{1}{2}\sqrt{\tr([\rho-\sigma]^2)}=\frac{1}{2}\|\rho-\sigma\|_{\HS} \ ,
\end{equation}
and  $\cstates=0.31$,  $\Cstates=\sqrt{\frac{d}{d+1}}\leq 1$, $\lstates(\delta') = \sqrt{\frac{(1-\frac{\delta'}{d^2})^{3}}{1+\delta'}}$, $\ustates(\delta') = \left( 1 + \frac{\delta'}{d^2}\right)^{\frac{1}{2}}$.

\end{thm}
\begin{proof}
In what follows we prove a version of the theorem for exact $4$-designs (i.e., setting $\delta=0$).
We start by proving the upper bound in \eqref{eq:statesBOUNDS}.
To this aim, we utilize the upper bound in \eqref{eq:avUPPERbound} (from Jensen's inequality) to obtain
\begin{equation}\label{eq:statesINTERbound}
     \dtv_{\av}(\rho,\sigma)\leq \frac{1}{2}\sum_{i=1}^d \sqrt{\expect{U\sim\nu} \tr(\kb{i}{i} U\Delta U^\dagger)^2}\ ,
\end{equation}
where we set $\Delta=\rho-\sigma$.
Using the assumed $2$-design property of $\nu$ and the standard techniques of Haar measure integration (cf. Corrolary~\ref{corr:intAUX1}) we get 
\begin{equation}\label{eq:states_2Moment}
    \expect{U\sim\nu} \tr(\kb{i}{i} U\Delta U^\dagger)^2 = \frac{1}{d(d+1)} \tr(\Delta^2)\ ,
\end{equation}
which follows directly from Eq.~\eqref{eq:2MomentSimple} and the fact that $\Delta$ is traceless.
Inserting the above into \eqref{eq:statesINTERbound} we obtain the upper bound from \eqref{eq:statesBOUNDS}. 

In order to prove the lower bound we use Berger's inequality (cf. Eq.~\eqref{eq:bergers_inequality}) for variable $X= \tr(U \kb{i}{i}U^\dag \Delta)$:
\begin{equation}\label{eq:BergerStates}
 \expect{U\sim\nu} |\tr(U \kb{i}{i}U^\dag \Delta)| \geq \frac{\left(\expect{U\sim\nu} [\tr(U \kb{i}{i}U^\dag \Delta)]^2\right)^{3/2}}{\left(\expect{U\sim\nu} [\tr(U \kb{i}{i}U^\dag \Delta)]^4\right)^{1/2}}\ .
\end{equation}
The numerator of the above fraction contains the power of the second moment already calculated in Eq.~\eqref{eq:states_2Moment}, hence we get that it is equal to $K=[\frac{1}{d(d+1)} \tr(\Delta^2)]^{3/2}$.
To get the upper bound for the denominator, we first note that from Lemma~\ref{lem:any_moment} it follows directly that the denominator is equal to  $L=[\binom{d+3}{4}^{-1}\tr(\Psym{4} \Delta^{\ot 4})]^{1/2}$, where $\Psym{4}$ is a projector onto the 4-fold symmetric subspace of $\H_\dim^{\otimes{4}}$.
Now we get
\begin{equation}
       \tr(\Psym{4} \Delta^{\ot 4}) \leq
     C\ \rbracket{\tr\left(\Psym{2}\Delta^{\ot 2}\right)}^{2}
    =
    \frac{C}{4} \rbracket{\tr(\Delta^2)}^2\ ,\
\end{equation}
with $C=\frac{10.1}{6}$.
The inequality above is a direct application of Lemma~\ref{lem:curiousInequality1}, while the equality follows from the fact that $\Delta$ is traceless and explicit form of $\Psym{2}$. Inserting everything into Eq.~\eqref{eq:BergerStates} we obtain
\begin{equation}
    \expect{U\sim\nu} |\tr(U \kb{i}{i}U^\dag \Delta)| \geq \frac{K}{L} \geq \frac{w}{d} \sqrt{\tr{\Delta^2}}\ ,
\end{equation}
for $w= \sqrt{\frac{4}{C}\ \frac{\binom{d+3}{4}}{d^2\rbracket{d+1}^2}} \geq 0.31=\cstates$. 
Finally, summing over $i=1,\ldots,d$, we obtain lower bound on average TV distance 
\begin{equation}
      \expect{U\sim\nu} \frac{1}{2}\sum_{i=1}^{d} |\tr(U \kb{i}{i}U^\dag \Delta)| \geq \cstates \  \frac{1}{2}||\Delta||_{\text{HS}}\ ,
\end{equation}
which concludes the proof.
\end{proof}

\begin{rem}
The proof of Theorem \ref{th:STATESav} is inspired by the proof of  Theorem 4 from \cite{AmbainisEmerson2007} where Berger inequality was used to prove that two states far apart in Hilbert-Schmidt norm can be information-theoretically distinguished by a POVM constructed from approximate $4$-design. 
\end{rem}
\begin{rem}
We note that in existing literature, the trace distance was usually preferred to the Hilbert-Schmidt distance, one of the reasons being the lack of an operational interpretation for the latter.
The above considerations provide such an interpretation for H-S distance in terms of average statistical distinguishability between quantum states, thus providing a sound physical motivation for its use.
\end{rem}

\begin{rem}
We note that random quantum circuits in the 1D architecture formed from \emph{arbitrary} universal gates that randomly couple neighboring qubits, generate approximate $k$-designs efficiently with the number of qubits $N$ \cite{LocalRandomCircuitsDesigns, designsNETS,ExplicitDesignsNickJonas2021}.
Specifically, $\delta$- approximate $4$-designs are generated by the 1D random brickwork architecture in depth $O(N+\log(1/\delta))$, with moderate numerical constants \cite{ExplicitDesignsNickJonas2021}.
This implies that ensembles appearing in Theorem~\ref{th:STATESav} can be easily realized.
Furthermore, it is expected that some of the classes of variational quantum circuits are expected to have, on average, unitary design-like properties \cite{Barren2018}.
This suggests that average-case quantum distances might be used to quantify the average-case performance of hybrid quantum-classical algorithms.
Naturally, the same remarks hold for Theorem~\ref{th:MEASav} for measurements and Theorem~\ref{th:CHANNELSav} for channels.
\end{rem}

\subsection{Quantum measurements}\label{sec:main_measurements}

 Let $\p^{\M,\psi_V}$  denote the probability distribution of a quantum process in which a fixed pure quantum state $\psi_0$ on $\H_\dim$ is evolved according to unitary $V$ and is subsequently measured via a $n$-outcome POVM $\M=(M_1,M_2,\ldots, M_n)$. 
 In other words $p_i^{\M,\psi_V}=\tr(V\psi_0 V^\dagger M_i)$.   

\begin{thm}[Average-case quantum distance between quantum measurements]\label{th:MEASav} 
Let $\M,\N$ be $n$-outcome POVMs on $\H_{\dim}$ and $V$ be a random unitary on $\H_\dim$ drawn from measure $\nu$ that forms a $\delta$-approximate $4$-design, with $\delta \coloneqq \frac{\delta'}{2d^4}$, $\delta' \in \sbracket{0,\frac{1}{3}}$.
We then have the following inequalities
\begin{equation}\label{eq:povmBOUNDS}
  \lpovms(\delta')\   \cpovms\ \davM(\M,\N)\leq \expect{V\sim\nu} \dtv(\p^{\M,\psi_{V}},\p^{\N,\psi_{V}}) \leq  \upovms(\delta')\ \Cpovms\ \davM(\M,\N)\ ,
\end{equation}
where we define an average-case quantum distance between measurements
\begin{equation}\label{eq:avDISTpovm}
    \davM(\M,\N) = \frac{1}{2d}\sum_{i=1}^n \sqrt{ \|M_i-N_i\|_{\HS}^2 + \tr(M_i-N_i)^2}\ .
\end{equation}
and
 $\cpovms=0.31$,  $\Cpovms=\sqrt{\frac{d}{d+1}}\leq 1$, $\lpovms(\delta') =\sqrt{\frac{(1-\frac{\delta'}{d^2})^{3}}{1+\delta'}}$, $\upovms(\delta') = \left( 1 + \frac{\delta'}{d^2}\right)^{\frac{1}{2}}$.
 
\end{thm}
\begin{proof}
 In what follows we prove a version of the theorem for exact $4$-designs.
The proof is in fact almost exactly the same as of Theorem~\ref{th:STATESav}.
We can define $\Delta_i = M_i-N_i$, now each $\Delta_i$ having the role of previous $\Delta$, namely in each summand appearing in the TV distance is of the form $|\tr(V\psi_0V^{\dag}\Delta_i)|$.
We now note that arbitrary fixed pure state $\psi_0=U_0\ketbra{0}{0}U_0^{\dag}$ is unitarily equivalent to computational basis state via some unitary $U_0$, and that Haar measure is invariant under transformation $U\rightarrow U U_0$.
From those facts, it follows that we can apply exactly the same steps as for proof of Theorem~\ref{th:STATESav}.
For the second moment we obtain
\begin{equation}
    \expect{\psi\sim\nu_\pstates} \tr(\psi \Delta_i )^2 = \frac{1}{d(d+1)}\left(\tr(\Delta_i^2) +\tr(\Delta_i)^2 \right)\ ,
\end{equation}
which differs from Eq.~\eqref{eq:states_2Moment} by additional summand, because now $\Delta_i$ is not necessarily traceless.
The rest of the steps is exactly analogous to the proof of Theorem~\ref{th:STATESav}.

\end{proof}

\subsection{Quantum channels}\label{sec:main_channels}

Let $\p^{\Lambda,\psi_V,U}$ be the probability distribution associated with a quantum process in which a fixed pure state $\psi_0\in\pstates(\H)$ is transformed by unitary transformation $V$,  channel $\Lambda$,   unitary $U$, and is subsequently measured in the computational basis of $\H_\dim$. In other words we have $p_i^{\Lambda,\psi_V,U}=\tr(\kb{i}{i} U \Lambda(V \psi_0 V^\dag) U^\dagger)$. 

\begin{thm}[Average-case distinguishability of quantum channels]\label{th:CHANNELSav} Let $\Lambda,\Gamma$ be quantum channels acting on $\H_d$. let $\nu$ be a distribution on on $\mathrm{U}(\H_d)$ forming $\delta$-approximate $4$-design for $\delta=\frac{\delta'}{(2d)^8}$. 
Then we have the following inequalities 
\begin{align}\label{eq:channelsBOUND}
    \lchannels(\delta')\ \cchannels\ \davC(\Lambda,\Gamma) \leq \expect{V\sim\nu} \ \expect{U\sim\nu} \dtv(\p^{\Lambda,\psi_V,U},\p^{\Gamma,\psi_V,U})\    \leq   \uchannels(\delta')\  \Cchannels \  \davC(\Lambda,\Gamma)
\     ,
\end{align}
where we defined the average-case quantum distance between channels
\begin{equation}\label{eq:avDistanceChannels}
    \davC(\Lambda,\Gamma)\coloneqq \frac{1}{2} \sqrt{\left\|\J_{\Lambda}-\J_{\Gamma}\right\|_{\HS}^2+\tr\left((\Lambda-\Gamma) [\tau_\dim]^2 \right) } ,
\end{equation}
and 
$\cchannels=0.087$, $\Cchannels=\frac{d}{d+1} \leq 1$, 
$ \lchannels(\delta') = \frac{\left(1-\frac{\delta'}{d^2}\right)^{3}}{ 1 +\delta'}$, $\uchannels(\delta')=1+\frac{\delta'}{d^2} $. 

\end{thm}
\begin{proof}
 In what follows we prove a version of a theorem for exact $4$-designs (i.e., setting $\delta=0$).
In order to simplify the notation we will use the notation $\Delta\defeq \Lambda - \Gamma$ (note that $\Delta$ is a superoperator and has a different meaning than $\Delta$ used in the proof of Theorem \ref{th:STATESav}).
We will make use of the Theorem \ref{th:STATESav} which implies that for fixed $\psi_V\in\pstates(\H)$ the following inequalities hold
\begin{align}
  \frac{\cstates}{2}   \| \Delta[\psi_V]  \|_\HS & \leq  \expect{U\sim\nu}\dtv(\p^{\Lambda,\psi_V,U},\p^{\Gamma,\psi_V,U})\ \label{eq:chanLOW1} ,
  \\
 \frac{\Cstates}{2}  \|\Delta [\psi_V]  \|_\HS\ & \geq 
\expect{U\sim\nu}\dtv(\p^{\Lambda,\psi_V,U},\p^{\Gamma,\psi_V,U})\  \label{eq:chanUP1}.
\end{align}
In what follows we prove bounds on $\expect{V\sim\nu} \|\Delta[\psi_V] \|_\HS=\expect{\psi\sim\nu_\pstates} \left\|\Delta[\psi]  \right\|_\HS$. 
We first establish the upper bound by  employing  Jensen's inequality 
\begin{equation}\label{eq:chanUPaver}
    \expect{\psi\sim\nu_\pstates} \left\|\Delta[\psi]  \right\|_\HS \leq  \sqrt{\expect{\psi\sim\nu_\pstates} \tr(\Delta[\psi]^2 ) }\ . 
\end{equation}
The  average of $\tr(\Delta[\psi]^2 )$ can be computed explicitly using the $2$-design property and Lemma~\ref{lem:any_moment}.
We first rewrite using the same trick as in the proof of Corollary~\ref{corr:intAUX1}
\begin{align}
    \tr\left(\Delta[\psi]^2\right) = \tr\left(\Delta[\psi]^{\ot 2} 
 \mathbb{S}\right) 
\end{align}
where $\S=\sum_{i,j=1}^d \kb{i}{j} \ot \kb{j}{i} $ is the swap operator acting on $\H_d^{\ot 2}$.
Inserting the above into Lemma~\ref{lem:any_moment} yields
\begin{align}
    \frac{2}{d(d+1)} \tr\left( \mathbb{S}\ \Delta^{\ot 2}[\Psym{2}]\ 
 \right) = \frac{2}{d(d+1)}  \frac{1}{2} \left(\tr\left(\mathbb{S} \ \Delta^{\ot 2}[\iden]\right)+\tr\left(\mathbb{S}\ \Delta^{\ot 2}[\mathbb{S}]
 \right)\right) \ ,
\end{align}
where we used the identity $\Psym{2}=\frac{1}{2}(\I\ot\I +\mathbb{S})$.
We now rewrite the first term as $\tr\left( \mathbb{S}\Delta^{\ot 2}[\iden]
  \right)  = \tr\left(\Delta[\iden]^2\right) \ $.
Inserting the above into the integral (with multiplication and division by $d^2$) gives
\begin{equation}\label{eq:chanAVexplicit}
    \expect{\psi\sim\nu_\pstates} \tr(\Delta[\psi]^2 ) 
    =\frac{d^2}{d(d+1)}\left(\tr\left(\Delta\left[\frac{\I}{d}\right]^2\right)   + \tr\left(\S\Delta^{\ot 2}\left[\frac{\S}{d^2}\right] \right)\right)\ .
\end{equation}
Recall that $\J_{\Delta}=(\iden\ot\Delta)(\Phi_+)$, where $\ket{\Phi_+}=\frac{1}{\sqrt{d}
}\sum_{i=1}^d \ket{i}\ket{i} $ is the maximally entangled state in $\H_\dim \ot \H_\dim$. 
Explicit computation gives $    \|\J_{\Delta}\|^2_\HS = \tr\left(\S \Delta^{\ot 2}\left[\frac{\S}{d^2}\right] \right)
$.
\begin{equation}\label{eq:chanAVexplicitFIN}
    \expect{\psi\sim\nu_\pstates} \tr(\Delta[\psi]^2 ) 
    =\frac{d^2}{d(d+1)}\left(\tr\left(\Delta\left[\frac{\I}{d}\right]^2\right) + \|\J_{\Delta}\|^2_\HS \right)\ .
\end{equation}
Inserting this expression into \eqref{eq:chanUPaver} and using \eqref{eq:chanUP1} finally gives the upper bound in Eq. \eqref{eq:channelsBOUND}. 

To prove the lower bound integrate both sides of  \eqref{eq:chanLOW1} and apply Berger's inequality for $X_\psi=\| \Delta[\psi]  \|_\HS= \sqrt{\tr(\Delta[\psi]^2)} $
\begin{equation}\label{eq:bergerChan}
     \expect{\psi\sim\mu_\pstates}\| \Delta[\psi]  \|_\HS \geq \frac{\left(\expect{\psi\sim\mu_\pstates} \tr(\Delta[\psi]^2) \right)^{3/2}}{\left(\expect{\psi\sim\mu_\pstates} \tr(\Delta[\psi]^2)^2 \right)^{1/2}}\ .
\end{equation}
We proceed by rewriting the integral in the denominator of the above expression
\begin{equation}
    \expect{\psi\sim\nu_\pstates} \tr(\Delta[\psi]^2)^2  = 4 \expect{\psi\sim\nu_\pstates} \tr(\Delta[\psi]^{\ot 4} \Psym{2}\ot\Psym{2})\ .
\end{equation}
where we used the identity $2\tr{\Psym{2} X \ot Y} = \tr{X}\tr(Y)+\tr(XY)$ and $\tr(\Delta[\psi])=0$.  
By expanding $\Delta[\psi]^{\ot 4}= \Delta^{\ot 4} \left[\psi^{\ot4 } \right]$ and integrating over $\psi$  (cf. Lemma~\ref{lem:any_moment}) we obtain
\begin{equation}\label{eq:chanAUX1}
    \expect{\psi\sim\nu_\pstates} \tr(\Delta[\psi]^2)^2 = \frac{4}{\binom{d+3}{4}} \tr\left(\Delta^{\ot 4}\left[\Psym{4}\right] \Psym{2}\ot\Psym{2}\right)\ .
\end{equation}
Substituting $\Psym{2}=\binom{d+1}{2}\expect{\psi\sim\nu_\pstates} \psi^{\ot 2}$ we get
\begin{equation}\label{eq:chanAUX2}
    \expect{\psi\sim\nu_\pstates} \tr(\Delta[\psi]^2)^2 = \frac{4\binom{d+1}{2}^2}{\binom{d+3}{4}} \expect{\psi \sim \nu_\pstates} \expect{\varphi \sim \nu_\pstates}  \tr\left(\Psym{4} \Delta^{\dagger} [\psi]^{\ot 2} \ot \Delta^{\dagger} [\varphi]^{\ot 2} \right) \ .
\end{equation}
We now utilize Lemma \ref{lem:curiousInequality2} to upper bound the function inside the integral
\begin{equation}
    \tr\left(\Psym{4} \Delta^{\dagger} [\psi]^{\ot 2} \ot \Delta^{\dagger} [\varphi]^{\ot 2} \right) \leq C \tr(\Psym{2} \Delta^{\dagger} [\psi]^{\ot 2}) \tr(\Psym{2} \Delta^{\dagger} [\varphi]^{\ot 2})\ ,
\end{equation}
where $C=\frac{13}{6}$. Inserting this into \eqref{eq:chanAUX2} and carrying over the integrals over $\psi$ and $\phi$ (with the help of Corollary \ref{corr:intAUX1}) we get
\begin{equation}\label{eq:chanAUX3}
    \expect{\psi\sim\nu_\pstates} \tr(\Delta[\psi]^2)^2 \leq 
    \frac{4 C}{\binom{d+3}{4}} \tr\left(\Psym{2} \Delta^{\ot 2} \left[\Psym{2}\right]\right)^2\ .
\end{equation}
We now calculate
\begin{align*}\label{eq:chanAUX4}
    \tr\left(\Psym{2} \Delta^{\ot 2} \left[\Psym{2}\right]\right) 
    & = 
    \binom{d+1}{2}\expect{\psi\sim\nu_\pstates} \tr\left(\Psym{2} \Delta[\psi]^{\ot 2}  \right)=  \\
    & = \frac{1}{2}
    \binom{d+1}{2}\expect{\psi\sim\nu_\pstates} \tr\left(\Delta[\psi]^{\ot 2} \mathbb{S}  \right)=
    \\
    & = \frac{1}{2}\binom{d+1}{2}\expect{\psi\sim\nu_\pstates} \tr\left( \Delta[\psi]^2 \right)\ . \numberthis
\end{align*}
In the first equality, we used the fact that since $\nu_{\pstates}$ forms a $2$-design we can substitute the projector onto $2$-fold symmetric subspace with corresponding (renormalized) average.
In the second equality, we exploited the fact that states of the type $\psi^{\otimes{2}}$ are in an invariant subspace of $\Psym{2}$.
Third equality follows directly from Corollary~\ref{corr:intAUX1} and fact that $\tr(\Delta[\psi])=0$.
Combining \eqref{eq:chanAUX3} and \eqref{eq:chanAUX4} we obtain
\begin{equation}\label{eq:4th_moment_channels_bound}
    \expect{\psi\sim\nu_\pstates} \tr(\Delta[\psi]^2)^2 \leq \frac{C\binom{d+1}{2}^2}{\binom{d+3}{4}} \left(\expect{\psi\sim\nu_\pstates} \tr\left( \Delta[\psi]^2 \right) \right)^2\ . 
\end{equation}
Inserting the above bound into \eqref{eq:bergerChan} gives
\begin{equation}
     \expect{\psi\sim\nu_\pstates}\| \Delta[\psi]  \|_\HS \geq b_d \sqrt{\expect{\psi\sim\nu_\pstates} \tr(\Delta[\psi]^2) }\ , 
\end{equation}
with $b_d=\frac{1}{\sqrt{13}}\sqrt{\frac{(d+2)(d+3)}{d(d+1)}}$. 
Integrating both sides of \eqref{eq:chanLOW1} over $\psi\sim\nu_{\pstates}$ and using the the above inequality we finally obtain
\begin{equation}
    \cchannels\ \davC(\Lambda,\Gamma)   \leq  \expect{V\sim\nu} \ \expect{U\sim\nu} \dtv(\p^{\Lambda,\psi_V,U},\p^{\Gamma,\psi_V,U})\ ,
\end{equation}
with $\cchannels = a\cdot b_\dim \approx 0.087$.
\end{proof}

\begin{rem}
We note that one can view quantum states and measurements as special types of quantum channels. While for state preparation channels the operational procedure of discrimination is equivalent and one gets the same expression (see Example~\ref{exa:state_prep_channels}), for measurements it is not the case (this is because for measurement channels one can average only over input states). 
Moreover, in the case of $\delta$-approximate designs, applying the above Theorem~\ref{th:CHANNELSav} for the average-case distance between state preparation channels gives worse than Theorem~\ref{th:STATESav} constants and functional dependence on $\delta$. 
This approach thus leads to less tight bounds for states than treating them separately. 
\end{rem}

\section{Properties of distance measures}\label{sec:distances_properties}

While expressions for average-case distances introduced in Section~\ref{sec:main_section} might seem abstract (especially in the case of measurements and channels), it turns out that they share multiple desired properties with common distances used in quantum information \cite{gilchrist2005distances,mike&ike}.
In particular, our distances indeed fulfill metric axioms, they are subadditive with respect to tensor products, and have a joint convexity property.
They are also non-increasing under \emph{unital} quantum channels.
Finally, the average-case quantum distance between \emph{unital} channels possesses two additional physically well-motivated properties -- stability (it does not change when both channels are extended by identity channel) and chaining (distance between compositions of multiple channels is at most the sum of distances between constituting channels) \cite{gilchrist2005distances}.
In this section, we state and prove those properties for states (Section~\ref{sec:properties_states}), measurements (Section~\ref{sec:properties_measurements}), and channels (Section~\ref{sec:properties_channels}).
To make navigation easier, each subsection starts with a table of properties, a comparison with relevant worst-case distance, and the text reference in which the properties are proved -- Table~\ref{tab:properties_states} for states, 
Table~\ref{tab:properties_measurements} for measurements,
and Table~\ref{tab:properties_channels} for channels.

\subsection{Quantum states}\label{sec:properties_states}

The following Table~\ref{tab:properties_states} summarizes properties of the average-case quantum distance between states and compares it to the worst-case trace distance. 
\begin{table}[!h]
\hspace*{-1.1cm}
\begin{tabular}{|c|c|c|c|}
\hline

\textbf{Property}          & \textbf{Worst-case distance $\dtr(\rho,\sigma)$}                                                                                                                       & \: \textbf{Average-case distance $\davS(\rho,\sigma)$}\:                                                                   & \textbf{\begin{tabular}[c]{@{}c@{}}Text reference for \\ \: average-case distance \: \end{tabular}} \\ 
& & & \\[-1em] 
\hline
& & & \\[-1em]
Function                   & $\frac{1}{2}||\rho-\sigma||_{1}$                                                                                                                                       & $ \frac{1}{2}\|\rho-\sigma\|_{\HS}$                                                                                   & Theorem~\ref{th:STATESav}, Lemma~\ref{lem:metric_states}                                     \\ 
& & & \\[-1em] 
\hline
& & & \\[-1em]
Subadditivity              & \: $\dtr(\rho_1\ot\rho_2,\sigma_1\ot\sigma_2)\leq  \dtr(\rho_1,\sigma_1)+ \dtr(\rho_2,\sigma_2)$ \:                                                                         & same as for $\dtr(\rho,\sigma)$                                                                                       & Lemma~\ref{lem:subbadditivity_states}                                                        \\ 
& & & \\[-1em] 
\hline
& & & \\[-1em]
Joint convexity            & $\dtr \rbracket{\sum_{\alpha}p_{\alpha}\rho_{\alpha}, \sum_{\alpha}p_{\alpha}\sigma_{\alpha}} \leq \sum_{\alpha}p_{\alpha}\dtr\rbracket{\rho_{\alpha},\sigma_{\alpha}}$ & same as for $\dtr(\rho,\sigma)$                                                                                       & Lemma~\ref{lem:convexity_states}                                                             \\ 
& & & \\[-1em] 
\hline
& & & \\[-1em]
\: Data processing inequality \: & \begin{tabular}[c]{@{}c@{}}$\dtr(\Lambda(\rho),\Lambda(\sigma)) \leq \dtr(\rho,\sigma)$ \\  for CPTP $\Lambda$\end{tabular}                                            & \begin{tabular}[c]{@{}c@{}}$\davS(\Phi(\rho),\Phi(\sigma)) \leq \davS(\rho,\sigma) $ \\ for unital $\Phi$\end{tabular} & Lemma~\ref{lem:dpi_states}                                                                   \\ \hline
\end{tabular}
\caption{\label{tab:properties_states}}
\end{table}

\begin{lem}[$\davS$ fulfils axioms of a metric]\label{lem:metric_states}
Let $\davS$, denote average distances between states (Eq.~\eqref{eq:avDistance}).
Then $\davS$ satisfies axioms of a metric in space of quantum states.
Specifically, it satisfies the triangle inequality, symmetry, and identity of indiscernibles: 
    \begin{equation}
        \dav^{s}(\rho,\sigma) \leq \dav^{s}(\rho,\tau) + \dav^{s}(\tau,\sigma)\ \ \ 
         \text{for all}\ \ \rho,\sigma,\tau\in\states(\H_{dim}) 
    \end{equation}
    \begin{equation}
        \dav^{s}(\rho,\sigma)= \dav^{s}(\sigma,\rho)\ \ \ 
         \text{for all}\ \ \rho,\sigma\in\states(\H_{dim}) 
    \end{equation}
    \begin{equation}
         \dav^{s}(\rho,\sigma)=0 \Longleftrightarrow \rho=\sigma \ \ \ 
         \text{for all}\ \ \rho,\sigma\in\states(\H_{dim}) \ .
    \end{equation}
\end{lem}

\begin{proof}
The result follows directly from the fact, that $\dav^{s}$ is a Hilbert Schmidt distance.
\end{proof}

\begin{lem}[$\davS$ is subadditive]\label{lem:subbadditivity_states}
 For arbitrary quantum states $\rho_1,\sigma_1\in\states(\H), \rho_2,\sigma_2\in\states(\H)$, we have
\begin{equation}
    \davS(\rho_1\ot\rho_2,\sigma_1\ot\sigma_2)\leq  \dav(\rho_1,\sigma_1)+ \dav(\rho_2,\sigma_2)\ .
\end{equation}
\end{lem}

\begin{proof}
The proof follows from triangle inequality and multiplicativity with respect to the tensor product, i.e.
\begin{equation}
\begin{split}
    \| \rho_1\ot\rho_2 - \sigma_1\ot\sigma_2  \|_{\HS} &=
    \| \rho_1\ot(\rho_2 - \sigma_2) - (\sigma_1 - \rho_1)\ot\sigma_2  \|_{\HS}  \\
    &\leq 
     \| \rho_1 \|_{\HS} \| \rho_2 - \sigma_2 \|_{\HS}  + 
     \| \sigma_1 \|_{\HS} \|\sigma_1 - \rho_1 \|_{\HS}  \\
     &\leq 
     \| \rho_2 - \sigma_2 \|_{\HS}  + \|\sigma_1 - \rho_1 \|_{\HS}.
\end{split}
\end{equation}
\end{proof}

\begin{lem}[$\davS$ has joint-convextiy property]\label{lem:convexity_states}
 For arbitrary sets of quantum states $\cbracket{\rho_{\alpha}}_{\alpha}, \ \cbracket{\sigma_{\alpha}}_{\alpha}$ and probability distributions $\cbracket{p_{\alpha}}$, we have
\begin{equation}
    \davS \rbracket{\sum_{\alpha}p_{\alpha}\rho_{\alpha},
    \sum_{\alpha}p_{\alpha}\sigma{\alpha}} \leq \sum_{\alpha}p_{\alpha}\davS\rbracket{\rho_{\alpha},\sigma_{\alpha}} \ .
\end{equation}
\end{lem}

\begin{proof}
    The proof follows directly from triangle inequality,
\begin{equation}
\davS \rbracket{\sum_{\alpha}p_{\alpha}\rho_{\alpha},
    \sum_{\alpha}p_{\alpha}\sigma{\alpha}} 
=
\left \|
\sum_{\alpha}p_{\alpha} (\rho_{\alpha} -\sigma_{\alpha})
\right \|_{\HS} 
\leq \sum_{\alpha}p_{\alpha}
\left \| \rho_{\alpha} -\sigma_{\alpha} \right \|_{\HS}
= \sum_{\alpha}p_{\alpha}\davS\rbracket{\rho_{\alpha},\sigma_{\alpha}} \ .
\end{equation}
\end{proof}

\begin{lem}[Data-processing inequalities for average-case distance between 
states]\label{lem:dpi_states}
Average-case distance between states is monotonic with respect to unital maps, i.e., for a unital $\Phi$, we have
\begin{equation}
	\davS(\rho,\sigma) \geq \davS(\Phi(\rho),\Phi(\sigma))\ .
\end{equation}

\end{lem}

\begin{proof}
We begin the proof by reminding celebrated Uhlmann theorem~\cite{alberti1982stochasticity}, that for unital channel $\Phi$ and a Hermitian operator $H$, we have 
\begin{equation}
	\Phi(H) \prec H\ ,
\end{equation}
where the majorization relation above can be seen as a 
majorization between real vectors of eigenvalues.
We also note, that using the fact, that Hilbert-Schmidt norm is a Schur-convex function of eigenvalues, we get \begin{equation}
	\davS(\rho,\sigma) \geq \davS(\Phi(\rho),\Phi(\sigma))\ .
\end{equation}
\end{proof}

\begin{lem}[Separation between $\davS$ and $\dtr$]\label{lem:separation_states}
Let $\rho,\sigma \in \H_{\dim}$ be quantum states.
Then from standard inequalities between $1$ and $2$ norms, it follows that
\begin{align}
 \davS(\rho,\sigma) \    \leq\  \dtr(\rho,\sigma)\ \leq\  \sqrt{d}\ \davS(\rho,\sigma) \ .
\end{align}
\end{lem}
We now consider an example that attains the bound in Lemma~\ref{lem:separation_states}.
\begin{exa}[Two orthogonal maximally mixed states of rank $\frac{\dim}{2}$]\label{ex:states_separation}
Consider two states $\rho,\sigma \in \H_{\dim}$, such that $\rho=\frac{\iden_{\dim'}}{\dim'}, \sigma = \frac{\iden_{\dim'}}{\dim'}$, where $\dim'=\frac{\dim}{2}$ and $\tr(\rho\sigma)=0$.
Direct calculation yields
\begin{align*}
    \davS\rbracket{\rho, \sigma} & = \frac{1}{\sqrt{d}}  , \\
    \dtr\rbracket{\rho, \sigma} & = 1 \ . \numberthis
\end{align*}
Clearly, the above shows that in the asymptotic limit, the average-case distance between states goes to 0.
From the perspective of statistical distinguishability, it means that the states can be distinguished perfectly with optimal strategy, while randomized strategy fails dramatically.
\end{exa}

\begin{exa}[Counterexample for data-procesiing inequality for quantum states]\label{ex:states_dpi}
Consider two mixed states $\rho,\sigma \in \states(\H)$
from previous Example~\ref{ex:states_separation}.
Now consider a \emph{non-unital} quantum channel $\Lambda$ s.t. $\Lambda(\rho)=\ketbra{0}{0}$ and $\Lambda(\sigma)=\ketbra{1}{1}$.
Explicit computation combined with results of the previous example yields
\begin{align*}
    %
    \davS\rbracket{\Lambda(\rho),
    \Lambda(\sigma)} = \frac{1}{\sqrt{2}}\ \ > \ \  \frac{1}{\sqrt{d}} = \davS\rbracket{\rho, \sigma}  \ , \numberthis
\end{align*}
for $d>2$.
\end{exa}

\subsection{Quantum measurements}\label{sec:properties_measurements}

The following Table~\ref{tab:properties_measurements} summarizes properties of the average-case quantum distance between measurements and compares it to the worst-case operational distance.
For POVMs $\M$ and $\N$, symbol $\M\ot\N$ denotes a POVM with effects $\{M_i \ot N_j \}_{i,j}$.
Pre-processing channel $\Gamma$ acts on the state just before measurement $\M$.
This is equivalent to performing new POVM with effects transformed via dual channel $M_i\rightarrow\Gamma^{\ast}(M_i)$ on the original state \cite{cleanPOVM}.
The fact that channel $\Gamma$ is trace-preserving implies that dual channel $\Gamma^{\ast}$ is unital, which ensures that $\cbracket{\Gamma^{\ast}(M_i)}_{i}$ a proper POVM.
The post-processing stochastic map described by matrix $\Lambda$ transforms POVM's effects as $M_i\rightarrow \sum_{j}\Lambda_{ij}M_j$ (this can be interpreted as classical post-processing of classical outputs of the measurement).

\begin{table}[!h]
\hspace*{-1.1cm}
\begin{tabular}{|c|c|c|c|}
\hline
\textbf{Property}          & \textbf{Worst-case distance $\dop(\M,\N)$}                                                                                                                                & \textbf{Average-case distance $\davM(\M,\N)$}                                                                                                                                 & \textbf{\begin{tabular}[c]{@{}c@{}}Text reference for \\ \: average-case distance\: \end{tabular}} \\
& & & \\[-1em] 
\hline
& & & \\[-1em]
Function                 & $\frac{1}{2}\sup_{\rho \in \states(\H)} \sum_{i=1}^n|\tr(M_i\rho)-\tr(N_i\rho)|$                                                                                                             & \: $\frac{1}{2d}\sum_{i=1}^n \sqrt{ \|M_i-N_i\|_{\HS}^2 + \tr(M_i-N_i)^2}$\:                                                                                                       & Theorem~\ref{th:MEASav}, Lemma~\ref{lem:metric_measurements}                                 \\ 
& & & \\[-1em] 
\hline
& & & \\[-1em]
Subadditivity              & \: $\dop(\M_1\ot\M_2,\N_1\ot\N_2)\leq  \dop(\M_1,\N_1)+\dop(\M_2,\N_2)$\:                                                                                                      & same as for $\dop(\M,\N)$                                                                                                                                                     & Lemma~\ref{lem:subbadditivity_measuremenets}                                                 \\ 
& & & \\[-1em] 
\hline
& & & \\[-1em]
Joint convexity            & $\dop \rbracket{\sum_{\alpha}p_{\alpha}\M_{\alpha},  \sum_{\alpha}p_{\alpha}\N_{\alpha}} \leq \sum_{\alpha}p_{\alpha}\dop\rbracket{\M_{\alpha},\N_{\alpha}}$              & same as for $\dop(\M,\N)$                                                                                                                                                     & Lemma~\ref{lem:convexity_measurements}                                                       \\ 
& & & \\[-1em] 
\hline
& & & \\[-1em]
\: Data processing inequality \: & \begin{tabular}[c]{@{}c@{}}$\dop(\Lambda \circ \M \circ \Gamma, \Lambda \circ \N \circ \Gamma) \leq    \dop(\M ,\N )$ \\ for CPTP $\Gamma$, stochastic $\Lambda$\end{tabular} & \begin{tabular}[c]{@{}c@{}}$\davM(\Lambda \circ \M \circ \Phi, \Lambda \circ \N \circ \Phi) \leq    \davM(\M ,\N )$ \\ for unital $\Phi$, stochastic $\Lambda$\end{tabular} & Lemma~\ref{lem:dpi_measurements}                                                             \\ \hline
\end{tabular}
\caption{\label{tab:properties_measurements}}
\end{table}

\begin{lem}[$\davM$ fulfils axioms of a metric]\label{lem:metric_measurements}
Let $\davM$, denote average distances between quantum measurements (Eq.~\eqref{eq:avDISTpovm}).
Then $\davM$ satisfies axioms of a metric in space of POVMs.
Specifically, it satisfies the triangle inequality, symmetry, and identity of indiscernibles: 
\begin{equation}
        \davM(\M,\N) \leq \davM(\M,\mathsf{L}) + \davM(\mathsf{L},\N)\ \ \ 
         \text{for all}\ \ \M,\N, \mathsf{L} \in \povms(\H)
\end{equation}
\begin{equation}
        \davM(\M,\N)= \davM(\N,\M)\ \ \ 
         \text{for all}\ \ \M,\N  \in  \povms(\H)
\end{equation}
\begin{equation}
         \davM(\M,\N)=0 \Longleftrightarrow \M=\N \ \ \ 
         \text{for all}\ \ \M,\N \in  \povms(\H)  \ .
\end{equation}
Note, that $\davM(\M,\N)$ is absolute homogeneous, i.e. if we extend the definition of $\davM$ to arbitrary collections of operators, we see, that $\davM(s \M,s \N) = |s| \davM(\M,\N)$. \
\end{lem}

\begin{proof}
We note first that according to Eq.~\eqref{eq:avDISTpovm}, $\davM(\M,\N)$ is proportional to the sum of non-negative terms of the form
\begin{equation}\label{eqn:avM-terms}
\sqrt{ \|M_i-N_i\|_{\HS}^2 + \tr(M_i-N_i)^2} \ .
\end{equation}
First, we note, that both terms, treated as a functions  
$(M,N) \mapsto \|M-N\|_{\HS}$ and  $(M,N) \mapsto  |\tr(M-N)|$ satisfies triangle inequality, moreover the function $(a,b) \mapsto \sqrt{|a|^2+|b|^2}$ is subadditive and increasing in each argument. Therefore $\davM(\M,\N)$ obeys triangle inequality.  Symmetry, absolute homogeneity, and identity of indiscernibles follow from direct inspection.
\end{proof}

\begin{lem}[$\davM$ is subadditive]\label{lem:subbadditivity_measuremenets}
For arbitrary quantum measurements  $\M_1,\N_1\in \povms(\H_{\dim}, n), \M_2,\N_2\in \povms(\H_{\dim'},n')$, we have
\begin{equation}
    \dav^{m}(\M_1\ot\M_2,\N_1\ot\N_2)\leq  \davM(\M_1,\N_1)+ \davM(\M_2,\N_2)\ .
\end{equation}
\end{lem}

\begin{proof}
By triangle inequality we have 
\begin{equation}\label{eqn:measure-triangle}
\dav^{m}(\M_1\ot\M_2,\N_1\ot\N_2) 
\leq
\dav^{m}(\M_1\ot\M_2,\N_1\ot\M_2) 
+
\dav^{m}(\N_1\ot\N_2,\N_1\ot\M_2). 
\end{equation}
Now we consider one of the terms from the right-hand side of the inequality above 
and bound it by
\begin{equation}\label{eqn: sub-add-ineq1}
     \dav^{m}(\M_1\ot\M_2,\N_1\ot\M_2) \leq \dav^{m}(\M_1,\N_1). 
\end{equation}
The above inequality follows from direct calculations, since $\dav^{m}$ is a sum of square roots of the formulas for a single effect, for which we can write 
\begin{equation}
\begin{split}
& \|(M_1)_i\ot(M_2)_j  - (N_1)_i\ot(M_2)_j \|_{\HS}^2 + \tr((M_1)_i\ot(M_2)_j  - (N_1)_i\ot(M_2)_j)^2  \\
& = 
\|(M_1)_i - (N_1)_i\|_{\HS}^2 \|(M_2)_j \|_{\HS}^2 +
\tr((M_1)_i - (N_1)_i)^2 \tr( (M_2)_j)^2   \\
&\leq
\tr( (M_2)_j)^2 \left(
\|(M_1)_i - (N_1)_i\|_{\HS}^2 + \tr((M_1)_i - (N_1)_i)^2
\right) \ .
\end{split}
\end{equation}
Combining the terms above together with the fact, that $\sum_{j} \tr (M_2)_j = d'$, 
we obtain Eq.~\eqref{eqn: sub-add-ineq1}. Similarly, we can bound 
$\dav^{m}(\N_1\ot\N_2,\N_1\ot\M_2) \leq \dav^{m}(\N_2,\M_2)$ which, together with~Eq.~\eqref{eqn:measure-triangle} gives us the result.
\end{proof}

\begin{lem}[$\davM$ has joint-convexity property]\label{lem:convexity_measurements}
 For arbitrary sets of quantum measurements $\cbracket{\M_{\alpha}}_{\alpha}, \ \cbracket{\N_{\alpha}}_{\alpha}$ and probability distributions $\cbracket{p_{\alpha}}$, we have
\begin{equation}
    \davM \rbracket{\sum_{\alpha}p_{\alpha}\M_{\alpha},
    \sum_{\alpha}p_{\alpha}\N_{\alpha}} \leq \sum_{\alpha}p_{\alpha}\davM\rbracket{\M_{\alpha},\N_{\alpha}} \ .
\end{equation}
\end{lem}

\begin{proof}
The proof is analogous to the one for states and follows from
triangle inequality, and  absolute homogeneity:
\begin{equation}
    \davM \rbracket{\sum_{\alpha}p_{\alpha}\M_{\alpha},
    \sum_{\alpha}p_{\alpha}\N_{\alpha}}
\leq
\sum_{\alpha}\davM\rbracket{ p_{\alpha} \M_{\alpha},p_{\alpha} \N_{\alpha}} 
= \sum_{\alpha}p_{\alpha}\davM\rbracket{\M_{\alpha},\N_{\alpha}} \ .
\end{equation}
\end{proof}

\begin{lem}[Data-processing inequalities for average-case distance between measurements]\label{lem:dpi_measurements}
Average-case distance between quantum measurements is monotonic with respect to a unital pre- and general post-processing, i.e. for a stochastic matrix $\Lambda$ and a general unital CPTP map $\Phi$, we have
\begin{equation}
    \davM(\Lambda \circ \M \circ \Phi, \Lambda \circ \N \circ \Phi) \leq
    \davM(\M ,\N )\ .
\end{equation}

\end{lem}

\begin{proof}
We will show, that the average-case distance between measurements is monotonic with respect to post-processing. Since the outcome of a measurement is classical, we will consider only classical post-processing, given by a stochastic matrix $\Lambda$.
We denote by $\Delta_j = M_j - N_j$.
We will use the fact, that each term in the sum which defines $\davM(\M,\N)$,  is absolutely homogeneous and obeys the triangle inequality 
(see Eq.~\eqref{eq:avDISTpovm} and discussion under Eq.~\eqref{eqn:avM-terms})
\begin{equation}
\begin{split}
&
\davM(\Lambda \circ \M, \Lambda \circ \N) = \\ 
&=
\frac{1}{2 d }\sum_{i}^n \sqrt{\tr{(\sum_{j} \Lambda_{ij} \Delta_j)^2} + (\tr{\sum_{j} \Lambda_{ij} \Delta_j})^2} \leq \\
&\leq 
\frac{1}{2 d }\sum_{i,j}^n \Lambda_{ij}  \sqrt{\tr{(\Delta_j)^2} + (\tr{\Delta_j})^2} = \\
&=
\frac{1}{2 d }\sum_{j}^n \sqrt{\tr{(\Delta_j)^2} + (\tr{\Delta_j})^2} = 
\davM(\M,  \N)\ .
\end{split}
\end{equation}
In order to show that the average-case distance between quantum measurements
is monotonic with respect to unital pre-processing, we consider a general 
unital CPTP map $\Phi$. 
Note, that the adjoint map $\Phi^*$ is also unital and CPTP. 
Recall that we can look at the adjoint action of the channel on the effects of $\M$ as
\begin{equation}
\tr M_i \Phi (\rho) =   \tr \Phi^{*}(M_i) \rho\ .
\end{equation}
The fact, that $\Phi^*$ is unital assures us that $\M'$ with effects $\{\Phi^{*}(M_i)\}_i$ forms a POVM 
.
Now we consider the basic terms, which define $\davM$, first we see (again using Uhlmann's theorem ~\cite{alberti1982stochasticity} and Schur convexity of HS-norm)
\begin{equation}
    \| \Phi^{*}(\Delta_i) \|_{\HS}^2 \leq \| \Delta_i \|_{\HS}^2\ .
\end{equation}
Next, since $\Phi^{*}$ is trace-preserving we have 
\begin{equation}
    \tr ( \Phi^{*}(\Delta_i) )^2 = \tr (\Delta_i)^2\ ,
\end{equation}
which finishes the proof of monotonicity with respect to unital pre-processing.
\end{proof}

\begin{lem}[Separation between $\davM$ and $\dop$]\label{lem:separation_povms}
For any quantum measurements $\M,\N \in \povms(\H_{\dim})$, we have
\begin{align}
 \cpovms\ \davM(\M,\N) \ \leq \  \dop(\M,\N)\leq d\ \davM(\M,\N) \ ,
\end{align}
where $\cpovms=0.31$.
\end{lem}
\begin{proof}
The lower bound follows from Theorem~\ref{th:MEASav}.
For the upper bound, we directly calculate
    \begin{align*}
       \dop(\M,\N) &=  \frac{1}{2} \max_\rho  \sum_i | \tr (M_i - N_i) \rho| \leq
\frac{1}{2} \sum_i \sqrt{\|(M_i - N_i)\|_\HS^2 } \leq \\ &\leq
 \frac{1}{2} \sum_i \sqrt{\|(M_i - N_i)\|_\HS^2 + \tr(M_i-N_i)^2} 
= d \ \davM(M,N) \ .
    \end{align*}
\end{proof}
The following example attains bound in Lemma~\ref{lem:separation_povms} up to a constant.
\begin{exa}[Swapping two outcomes of standard measurement]\label{ex:povms_separation}
Consider computational-basis measurement $\P$ in $\H_\dim$  with effects $P_{i}=\ketbra{i}{i})$,
and second measurement $\M$ that is obtained from $\P$ by exchanging the first two effects, leaving others intact, i.e., $M_1=\ketbra{2}{2}, M_2=\ketbra{1}{1}$, and $M_i=\ketbra{i}{i}$ for $i=3,\ldots,\dim$.
In this scenario, direct calculation yields
\begin{align*}
    \davM(\P,\M) &= \sqrt{2}\ \frac{1}{d} \ , \\
    \dop(\P,\M) &= 1\ .
\end{align*}
The above implies that in the asymptotic limit, similarly to Example~\ref{ex:states_separation} for states, considered measurements can be distinguished perfectly with optimal strategy, while randomized one will not work.
On the other hand, if we interpret the second measurement $\M$ as a noisy version of target $\P$, then this particular type of noise (that swaps two measurement outcomes) will not highly affect the results of generic experiments.

We note that the above example, together with asymptotic separation, can be easily generalized to a scenario where the second measurement, instead of swapping only 2 outcomes of $\P$, swaps some \emph{constant} number of them.
\end{exa}

\begin{exa}[Counterexample for data-procesing inequality for quantum measurements]\label{ex:povms_dpi}
Consider POVMs $\P$ and $\M$ from previous Example~\ref{ex:povms_separation}.
Consider now a \emph{non-unital} channel $\Lambda$ that regardless of the input state prepares a state $\ketbra{1}{1}$ (which is a possible choice for optimal discriminator of POVMs $\M$ and $\P$).
Dual action of this channel on POVM's effects is $\Lambda^{\dag}(M_i)=\ \tr(M_i\ketbra{1}{1})\ \iden$.
The direct calculation, together with results from the previous example, yields
\begin{align*}
    \davM(\P\circ \Lambda,\M\circ \Lambda) =  \sqrt{1+\frac{1}{d}}\ \ > \ \ \sqrt{2}\ \frac{1}{d} = \davM(\P,\M)\ . \numberthis
\end{align*}
\end{exa}

\subsection{Quantum channels}\label{sec:properties_channels}

The following Table~\ref{tab:properties_channels} summarizes properties of the average-case quantum distance between channels and compares it to the worst-case diamond-norm distance. 
Compared to previous Tables, here we also consider two additional properties relevant for quantum channels, namely stability and chaining \cite{gilchrist2005distances} which for average-case distance hold for \emph{unital} quantum channels.
Stability means that a given distance measure does not change if a channel is extended by an identity channel.
In other words, trivial extensions of  maps by an ancillary system do not affect their distance measure.
Chaining means that distance between multiple compositions of the channel is at most a sum of the distances between constituting channels.
If one sequence is a composition of target gates, and the other is their noisy version, this property implies that the total error is at most additive in a given distance measure.

\begin{table}[!h]
\hspace*{-1.1cm}
\begin{tabular}{|c|c|c|c|}
\hline
\textbf{Property}          & \textbf{Worst-case distance $\ddiam(\Lambda,\Gamma)$}                                                                                                                               & \textbf{Average-case distance $\davC(\Lambda,\Gamma)$}                                                                                                     & \textbf{\begin{tabular}[c]{@{}c@{}}Text reference 
for \\ \: average-case distance\:
\end{tabular}} \\ 
& & & \\[-1em] 
\hline
& & & \\[-1em]
Function                   & $||\Lambda-\Gamma||_{\diamond}$                                                                                                                                                     & \: $ \frac{1}{2} \sqrt{\left\|\J_{\Lambda}-\J_{\Gamma}\right\|_{\HS}^2+\tr\left((\Lambda-\Gamma) [\tau_\dim]^2 \right) }$ \:                               &

Theorem~\ref{th:CHANNELSav},  Lemma~\ref{lem:metric_channels}
\\ 
& & & \\[-1em] 
\hline
& & & \\[-1em]
Subadditivity              & \: $\ddiam(\Lambda_1\ot\Lambda_2,\Gamma_1\ot\Gamma_2)\leq  \ddiam(\Lambda_1,\Gamma_1)+ \ddiam(\Lambda_2,\Gamma_2)$ \:                                                                  & same as for $\ddiam(\Lambda,\Gamma)$                                                                                                                    & Lemma~\ref{lem:subbadditivity_channels}                                                      \\ 
& & & \\[-1em] 
\hline
& & & \\[-1em]
Joint convexity            & $\ddiam \rbracket{\sum_{\alpha}p_{\alpha}\Lambda_{\alpha},    \sum_{\alpha}p_{\alpha}\Gamma_{\alpha}} \leq\sum_{\alpha}p_{\alpha}\ddiam\rbracket{\Lambda_{\alpha},\Gamma_{\alpha}}$ & same as for $\ddiam(\Lambda,\Gamma)$                                                                                                                    & Lemma~\ref{lem:convexity_channels}                                                           \\ 
& & & \\[-1em] 
\hline
& & & \\[-1em]
\: \begin{tabular}[c]{@{}c@{}}Data processing\\  inequality\end{tabular} \: & \begin{tabular}[c]{@{}c@{}}$\ddiam(\Phi_o \circ  \Lambda \circ  \Phi_i, \Phi_o \circ  \Gamma \circ  \Phi_i) \leq \ddiam(\Lambda,\Gamma)$ \\  for CPTP $\Phi_i,\ \Phi_o$\end{tabular}                            & \begin{tabular}[c]{@{}c@{}}$\davC(\Phi_o \circ  \Lambda \circ  \Phi_i, \Phi_o \circ  \Gamma \circ  \Phi_i) \leq \davC(\Lambda,\Gamma)$ \\ for unital $\Phi_i,\ \Phi_o$\end{tabular} & Lemma~\ref{lem:dpi_channels}                                                                \\ 
& & & \\[-1em] 
\hline
& & & \\[-1em]

Stability                  & \begin{tabular}[c]{@{}c@{}}$\ddiam(\Lambda \otimes \idenC, \Gamma \otimes \idenC) = \ddiam(\Lambda,\Gamma) $\\ for CPTP $\Lambda, \Gamma$\end{tabular}                                                            & \begin{tabular}[c]{@{}c@{}}$\davC(\Lambda \otimes \idenC, \Gamma \otimes \idenC) = \davC(\Lambda,\Gamma)$\\ for unital $\Lambda, \Gamma$\end{tabular}                                                           & Lemma~\ref{lem:stability}                                                                   \\ 
& & & \\[-1em] 
\hline
& & & \\[-1em]

Chaining                   & \begin{tabular}[c]{@{}c@{}}$\ddiam(\Lambda_1 \circ \Lambda_2, \Gamma_1 \circ \Gamma_2) \leq \ddiam(\Lambda_1,\Gamma_1)+\ddiam(\Lambda_2,\Gamma_2)$\\ for CPTP $\Lambda_1,\Lambda_2, \Gamma_1,\Gamma_2$\end{tabular} & \begin{tabular}[c]{@{}c@{}}$\davC(\Lambda_1 \circ \Lambda_2, \Gamma_1 \circ \Gamma_2) \leq \davC(\Lambda_1,\Gamma_1)+\davC(\Lambda_2,\Gamma_2)$\\ for unital $\Lambda_1,\Lambda_2, \Gamma_1,\Gamma_2$\end{tabular} & Lemma~\ref{lem:chaining}    
\\[-1em]
&&& \\
\hline

\end{tabular}
\caption{\label{tab:properties_channels}}
\end{table}

\begin{lem}[$\davC$ fulfils axioms of a metric]\label{lem:metric_channels}
Let $\davC$, denote average distances between channels (Eq.~\eqref{eq:avDistanceChannels}).
Then $\davC$ satisfies axioms of a metric in space of quantum channels.
Specifically, it satisfies the triangle inequality, symmetry, and identity of indiscernibles: 
\begin{equation}
        \davC(\Lambda ,\Gamma) \leq \davC(\Lambda,\Phi) + \davC(\Phi,\Gamma )\ \ \ 
         \text{for all}\ \ \Lambda ,\Gamma ,\Phi \in \channels(\H_{\dim})
\end{equation}
\begin{equation}
        \davC(\Lambda ,\Gamma)= \davC(\Gamma,\Lambda)\ \ \ 
         \text{for all}\ \  \Lambda ,\Gamma \in  \channels(\H_{\dim})
\end{equation}
\begin{equation}
         \davC(\Lambda ,\Gamma)=0 \Longleftrightarrow \Lambda = \Gamma \ \ \ 
         \text{for all}\ \ \Lambda ,\Gamma \in  \channels(\H_{\dim})\ .
\end{equation}
Note, that $\davC(\Lambda ,\Gamma)$ is absolute homogeneous, i.e. 
$\davC(s \Lambda ,s \Gamma) = |s| \davC(\Lambda ,\Gamma)$.
\end{lem}

\begin{proof}
Note that $\davC$ is a function of a distance measure
($\|\J_{\Lambda} - \J_{\Gamma}\|_{HS}$) and a term 
($ \sqrt{\tr\left((\Lambda-\Gamma) [\tau_d]^2 \right)} $), which treated as a function, obeys the triangle inequality.
Since the function $(a,b) \mapsto \sqrt{|a|^2+|b|^2}$ is subadditive and increasing in each argument,
thus $\davC$ obeys triangle inequality.  Symmetry and identity of indiscernibles follows from direct inspection.
\end{proof}

\begin{lem}[$\davC$ is subadditive]\label{lem:subbadditivity_channels}
 For arbitrary quantum channels $\Lambda_1,\Gamma_1\in \channels(\H), \Lambda_2,\Gamma_2\in \channels(\CC^{d'})$, we have
\begin{equation}
    \dav^{ch}(\Lambda_1\ot\Lambda_2,\Gamma_1\ot\Gamma_2)\leq  \dav(\Lambda_1,\Gamma_1)+ \dav(\Lambda_2,\Gamma_2)\ .
\end{equation}
\end{lem}

\begin{proof} 
We begin with triangle inequality 
\begin{equation}\label{eqn:ch-triangle}
\dav^{ch}(\Lambda_1\ot\Lambda_2,\Gamma_1\ot\Gamma_2)
\leq
\dav^{ch}(\Lambda_1\ot\Lambda_2,\Gamma_1\ot\Lambda_2)
+
\dav^{ch}(\Gamma_1\ot\Gamma_2,\Gamma_1\ot\Lambda_2) \ .
\end{equation}
Now we consider 
\begin{equation}
\dav^{ch}(\Lambda_1\ot\Lambda_2,\Gamma_1\ot\Lambda_2)
=
\sqrt{
\| \J_{\Lambda_1\ot\Lambda_2} -\J_{\Gamma_1\ot\Lambda_2} \|_{\HS}^2
+
\tr \left(  ((\Lambda_1-\Gamma_1)\ot\Lambda_2) \left(\frac{\I_{dd'}}{d d'} \right) ^2 \right) 
}\ .
\end{equation}
First we note, that $\J_{\Lambda_1\ot\Lambda_2} $ is permutationally similar, to $\J_{\Lambda_1} \ot \J_{\Lambda_2}$.
In fact, we have $\J_{\Lambda_1\ot\Lambda_2}  =  \mathbb{S}_{23} \left(\J_{\Lambda_1} \ot \J_{\Lambda_2} \right) \mathbb{S}_{23}$, where $\mathbb{S}_{23}$ is the swap of the second and third subsystems, i.e. $\Lambda_1$-input system and $\Lambda_2$-output system. Therefore 
\begin{equation}
\begin{split}
\| \J_{\Lambda_1\ot\Lambda_2} -\J_{\Gamma_1\ot\Lambda_2} \|_{\HS}
&=
\| \J_{\Lambda_1}\ot\J_{\Lambda_2} -\J_{\Gamma_1}\ot\J_{\Lambda_2} \|_{\HS} \\
&=
\| \J_{\Lambda_1} -\J_{\Gamma_1}\|_{\HS}  \| \J_{\Lambda_2} \|_{\HS} \\
&\leq 
\| \J_{\Lambda_1}  - \J_{\Gamma_1}\|_{\HS}\ .
\end{split}
\end{equation}
Next, we note 
\begin{equation}
\begin{split}
\tr \left(  ((\Lambda_1-\Gamma_1)\ot\Lambda_2) \left(\frac{\I_{dd'}}{d d'} \right) ^2 \right) 
&=
\tr \left(  (\Lambda_1-\Gamma_1)\left(\frac{\I_{d}}{d} \right) ^2 \right)
\tr \left(\Lambda_2 \left(\frac{\I_{d'}}{d'} \right)^2 \right) \\    
&\leq 
\tr \left(  (\Lambda_1-\Gamma_1)\left(\frac{\I_d}{d} \right) ^2 \right).
\end{split}    
\end{equation}
Combining the above we get 
\begin{equation}
\dav^{ch}(\Lambda_1\ot\Lambda_2,\Gamma_1\ot\Lambda_2)
\leq 
\sqrt{\| \J_{\Lambda_1}  - \J_{\Gamma_1}\|_{\HS}^2 + \tr \left(  (\Lambda_1-\Gamma_1)\left(\frac{\I}{d} \right) ^2 \right)}
= \dav^{ch}(\Lambda_1,\Gamma_1)\ .
\end{equation}
We can analogously bound $\dav^{ch}(\Gamma_1\ot\Gamma_2,\Gamma_1\ot\Lambda_2) 
\leq \dav^{ch}(\Gamma_2,\Lambda_2)$ and using~Eq.~\eqref{eqn:ch-triangle} we obtain the result.
\end{proof}

\begin{lem}[$\davC$ has joint-convexity property]\label{lem:convexity_channels}
 For arbitrary sets of quantum channels  $\cbracket{\Lambda_{\alpha}}_{\alpha}, \ \cbracket{\Gamma_{\alpha}}_{\alpha}$ and probability distributions $\cbracket{p_{\alpha}}$, we have
\begin{equation}
    \davC \rbracket{\sum_{\alpha}p_{\alpha}\Lambda_{\alpha},
    \sum_{\alpha}p_{\alpha}\Gamma_{\alpha}} \leq
    \sum_{\alpha}p_{\alpha}\davC\rbracket{\Lambda_{\alpha},\Gamma_{\alpha}} \ .
\end{equation}
\end{lem}

\begin{proof}
The proof is analogous to the one for states and measurements and follows from
triangle inequality, and  absolute homogeneity of $\davC$.
\end{proof}

\begin{lem}[Data-processing inequalities for average-case distance between 
channels]\label{lem:dpi_channels}
Average-case distance between quantum channels is monotonic with respect to unital pre- 
and postprocessing, i.e. for a unital maps $\Phi_o, \Phi_i$, we have
\begin{equation}
\davC(\Phi_o \circ  \Lambda \circ  \Phi_i, \Phi_o \circ  \Gamma \circ  \Phi_i) 
\leq 
\davC( \Lambda , \Gamma) .
\end{equation}
\end{lem}
\begin{proof}

The inequality related to the postprocessing follows directly analogous 
results for states, in order to show the monotonicity with respect to the preprocessing inequality we write, for 
unital $\Phi$
\begin{equation}
\begin{split}
	&\davC(\Lambda \circ  \Phi,\Gamma \circ  \Phi) \\
	&= \frac{1}{2} 
	\sqrt{\left\|
		\J_{\Lambda \circ \Phi}-\J_{\Gamma \circ \Phi}\right
	\|_{\HS}^2+\tr\left((\Lambda-\Gamma) [\frac{\I}{d}] 
	\right)^2 } .
\end{split}
\end{equation}
We can consider only the term $\left\|
\J _{\Lambda \circ \Phi} -\J_{\Gamma \circ \Phi}\right
\|_{\HS}$, since the second one does not change under preprocessing by a unital 
map. First, we write a norm in terms of superoperators, i.e. 
\begin{equation}
\left\|
\J_{\Lambda \circ \Phi}-\J_{\Gamma \circ \Phi}\right
\|_{\HS}  
=
\left\|
(\hat \Lambda - \hat\Gamma) \hat\Phi
\right \|_{\HS},
\end{equation}
where $\hat\Lambda$ denotes the superoperator matrix (\cite{bengtsson_zyczkowski_2006}) of channel $\Lambda$.
Now we use inequality
\begin{equation}\label{eq:HS_inequality}
\left\| A B \right \|_{\HS}  
\leq  \left\| A \right \|_{\HS}   \left\| B \right \|_{\infty}
\end{equation}
and write
\begin{equation}
\left\|
\J_{\Lambda \circ \Phi}-\J_{\Gamma \circ \Phi}\right
\|_{\HS}  
\leq 
\left\|
\hat \Lambda - \hat\Gamma
\right \|_{\HS}  
\left\|
\hat\Phi
\right \|_{\infty}.  
\end{equation}
Now since for any unital map we have $\left\|\hat\Gamma\right \|_{\infty} = 
1$ (see~\cite[Theorem 1]{roga2013entropic}), we obtain the result.
\end{proof}

\begin{lem}[Stability property of average-case distance between unital
channels]\label{lem:stability}
Average-case distance between unital quantum channels fulfills stability property \cite{gilchrist2005distances}, i.e. for unital maps $\Lambda, \Gamma$, and identity channel $\idenC$ acting on arbitrary dimension, we have
\begin{equation}
\davC(\Lambda \otimes \idenC, \Gamma \otimes \idenC) 
=
\davC( \Lambda , \Gamma) \ .
\end{equation}
\end{lem}
\begin{proof}
    For unital channels we have $\davC(\Lambda,\Gamma) = \frac{1}{2}||\Choi_{\Lambda}-\Choi_{\Gamma}||_{HS}$. 
    We now recall that for any channels $\Lambda,\Gamma$, Choi matrix $\Choi_{\Lambda \otimes \Gamma}$ is permutationally similar to $\Choi_{\Lambda}\otimes\Choi_{\Gamma}$.
    This allows to rewrite the HS norm as
    \begin{align*}
        ||\Choi_{\Lambda \otimes \idenC}-\Choi_{\Gamma \otimes \idenC}||_{\HS} & = ||(\Choi_{\Lambda}-\Choi_{\Gamma})\otimes\Choi_{\idenC}||_{\HS} = ||\Choi_{\Lambda}-\Choi_{\Gamma}||_{\HS}\ \underbrace{||\Choi_{\idenC}||_{\HS}}_{=1} = ||\Choi_{\Lambda}-\Choi_{\Gamma}||_{\HS}\ ,
    \end{align*}
which concludes the proof.
\end{proof}

\begin{rem}
For generic, non-unital channels, the expression for average-case distance has additional term $\tr\left((\Gamma-\Lambda)(\frac{\iden}{d})\right)^2$\ . 
If we extend our channels by identity $\idenC_{d'}$ on dimension $d'$, this 'non-unitality' term changes to $\tr\left(((\Gamma-\Lambda)\otimes \idenC_{d'})(\frac{\iden}{d d'})\right)^2 = \frac{1}{d'} \tr\left((\Gamma-\Lambda)(\frac{\iden}{d})\right)^2$. Therefore, the contribution to the average-case distance of the 'non-unitality' decreases as $d'$ increases. 
Note that this scenario corresponds to channel discrimination (via random circuits) with the use of an ancillary system.
\end{rem}

\begin{lem}[Chaining property of average-case distance between unital
channels]\label{lem:chaining}
Average-case distance between unital quantum channels fulfills chaining property \cite{gilchrist2005distances}, i.e. for unital maps $\Lambda_1, \Lambda_2, \Gamma_1, \Gamma_2$, we have
\begin{equation}
\davC(\Lambda_1 \circ \Lambda_2, \Gamma_1 \circ \Gamma_2) 
\leq 
\davC( \Lambda_1 , \Gamma_1) + \davC( \Lambda_2 , \Gamma_2) .
\end{equation}
\end{lem}
\begin{proof}
To prove the theorem, we apply triangle inequality followed by the data-processing inequality for unital channels (Lemma~\ref{lem:dpi_channels})
\begin{align}
  \davC(\Lambda_1 \circ \Lambda_2, \Gamma_1 \circ \Gamma_2) 
\leq 
\davC(\Lambda_1 \circ \Gamma_2,
\Gamma_1 \circ \Gamma_2) 
+
\davC(\Lambda_1 \circ \Gamma_2,
\Lambda_1 \circ \Lambda_2
) 
\leq 
\davC( \Lambda_1 , \Gamma_1) + \davC( \Lambda_2 , \Gamma_2) \ .
\end{align}
\end{proof}

\begin{rem}
We note that for generic, non-unital channels, the chaining property of average-case distance does not hold.
To see that, we note that if we choose channels $\Lambda_1=\Gamma_1$ to be the same, the chaining property effectively reduces to data-processing inequality, which we know does not hold for generic channels (see below for a counterexample). 
\end{rem}

\begin{lem}[Separation between $\davC$ and $\ddiam$]\label{lem:separation_channels}
For any quantum measurements $\Lambda,\Gamma \in \channels(\H_{\dim})$, we have
\begin{align}
 \cchannels\ \davC(\Lambda,\Gamma) \ \leq \  \ddiam(\Lambda,\Gamma)\leq d^{\frac{3}{2}}\ \davC(\Lambda,\Gamma) \ ,
\end{align}
where $\cchannels=0.087$.
\end{lem}
\begin{proof}
The lower bound is a consequence of Theorem~\ref{th:CHANNELSav}, note that the constant
here can be improved.
To show the other inequality we begin with the upper bound for diamond norm (see ~\cite[Thm. 2]{jencova2016} and \cite[Prop. 1]{nechita2018almost}), which for Hermiticity preserving operation can be written in our notation as
\begin{align}
 \ddiam(\Lambda,\Gamma)
 \leq 
 \frac{d}{2} \left\|\tr_2 ( |\Choi_\Lambda - \Choi_\Gamma|) \right\|_{\infty}\ ,
\end{align}
Next express the operator norm via maximization over pure states on the first subsystem 
\begin{equation}
 \left\|\tr_2 ( |\Choi_\Lambda - \Choi_\Gamma|) \right\|_{\infty}
 = \max_{\psi\in\pstates(\H_\dim)} |\tr \left( \psi\ot \iden_\dim  |\Choi_\Lambda - \Choi_\Gamma| \right)|\ . 
\end{equation}
Applying to the above Cauchy-Schwarz inequality we obtain 
\begin{equation}
     \left\|\tr_2 ( |\Choi_\Lambda - \Choi_\Gamma|) \right\|_{\infty}\leq   \max_{\psi\in\pstates(\H_\dim)}\|\psi \otimes \iden_\dim  \|_{\HS} \, \|\Choi_\Lambda - \Choi_\Gamma\|_{\HS}
 =
 \sqrt{d} \, \|\Choi_\Lambda - \Choi_\Gamma\|_{\HS} \ .
\end{equation}
Combining the above  we obtain the desired result
\begin{align}
 \ddiam(\Lambda,\Gamma)
 \leq 
 \frac{d}{2} \sqrt{d} \, \|\Choi_\Lambda - \Choi_\Gamma\|_{\HS}
 \leq d^{\frac32} \  \davC(\Lambda, \Gamma)\ .
\end{align}
\end{proof}

\begin{exa}[Separation example]
Let us consider even dimensional Hilbert space $\H_\dim$  and a Hermitian matrix $A$, such that $\tr A = 0$ and $A^2 = \iden_\dim$. Next, we define a pair of channels $\Lambda$ and $\Gamma$
by their Jamiołkowski states as
\begin{equation}
\begin{split}
    \Choi_\Lambda &= \frac{1}{d^2} \, \iden_{d^2} \, ,  \\ 
    \Choi_\Gamma  &= \frac{1}{d^2} \, \iden_{d^2} - \frac{1}{d^2}
    \ketbra{\psi}{\psi} \otimes A\ ,
\end{split}
\end{equation}
where $\psi\in \pstates(\H_\dim)$ is an arbitrary pure state.
The diamond norm between $\Lambda$ and $\Gamma$ can be calculated easily,
using an alternative formula for the diamond norm for Hermiticity preserving operations
(see e.g.~{\cite[Eqn. (11)]{Puchala2018}}),
i.e.,
\begin{equation}
\begin{split}
    \|\Lambda-\Gamma\|_{\diamond} &=
    d \max_{\rho \in \states(\H_\dim)} \|(\sqrt{\rho} \otimes \iden) \Choi_{\Lambda-\Gamma} (\sqrt{\rho} \otimes \iden) \|_1
    =
    d \max_{\rho \in \states(\H_\dim)} \|(\sqrt{\rho} \otimes \iden) (\frac{1}{d^2}
    \ketbra{\psi}{\psi} \otimes A) (\sqrt{\rho} \otimes \iden) \|_1 \\
&=
\frac{1}{d} \|A\|_1 = 1.
\end{split}
\end{equation}
The average distance can be evaluated as 
\begin{equation}
\begin{split}
\davC( \Lambda, \Gamma)  &= 
\frac12 \sqrt{
\| \Choi_{\Lambda - \Gamma}\|_{\HS}^2 + \|\Lambda(\iden / d) - \Gamma(\iden / d) \|_{\HS}^2
} \\
&=
\frac12 \sqrt{
\| \frac{1}{d^2} \ketbra{\psi}{\psi} \otimes A \|_{\HS}^2 +
\| \frac{1}{d^2} \tr_{1} (\ketbra{\psi}{\psi} \otimes A) \|_{\HS}^2
} \\
&=
\frac12 \sqrt{
\frac{1}{d^4} \| A \|_{\HS}^2 +
\frac{1}{d^4} \| A \|_{\HS}^2
} = 
\frac{\sqrt{2}}{2 d^\frac32}. 
\end{split}
\end{equation}
Which gives us finally the separation of order $d^{\frac32}$,
\begin{equation}
  \frac12 = \ddiam(\Lambda,\Gamma) = d^{\frac32} \frac{1}{\sqrt{2}} \ \davC(\Lambda,\Gamma).
\end{equation}
\end{exa}

\begin{exa}[Counterexample for general post-processing monotonicity for quantum channels]
Consider two state-preparation channels acting on $N$-qubit system as $\Lambda(\rho)= \tr(\rho) \ketbra{0}{0}\ot \frac{\iden}{2^{N-1}}$ and $\Gamma(\rho)=\tr(\rho) \ketbra{1}{1}\ot \frac{\iden}{2^{N-1}}$, for any input state $\rho \in \states(\H)$.
Then we have
\begin{align*}\label{eq:counter_channels_1}
    \davC(\Lambda,\Gamma) &= \frac{1}{2}\sqrt{\frac{1}{d}(1+\frac{1}{d})} \\
    \ddiam(\Lambda,\Gamma) &=1 \ , \numberthis
    \end{align*}
where expression for average-case distance follows from direct calculation, and the value of diamond norm follows from the fact that channels always prepare states that are orthogonal on first qubits, and thus can be perfectly distinguished.

Now consider additional \emph{non-unital} conditional state-preparation channel $\tilde{\Lambda}$ that acts as $\tilde{\Lambda}(\ketbra{0}{0}\ot \sigma) = \psi$ and $\tilde{\Lambda}(\ketbra{1}{1}\ot \sigma) = \psi^{\perp}$ for any $\sigma$, where $\psi,\psi^{\perp}$ are two orthogonal pure states.
Note that the composed action of the channels reduces to state-preparation channels $\tilde{\Lambda}\circ\Lambda(\rho)=\psi$ and $\tilde{\Lambda}\circ\Gamma(\rho)=\psi^{\perp}$ for any $\rho$.
Direct computation together with Eq.~\eqref{eq:counter_channels_1} yields
\begin{align*}
    \davC(\tilde{\Lambda}\circ\Lambda,\tilde{\Lambda}\circ\Gamma) = \sqrt{\frac{1}{2}(1+\frac{1}{d})} \ \ &> \ \ \frac{1}{2}\sqrt{\frac{1}{d}(1+\frac{1}{d})}=  \davC(\Lambda,\Gamma) \ . \numberthis \\ 
\end{align*}
\end{exa}

\begin{exa}[Counterexample for general pre-processing monotonicity for quantum channels]
Consider two perfectly distinguishable unitary channels of size $d>2$,
$\Lambda_U: \rho \mapsto U \rho U^\dagger$ and
$\Lambda_V: \rho \mapsto V \rho V^\dagger$.

The average distance can be calculated directly and is equal to (see also Example~\ref{exa:generic_unitary_channels})
\begin{equation}
\davC(\Lambda_U,\Lambda_V) =\frac12 \sqrt{2 - \frac{2}{d^2} |\tr U^\dagger V|^2}\ .
\end{equation}
Since channels $\Lambda_U$ and $\Lambda_V$ are perfectly distinguishable,
let $\ket{\psi}$ be the optimal discriminator, i.e. the state for which 
$\bra{\psi} U^\dagger V \ket{\psi} = 0$.
Note, that in the case of unitary channels, one does not need to attach 
an additional system in order to perform optimal discrimination.
Now we consider a channel $\Gamma : \rho \mapsto \tr(\rho) \ketbra{\psi}{\psi}$, which prepares the optimal discriminator. 
We then have 
\begin{equation}
\begin{split}
\Choi_{\Lambda_U \circ \Gamma} &= 
 \iden /d \otimes U\ketbra{\psi}{\psi}U^\dagger\ ,\\
\Choi_{\Lambda_V \circ \Gamma} &= 
\iden /d  \otimes  V\ketbra{\psi}{\psi}V^\dagger\ .
\end{split}
\end{equation}
Direct computations yield the following result
\begin{equation}
\begin{split}
\davC(\Lambda_U \circ \Gamma,\Lambda_V  \circ \Gamma) 
&=\frac12 \sqrt{
\|
\iden/d
\otimes 
( U\ketbra{\psi}{\psi}U^\dagger -V\ketbra{\psi}{\psi}V^\dagger )
\|_{\HS}^2
+ \tr(U\ketbra{\psi}{\psi}U^\dagger - V\ketbra{\psi}{\psi}V^\dagger)^2
} 
=
\frac12 \sqrt{\frac{2}{d} + 2}.
\end{split}
\end{equation}
Finally, we obtain that the data processing inequality for general pre-processing does not hold.
\begin{equation}
\davC(\Lambda_U \circ \Gamma,\Lambda_V  \circ \Gamma)  \ \ > \ \  
\davC(\Lambda_U,\Lambda_V). 
\end{equation}
If we choose $U = \iden, V = diag(1,-1,1, \dots ,1)$ we get  
$
\davC(\Lambda_U,\Lambda_V) =\frac12 \sqrt{2 - \frac{2}{d^2} (d-2)^2}
=\frac{1}{d} \sqrt{2 (d-1)}.
$

In fact, similar calculations can be performed on any distinguishable channels, with the pre-processing channel chosen to be the preparation of the optimal  discriminator.
\end{exa}


\section{Examples}\label{sec:further_examples}

In previous parts of the work, we discussed some specific scenarios in which scaling of average-case quantum distances with system size provided some insight into various areas of quantum information.
In this part, we investigate some further exemplary scenarios, and we provide a discussion of some of the consequences of our findings.

\subsubsection{Convergence to uniform distribution}
One particularly interesting consequence of our main theorems is that average-case distances allow us to easily study a convergence of the average Total-Variation distance between the noisy distribution and the uniform distribution. 
To this aim, one needs to calculate an average-case distance between a noisy state, measurement, or channel, and the maximally mixed state, trivial POVM, or maximally depolarizing channel,
respectively.
We summarize those observations in the following Lemmas~\ref{lem:uniform_states}, \ref{lem:uniform_povms}, \ref{lem:uniform_channels} -- the proofs follow directly from Theorems~\ref{th:STATESav}, \ref{th:MEASav}, and \ref{th:CHANNELSav}, respectively.
In what follows we denote uniform distribution as $\p^{\uniform}$, meaning $p^{\uniform}_i = \frac{1}{d}$ for all $i=1,\dots, d$.
All of the above Lemmas follow directly from Theorem~\ref{th:STATESav} (states), Theorem~\ref{th:MEASav} (measurements), and Theorem~\ref{th:CHANNELSav}.

\begin{lem}\label{lem:uniform_states}[Noisy states -- convergence to uniform distribution]
Let $\psi$ be a pure state and $\Lambda$ a quantum channel.
Then we have
\begin{align}\label{eq:states_uniform} 
\expect{\mathrm{U}\sim\nu} \dtv(\p^{\Lambda(\psi),U},\p^{\uniform}) \approx \davS(\Lambda(\psi),\frac{\iden}{d}) = \frac{1}{2}\sqrt{\tr{\left(\left(
\Lambda(\psi)
\right)^2\right)}-\frac{1}{d}} \ .
\end{align}
In the above, the $\approx$ sign means the approximation in a sense of Eq.~\eqref{eq:statesBOUNDS}.
The notation $\p^{\Lambda(\psi),U}$ is the same as for Theorem~\ref{th:STATESav}.
\end{lem}

From the above, it follows that the convergence of noisy distribution to the uniform in random circuits setting is controlled by the purity of the output state. 
For quantum measurements and channels, we have similar expressions.

\begin{lem}\label{lem:uniform_povms}[Noisy measurements -- convergence to uniform distribution]
Let $\M$ be a generic $d$-outcome quantum measurement on $d$-dimensional space, and $\M^{\idenC}$ a trivial POVM s.t. $M^{\idenC}_i = \frac{\iden}{d}$ for each $i=1,\dots,d$.
Then we have
\begin{align}\label{eq:povms_uniform} 
\expect{V\sim\nu} \dtv(\p^{\M,\psi_{V}},\p^{\uniform}) \approx \davM(\M, \M^{\idenC}) = \frac{1}{2d} \sum_{i=1}^{d}\ \sqrt{\tr{\left(
M_i^{2}\right)}+\left(\tr{M_i}-1\right)^2-\frac{1}{d}} \ .
\end{align}
In the above, the $\approx$ sign means the approximation in a sense of Eq.~\eqref{eq:povmBOUNDS}.
The notation $\p^{\M,\psi_{V}}$ is the same as for Theorem~\ref{th:MEASav}.
\end{lem}

\begin{lem}\label{lem:uniform_channels}[Noisy channels -- convergence to uniform distribution]
Let $\Lambda$ be a generic quantum channel and $\Lambda_{\mathrm{dep}}$ be a maximally depolarizing channel, i.e., $\Lambda_{\mathrm{dep}}(\rho)=\frac{\iden}{d}$ for any state $\rho$.
Then we have
\begin{align}\label{eq:channels_uniform} 
\expect{V\sim\nu} \ \expect{U\sim\nu} \dtv(\p^{\Lambda,\psi_V,U},\p^{\text{\uniform}}) \approx \davC(\Lambda,\Lambda_{\mathrm{dep}}) =   \frac{1}{2}\sqrt{\tr{\left(
\Choi_{\Lambda}^2\right)}+\tr{\left(\left(
\Lambda(\frac{\iden}{d})
\right)^2\right)}-\frac{1}{d}\left(1+\frac{1}{d}\right)}
\end{align}
In the above, the $\approx$ sign means the approximation in a sense of Eq.~\eqref{eq:channelsBOUND}.
The notation $\p^{\Lambda,\psi_V,U}$ is the same as for Theorem~\ref{th:CHANNELSav}.
\end{lem}

\subsubsection{More examples}

\begin{exa}[Two pure states]
For two pure states $\psi$ and $\phi$ we have
\begin{align*}
    \davS\rbracket{\psi,\phi} &= \frac{1}{\sqrt{2}} \sqrt{1-\tr\rbracket{\psi\phi}}\ , \\
    \dtr\rbracket{\psi,\phi} &= \sqrt{1-\tr\rbracket{\psi\phi}} \ . \numberthis
\end{align*}
Therefore, in this case, we see that $\dav\rbracket{\psi,\phi} = \frac{1}{\sqrt{2}}\dtr\rbracket{\psi,\phi}$, which gives only constant separation between average-case and worst-case scenarios.
\end{exa}

The consequences of the above example are twofold, depending on the perspective we adopt.
First, if we wish to perform a task of state discrimination between two pure states, then the above identity implies that there exists a strategy that uses random quantum circuits that is worse than the optimal strategy only by a constant.
Second, if we treat $\psi$ as our target state and $\phi$ as its noisy version affected by unwanted unitary rotation, then this type of noise will highly affect the quality of our results.
Specifically, for generic quantum states, it will behave similarly to the worst-case scenario.

\begin{exa}\label{ex:uniform_separable_states}[Pauli eigenstates and separable Pauli noise]
Consider state $\psi^{\text{pauli}} = \otimes_{i=1}^{N} \ketbra{\pm r_i}{\pm r_i}$, where $r_i \in \left\{x,y,z\right\}$, i.e., $\ket{\pm r_i}$ is any Pauli eigenstate on qubit $i$ (with eigenvalue $+1$ or $-1$.).
Consider separable Pauli channel $\Lambda^{\text{pauli}}=\otimes_{i=1}^N \Lambda_{i}^{\text{pauli}}$, where single-qubit channel is $\Lambda^{\text{pauli}}_{i}(\rho) = \sum_{j=1}\ p^{(i)}_j\sigma_j\rho\sigma_j$ with $j \in \left\{1,x,y,z\right\}$, $\sigma_1=\iden$, and $p^{(i)}_j\geq 0$, $\sum_{j}p^{(i)}_j=1$.
Define $q^{(i)} = p^{(i)}_1+p^{(i)}_{r_i}$, i.e., a probability of applying on qubit $i$ a gate that stabilizes the state of that qubit (namely, either identity or Pauli matrix of which $\ket{\pm r_i}$ is an eigenstate).
Furthermore, assume that for each qubit $i$ we have $q^{(i)}\geq\frac{1}{2}$.
Then we have
\begin{align}\label{eq:states_uniform_example}
    \davS(\Lambda^{\text{pauli}}(\psi^{\text{pauli}}), \frac{\iden}{d}) = \frac{1}{2} \sqrt{\Pi_{i=1}^{N}\left(1-2q^{(i)}(1-q^{(i)})\right)-\frac{1}{d}}\ ,
\end{align}
\begin{align*}\label{eq:states_pauli_example}
        &\davS(\Lambda^{\text{pauli}}(\psi^{\text{pauli}}), \psi^{\text{pauli}})=  \frac{1}{2}\sqrt{1-2\Pi_{i=1}^{N}q^{(i)}+\Pi_{i=1}^{N}(1-2q^{(i)}(1-q^{(i)}))}\ , \numberthis
\end{align*}
\begin{proof}
We start by analyzing the effects of Pauli noise on single-qubit Pauli eigenstate.
We first write $\ketbra{\pm r_i}{\pm r_i} =\frac{1}{2}\left(\iden\pm\sigma_{r_i}\right)$ and evaluate
\begin{align}
    \Lambda^{\text{pauli}}_i(\ketbra{\pm r_i}{\pm r_i}) = \frac{1}{2}\left(\iden \pm \left( (p^{(i)}_1+p^{(i)}_{r}-p^{(i)}_{k\neq r_i}-p^{(i)}_{l\neq r_i} \right)\right) = \frac{1}{2}\left(\iden \pm \left( (2(p^{(i)}_1+p^{(i)}_{r})-1\right)\right) = \frac{1}{2}\left(\iden \pm \left(2q^{(i)}-1\right)\right) \ ,
\end{align}
where $p^{(i)}_{k\neq r_i}$ and $p^{(i)}_{l\neq r_i}$ are error probabilities corresponding to two Pauli matrices that are not $\sigma_{r_i}$. 
We now notice that the above state has two eigenvalues which are $\frac{1}{2}(1\pm |2q^{(i)}-1|)$, which for assumed regime $q^{(i)}\geq \frac{1}{2}$ gives eigenvalues $q^{(i)}$ and $(1-q^{(i)})$.

To get Eq.~\eqref{eq:states_uniform_example} we refer to Lemma~\ref{lem:uniform_states} and use the fact that the purity of separable states is a product of purities.
For a single qubit, the purity of a noisy state is
$\tr{\left(\Lambda_i^{\text{pauli}}\left(\ketbra{\pm r_i}{\pm r_i}\right)^2\right)}= (q^{(i)})^2+(1-q^{(i)})^2=1-2q^{(i)}(1-q^{(i)})$, which for multiple qubits yields Eq~\eqref{eq:states_uniform_example}.

To get Eq.~\eqref{eq:states_pauli_example}, we first diagonalize all noisy Pauli states, getting global state represented as $\bigotimes_{i=1}^{N}\left( q^{(i)}\ketbra{0}{0}+(1-q^{(i)})\ketbra{1}{1}\right)$.
In this basis, the noiseless Pauli eigenstate is simply $\ketbra{0}{0}^{\otimes N}$ (note that both states are simultaneously diagonalizable).
Having this in mind, we want to decompose the distance between states $||\Lambda^{\text{pauli}}(\psi^{\text{pauli}})-\psi^{\text{pauli}})||_{HS}^2$ into parts that are easy to handle.
To this aim, we use the fact that for any states $\rho$ and $\tilde{\rho}$, the HS distance can be written as $||\rho-\tilde{\rho}||_{HS}^{2} = \tr{\rho^2}+\tr{\tilde{\rho}^2}-2\tr{\left(\rho\ \tilde{\rho}\right)}$.
In our case $\rho=\Lambda^{\text{pauli}}(\psi^{\text{pauli}})$ and $\tilde{\rho} = \psi^{\text{pauli}}$.
Since the Pauli state is pure we get $\tr{\tilde{\rho}^2}=1$, while the purity of $\rho$ was already calculated above.
The cross-term can be evaluated by recalling that in basis we consider Pauli eigenstate is simply $\ketbra{0}{0}^{\otimes N}$, we thus need to simply take the value of the first matrix element of $\rho$, obtaining $\tr{\left(\rho\ \tilde{\rho}\right)} = \Pi_{i}q^{(i)}$.
Summing up and inserting into the definition of average-case distance yields Eq.~\eqref{eq:states_pauli_example}.
\end{proof}
\end{exa}

We now consider a scenario where our target POVM is computational-basis measurement $\P$, and we wish to calculate its distance from some other POVM $\M$.
This choice is motivated by the fact that in quantum computing the computational-basis measurement is often a model for ideal detector  \cite{mike&ike}, and $\M$ can be thought of as its noisy implementation.
In particular, we considered a situation in which $\M = \clc \ \P$, where $\clc$ is a left-stochastic map, i.e., its columns' are probability distributions.
Such noise is equivalent to classical post-processing of ideal statistics (i.e., probabilities one would have obtained on $\P$), hence we call it classical noise.
This is a practically relevant scenario, as it has been experimentally observed that classical noise is a dominant type of readout noise in contemporary quantum devices based on superconducting qubits \cite{Maciejewski2020}.

We now define, in analogy to average-case quantum distance, the \textit{average-case classical distance} between POVMs $\M$ and $\N$
\begin{align}\label{eq:classical_distance}
    \dav^{\text{classical}}(\M, \N) &\coloneqq \expect{\ketbra{k}{k}}  \dtv\rbracket{\mathbf{p}\rbracket{\ketbra{k}{k},\M},\mathbf{p}\rbracket{\ketbra{k}{k},\N}}  \ ,
\end{align}
where by $\expect{\ketbra{k}{k}}$ we denote average over all \emph{classical} deterministic states $\ketbra{k}{k}$.
The above distance turns out to be a helpful tool in investigating some of the properties of average-case quantum distances for quantum measurements.
In our considerations about the distances between measurements, the following Lemma~\ref{lem:classical_distance_ineq} and Lemma~\ref{lem:classical_noise_ineq} proved useful.

\begin{lem}[Average-case quantum vs classical distance]\label{lem:classical_distance_ineq}
Let $\P$  be measurement in computational basis, and $\clc$ arbitrary stochastic map, i.e., $\sum_{i}T_{ij}=1$. 
Define POVM $\clc \P$ via $\rbracket{\clc \ \P}_i = \sum_{j}\clc_{ij}P_j$.
Then we have
\begin{align}\label{eq:classical_distance_ineq}
    \frac{1}{2}  \  \dav^{\mathrm{classical}}(\clc \P, \P) \leq \davM(\clc\  \P, \P) \ .
\end{align}
\begin{proof}

We start by directly computing classical distance from Eq.~\eqref{eq:classical_distance} 
\begin{align*}
    \dav^{\text{classical}}(\clc \P, \P) &=\frac{1}{2d}\sum_{k=1}^{d} \sum_{i=1}^{d}|\tr(\ketbra{k}{k}(\sum_{j}T_{ij}\ketbra{j}{j}-\ketbra{i}{i}))| = \frac{1}{2d}\sum_{k=1}^{d}\sum_{i=1}^{d}|T_{ik}-\delta_{k,i}| \\
    &=\frac{1}{2d}\sum_{k=1}^{d}(1-T_{kk}+\sum_{i\neq k}T_{ik}) = 
    \frac{1}{2d}\sum_{k=1}^{d}2(1-T_{kk})=    1-\frac{\tr(T)}{d} \ , \numberthis
\end{align*}
where we used the fact that $T$ is left-stochastic, hence $\sum_{i\neq k}T_{ik}=1-T_{kk}$.
Now we notice that
\begin{align}\label{eq:classical_noise_explicit}
    \davM(\clc \P, \P) = \frac{1}{2d} \sum_{i=1}^{d} \sqrt{(1-\clc_{ii})^2+(1-\sum_{j}T_{ij})^2+\sum_{k\neq i}T_{ik}^2} \ \geq \frac{1}{2}\  \frac{1}{d} \sum_{i=1}^{d} \sqrt{(1-T_{ii})^2} = \frac{1}{2}\ (1-\frac{\tr(T)}{d}) \ ,
\end{align}
thus the expression on the RHS of Eq.~\eqref{eq:classical_noise_explicit} is exactly equal to $\frac{1}{2}\dav^{\mathrm{classical}}(\clc \P, \P)$, which concludes the proof.
\end{proof}
\end{lem}
\begin{lem}[Distance of classical part of the measurement noise]\label{lem:classical_noise_ineq}
Let $\M$ be an arbitrary $d$-outcome POVM, and $\P$ measurement in the computational basis.
Decompose $\M$ as $\M = \clc \ \P + \Delta$, where $\clc\  \P$ is a POVM obtained by taking only diagonal elements of operators $\M$, i.e., $(\clc \ \P)_i \coloneqq \mathrm{diag}(M_i)$ .
Then we have
\begin{align}\label{eq:classical_noise_ineq}
      \davM(\clc\  \P, \P) \leq \davM(\M,\P) \ .
\end{align}
\begin{proof}
Consider an action of (the dual of) completely dephasing noise $\Lambda^{\dag}_{\mathrm{deph}}$ on POVMs' effects, namely $\Lambda^{\dag}_{\mathrm{deph}}(M_i) = (\clc \P)_i$ (this is because, to begin with, we defined $\clc \ \P$ as diagonal part of POVM $\M$). 
Since dephasing noise is unital and it preserves computational-basis measurement $\P$, the above property follows directly from data-processing inequality for unital pre-processing of quantum measurements proved in Lemma~\ref{lem:dpi_measurements}.
\end{proof}
\end{lem}

\begin{rem}\label{rem:mitigation}
We note that while the decomposition of POVM $\M= \clc \ \P + \Delta$ into diagonal and off-diagonal parts may seem arbitrary, it has been in fact previously used in the context of measurement error mitigation.
In particular, the $\clc \P$ can be interpreted as a "classical part" of the noise and if we are able to reconstruct $\clc$ (for example, using Diagonal Detector Tomography), we can use it to reduce the noise via classical post-processing of the statistics estimated on faulty detector $\M$ \cite{Maciejewski2021, Geller2020efficient,Bravyi2020mitigating}.
\end{rem}

\begin{corr}\label{corr:classical_distance}
By combining Lemma~\ref{lem:classical_distance_ineq} with Lemma~\ref{lem:classical_noise_ineq}, we immediately get that for any POVM decomposed into the diagonal and off-diagonal part as $\M = \clc \ \P + \Delta$, its distance from standard measurement can be bounded from below via
\begin{align}
    \davM(\M,\P) \geq  \frac{1}{2}  \  \dav^{\text{classical}}(\clc \P, \P) \ .
\end{align}
\end{corr}

Let us now consider a simplified scenario where the target POVM is the computational basis measurement, and its noisy version corresponds to local, symmetric classical noise.
\begin{exa}[Computational basis and local symmetric bitflip]\label{exa:classical_local_symmetric}
Let $\P$ denote measurement in computational basis and its noisy version $\clc^{\text{sym}} \P$  affected by noise $\clc = \otimes_{i=1}^{N}\Lambda^{(\mathrm{sym})}_{i}$, where $\Lambda^{(\mathrm{sym})}_{i}$ denotes local stochastic noise describing symmetric bitflip specified by parameter $p_i$ (bitflip error probability).
In this case, we have 
\begin{align}
    \davM\rbracket{\clc^{\text{sym}}\P,\P} &= \frac{1}{2} \sqrt{1-2\Pi_{i=1}^{N}\rbracket{1-p_i}+\Pi_{i=1}^{N}(1-2p_i(1-p_i))}, \label{eq:symmetric_bitflip_distance_average}\\
    \dop\rbracket{\clc^{\text{sym}}\P,\P} &= 1-\Pi_{i=1}^{N}\rbracket{1-p_i} \label{eq:symmetric_bitflip_distance_worst} \ ,  \\
     \davM(\clc^{\text{sym}} \P, \M^{\idenC}) &=    \frac{1}{2}\sqrt{\Pi_{i=1}^{N}(1-2p_i(1-p_i))-\frac{1}{d}} \label{eq:symmetric_bitflip_distance_average_uniform}\ ,
\end{align}
where $N$ is the number of qubits.
\begin{proof}
To obtain \eqref{eq:symmetric_bitflip_distance_average} we calculate explicitly 
\begin{align}
    \davM(\P,\clc\P) = \frac{1}{2d} \sum_{i=1}^{d} \sqrt{(1-\clc_{ii})^2+(1-\sum_{j}T_{ij})^2+\sum_{k\neq i}T_{ik}^2} =  \frac{1}{2d} \sum_{i=1}^{d} \sqrt{(1-\clc_{ii})^2+\sum_{k\neq i}T_{ik}^2} \ ,
\end{align}
where the first equality follows from the fact that $\clc$ is bistochastic.
Then we notice that for identical symmetric bitflip, each term on RHS is the same, namely for each $i$ we have
\begin{align}
    (1-\clc^{\text{sym}}_{ii})^2 +\sum_{k\neq i}(\clc^{\text{sym}}_{ik})^2 = 1-2\clc^{\text{sym}}_{ii}+\sum_{k}(\clc^{\text{sym}}_{ik})^2 \ .
\end{align}
Furthermore, from the product structure of $\clc$ it follows that 
\begin{align}\label{eq:additional_equality}
  \clc^{\text{sym}}_{kk}&=\Pi_{i=1}^{N}(1-p_i) \\ ,
  \sum_{k}(\clc^{\text{sym}}_{lk})^2 &= \Pi_{i=1}^{N}((1-p_i)^2+p_i^2) \ .
\end{align}
Summing over $i=1,\ \dots,\ d$ yields Eq.~\eqref{eq:symmetric_bitflip_distance_average}.

To compute \eqref{eq:symmetric_bitflip_distance_worst} we notice that in case of (any) stochastic noise $\clc$ affecting standard measurement we have
\begin{align}
    \dop(\clc \P,\P) = \max_{j}(1-T_{jj}) \ ,
\end{align}
which after substituting $T_{jj}$ from Eq.~\eqref{eq:additional_equality} yields Eq.~\eqref{eq:symmetric_bitflip_distance_worst}.

To obtain Eq.~\eqref{eq:symmetric_bitflip_distance_average_uniform}, we first notice that multiqubit symmetric bitflip is represented by a bistochastic map that does not change trace -- thus second term in Eq.~\eqref{eq:povms_uniform} vanishes.
Then we calculate explicitly the purity of $i$th effect as
\begin{align}
    \tr{\left(((\clc^{\text{sym}}\P)_i)^2\right)} = \sum_{k}(\clc^{\text{sym}}_{ik})^{2} \ .
\end{align}
Combining the above observations with Eq.~\eqref{eq:additional_equality} yields Eq.~\eqref{eq:symmetric_bitflip_distance_average_uniform}. 
\end{proof}
\end{exa}

\begin{lem}[Computational basis and local asymmetric bitflip]\label{lem:asymetric_bitflip}
Consider a noisy version $\clc^{\text{asym}} \P$ of computational basis measurement $\P$, where $\clc^{\text{asym}} = \otimes_{i=1}^{N}\clc^{\text{asym}}_i$ is a separable, \emph{asymmetric} stochastic map. 
For each qubit $i$, such map is characterized by two parameters, $p_i(1|0)$ and $p_i(0|1)$, specifying the probability of erroneously measuring 1 (0) if the input state was $\ket{0}$ ($\ket{1}$).
Define average error probability
\begin{align}\label{eq:average_bitflip_probability}
q^{\text{av}}_i=\frac{p_i(1|0)+p_i(0|1)}{2}\ ,
\end{align}
and corresponding symmetric bitflip map $\clc^{\text{av}}_i = (1-q^{\text{av}}_i)\iden+q^{\text{av}}_i \sigma_x$, together with global map $\clc^{\text{av}} = \bigotimes_{i=1}^{N}\clc^{\text{av}}_i$.
Then we have
\begin{align}
    \davM(\clc^{\text{av}} \P, \P) &\leq \davM(\clc^{\text{asym}} \P, \P)  \ , \\
        \davM(\clc^{\text{av}} \P, \M^{\idenC}) &\leq \davM(\clc^{\text{asym}} \P, \M^{\idenC})   \ .
\end{align}
\end{lem}
\begin{proof}
The proof uses data-processing inequality for unital channels and stochastic post-processing proved in Lemma~\ref{lem:dpi_measurements}.
The idea is to present a strategy that "symmetrizes" stochastic noise on each qubit via randomized measurements and post-processing (while not changing computational basis measurement $\P$ or trivial POVM $\M^{\idenC}$).
Consider a strategy that applies combinations of $X$ and $\iden$ gates uniformly at random just before measurement, and then applies a post-processing strategy that combines the outcomes of measurements to "undo" the effects of the unital channel.
Namely, for each qubit, if the applied gate was $\iden$ do nothing, and if it was $X$ then flip the outcome.
Working out Kraus operators for this process shows that it corresponds to CPTP, unital map.
Finally, it follows from a direct calculation that if the initial POVM was $\clc^{\text{asym}}\P$, now the implemented POVM is exactly $\clc^{\text{av}}\P$.
Clearly, such a strategy does not affect the computational basis measurement (nor a trivial POVM $\M^{\idenC}$).
Recalling Lemma`\ref{lem:dpi_measurements} concludes the proof.

We note that the above strategy was used for single-qubit error mitigation in Ref.~\cite{Oszmaniec19}, and more general multi-qubit versions were considered in the context of noise characterization and mitigation Refs.~\cite{van_den_Berg_2022readout,Smith_2021readout,Tang2022readout}.
\end{proof}

From the Lemma~\ref{lem:asymetric_bitflip} it follows that when studying the separation between asymmetric stochastic noise and ideal measurement in the computational basis, one can instead study symmetric noise with "average" error probability (Eq.~\eqref{eq:average_bitflip_probability}), which is easier to handle computationally.
The same holds for studying separation from a uniform distribution.
The usefulness of this comes from the fact that asymmetric bitflip is a more realistic model of measurement noise than symmetric bitflip, (see, e.g., \cite{Maciejewski2021,Geller2020efficient}).

We now consider a few interesting scenarios for distances between channels.
\begin{exa}[Two arbitrary state preparation channels]\label{exa:state_prep_channels}
Denote by $\Lambda_{\rho}$ and $\Lambda_{\sigma}$ the state preparation channels that regardless of the input state always prepare state $\rho \in \states(\H_{\dim})$ or $\sigma\in \states(\H_{\dim})$, respectively.
Then we have 
\begin{align}
    \davC(\Lambda_{\rho},\Lambda_{\sigma}) = \sqrt{1+\frac{1}{\dim}} \ \  \frac{1}{2} ||\rho-\sigma||_{\HS}  \ .
\end{align}
\end{exa}

\begin{exa}[Two arbitrary unitary channels]\label{exa:generic_unitary_channels}
Denote by $\Lambda_{U}$ and $\Lambda_{V}$ the unitary channels associated with unitaries $U$ and $V$, i.e., $\Lambda_{U}(\rho) = U\rho U^{\dag}$ for any state $\rho \in \states(\H_{\dim})$.
Then we have 
\begin{align}\label{eq:average_distance_unitaries}
    \davC(\Lambda_{U},\Lambda_{V}) = \sqrt{\frac{1}{2}\left(1-\frac{|\tr{\left(U^{\dag}V\right)}|^2}{d^2}\right)}  \ .
\end{align}
\end{exa}

\begin{exa}[Identity channel and separable unitary rotations]

Let $\idenC$ denote identity channel, and $\Lambda_{V}$ be unitary channel corresponding to separable rotation $V=\bigotimes_{j=1}^N \exp(i\ \mathbf{n}_j\cdot \boldsymbol{\sigma}\  \frac{\phi_j}{2} )$, where $|\mathbf{n}_j|=1$  and $\phi_j>0$.
Assume that $\sum_{j=1}^{N}\phi_j \leq \frac{\pi}{2}$.
Define $\phi_{\max}=\max_{j}\phi_j$ and $\phi_{\min}=\min_{j}\phi_j$.
Then we have
\begin{align}
    \davC(\idenC,\Lambda_{V}) &\leq \sqrt{N}\ \frac{\phi_{\max}}{\sqrt{8}}\ ,\\
    \ddiam(\idenC,\Lambda_{V}) &\geq\frac{1}{\sqrt{2}} N\phi_{\min} \ .
\end{align}
To obtain the above, we first note that since the distances are unitarily invariant, we can rotate each unitary so it is a phase shift gate with an angle $\phi_j$.
To get the first inequality, we calculate explicitly (see Example~\ref{exa:generic_unitary_channels}) $\davC(\idenC,\Lambda_{V}) = \sqrt{\frac{1}{2}(1-\prod_{j=1}^{N}\cos^2(\frac{\phi_j}{2}))}$.
Then we use inequality $\cos^2(\frac{\phi_j}{2})\leq \cos^2(\frac{\phi_{\max}}{2})$ for $\phi_j\in \sbracket{0,\pi}$, and employ inequalities $\cos(x)^2\geq1-x^2$ and $(1-x)^N\geq 1-Nx$.
To get the second inequality we calculate diamond norm explicitly $\ddiam(\idenC,\Lambda_{V})=2|\sin(\sum_{j=1}^{N}\frac{\phi_j}{2})|$, and employ inequality $|\sin(x)|\geq \frac{x}{2\sqrt{2}}$ for $x\in \sbracket{0,\frac{\pi}{2}}$.
\end{exa}
From derivations in the above example it follows that if we adopt the perspective of average-case statistical distinguishability, any local coherent noise (when the target operation is identity) can be viewed simply as a phase shift error.
Furthermore, for angles such that $\frac{\phi_{\max}}{\phi_{\min}}= O(1)$, we see that worst-case distance grows quadratically faster than average-case.

\begin{exa}\label{ex:uniform_separable_channels}[Separable Pauli noise in the middle of the circuit]
Consider separable Pauli channel $\Lambda^{\text{pauli}}$ defined in Example~\ref{ex:uniform_separable_states}.
Then we have
\begin{align}\label{eq:channels_uniform_example}
    \davC(\Lambda^{\text{pauli}}, \Lambda_{\text{dep}}) = \frac{1}{2}\sqrt{\Pi_{i=1}^{N}||\mathbf{p}^i||^2_2-\frac{1}{d^2}} \ ,
\end{align}
\begin{align}\label{eq:channels_pauli_example}
        \davC(\Lambda^{\text{pauli}}, \idenC) = \frac{1}{2}\sqrt{1+\Pi_{i=1}^{N}||\mathbf{p}^i||^2_2-2\Pi_{i=1}^{N}p^i_1} \ ,
\end{align}
where $||\mathbf{p}^{i}||_2^2 = \sum_{j}(p^i_j)^2$ is a Euclidean norm of the vector of noise coefficients on $i$th qubit.
\begin{proof}
To begin the proof, we notice that the Pauli noise is a mixed unitary channel and is thus unital.
Since both completely depolarizing and identity channels are unital as well, in both average-case distances the terms that relate to the action on maximally-mixed state equal $0$.
We are therefore left with the task of calculating Hilbert-Schmidt norms of relevant Choi matrices.

To show that Eq.~\eqref{eq:channels_uniform_example} holds, we note that the purity of a separable Choi state is a product of purities -- this follows from the fact that any Choi matrix of product channel is permutationally similar to a tensor product of Choi matrices of those channels.
We thus need to consider only single-qubit purity (note that this is analogous to proof for states in Example~\ref{ex:uniform_separable_states}).
Denote by $\Choi^{(i)}_{\text{pauli}}$ a Choi matrix of Pauli channel on qubit $i$.
By directly evaluating the action of that channel on operators of the form $\ketbra{k}{l}$ (recall the definition of Choi matrix) we explicitly write down matrix representation of $\Choi^{(i)}_{\text{pauli}}$ and calculate 
\begin{align}
    \tr{\left(\Choi^{(i)}_{\text{pauli}}\right)^2} = \frac{1}{4}\left( \tr{\left(\Lambda^{(i)}_{\text{pauli}}(\ketbra{0}{0})\right)^2}+\tr{\left(\Lambda^{(i)}_{\text{pauli}}(\ketbra{1}{1})\right)^2}+2\tr{\left(\Lambda^{(i)}_{\text{pauli}}\left(\ketbra{0}{1}\right)\left(\Lambda^{(i)}_{\text{pauli}}\left(\ketbra{0}{1}\right)\right)^{\dag}\right)}\right)
\end{align}
From direct evaluation, we get that
\begin{align}
    \tr\left(\Lambda^{(i)}_{\text{pauli}}(\ketbra{k}{k})\right)^2 = (p_1+p_{z_i})^2+(p_{x_i}+p_{y_i})^2
\end{align}
 and  
\begin{align}
\tr\left(\Lambda^{(i)}_{\text{pauli}}(\ketbra{k}{l})\left(\Lambda^{(i)}_{\text{pauli}}(\ketbra{k}{l})\right)^{\dag}\right) = (p_1-p_{z_i})^2+(p_{x_i}-p_{y_i})^2
\end{align}
for $k\neq l$.
Summing up everything we get that cross-terms cancel and  $||\Choi^{(i)}_{\text{pauli}}||^2_{HS}= \sum_{j}p_j^{(i)} = ||\mathbf{p^{(i)}}||_2^2$ which combined with Lemma~\ref{lem:uniform_channels} yields Eq.~\eqref{eq:channels_uniform_example}.

To get Eq.~\eqref{eq:channels_pauli_example} we follow an identical strategy as for Example~\ref{ex:uniform_separable_states}.
Namely, we recall the fact that for any states $\rho$ and $\tilde{\rho}$, the HS distance can be written as $||\rho-\tilde{\rho}||_{HS}^{2} = \tr{\rho^2}+\tr{\tilde{\rho}^2}-2\tr{\left(\rho\ \tilde{\rho}\right)}$.
Now in our case $\rho = \Choi_{\Lambda^{\text{pauli}}}$ and $\tilde{\rho} = \Choi_{\idenC}$.
The Choi of the identity channel is a maximally-entangled state, its purity is thus equal to 1, while the purity of the Choi of the noisy channel was already calculated above.
To evaluate cross-term, we note that it factorizes into a product of single-qubit terms (as for purity, it follows from the permutational equivalence between the Choi matrix of product channel and tensor product of Choi matrices), each of them being equal to
\begin{align}
    \tr{\left(\Choi^{(i)}_{\text{pauli}}\Choi^{(i)}_{\idenC}\right)} = \frac{1}{4}\sum_{k,l \in \cbracket{0,1}} \tr{\left(\Lambda\left(\ketbra{k}{l}\right)\ketbra{l}{k}\right)} \ .
\end{align}
This evaluates to
\begin{align}
    \tr{\left(\Lambda^{(i)}_{\text{pauli}}\left(\ketbra{k}{k}\right)\ketbra{k}{k}\right)} = p^{(i)}_1+p^{(i)}_{z_i} \ ,
\end{align}
and 
\begin{align}
    \tr{\left(\Lambda^{(i)}_{\text{pauli}}\left(\ketbra{k}{l}\right)\ketbra{k}{l}\right)} = p^{(i)}_1-p^{(i)}_{z_i} \ ,
\end{align}
for $k\neq l$.
Summing up we obtain $\tr{\left(\Choi^{(i)}_{\text{pauli}}\Choi_{\idenC}\right)} = p_{1}^{(i)}$.
Combining all of the above with the definition of average-case distance yields Eq.~\eqref{eq:channels_pauli_example}.
\end{proof}
\end{exa}

\section{Open problems}\label{sec:open_problems}

Our work leaves many interesting problems left for future research.
The first important question one can ask is how to estimate average-case quantum distances in an easy-to-implement setting.
A natural candidate seems to be randomized-benchmarking types of experiments, as they also employ unitary designs 
\cite{Easwar2010RB,yoshi2021}. 
It would be also very interesting to connect quantum average-case distances with commonly used figures of merit used to assess the quality of quantum devices.
Those include measures such as average fidelity \cite{Easwar2010RB} (which is perhaps the most widely used quality measure), unitarity of quantum channels \cite{Dirkse2019unitarity,Cirstoiu2021,yoshi2021}, or partitioned trace distances \cite{rastegin2010}.
We note that the partitioned trace distances share with the average-case distances the property of being non-increasing under unital channels, which might suggest a deeper connection between the two.
Furthermore, one can ask whether similar results can be obtained for different functions of output probability distributions such us classical fidelity or  $f$-divergences  \cite{Petz1985quasientropies,Jarzyna2020geometric}.
Another natural direction to pursue is to  obtain better constants that appear in bounds (for example by considering higher moments)  relating the average TV distance with the quantum average-case distance. Another straightforward research direction is to check how the average-case quantum distances compare with worst-case distances for small subsets of qubits in actual quantum devices. For example, for pairs of qubits full Quantum Process Tomography \cite{mike&ike} (or even Gate Set Tomography \cite{Nielsen2020GST}) is possible, therefore one would be able to calculate the distances directly from objects in question.
Such studies could provide some insight into what to expect from existing devices in worst and average-case scenarios.

\emph{Acknowledgements} 
We would like to thank Richard Kueng, Victor Albert and Ingo Roth for interesting discussions and comments.  We sincerely thank Susane Calegari for help with tables formatting and proofreading the manuscript. The authors acknowledge the financial support by  TEAM-NET project co-financed by EU within the Smart Growth Operational Programme (contract no.  POIR.04.04.00-00-17C1/18-00).

\bibliographystyle{apsrev4-2}
\bibliography{nisqDISTbib.bib}

\appendix
\begin{center}
\Large{\textbf{Appendix}}
\end{center}

Here we provide proofs of more technical results from the main part -- proof of Lemma~\ref{lem:curiousInequality2} and proofs of main Theorems~\ref{th:STATESav},~\ref{th:MEASav},~\ref{th:CHANNELSav} for \emph{approximate} $4$-designs.

\section{Proof of Lemma~\ref{lem:curiousInequality2}}\label{app:proof_curiousInequality2}

In what follows we prove Lemma~\ref{lem:curiousInequality2}, which we repeat here for Reader's convenience.

\begin{lem}[Repeated Lemma~\ref{lem:curiousInequality2}]\label{lem:curiousInequality2_appendix}
Let $X,Y\in\Herm(\H)$ be Hermitian operators acting on $\H\simeq\H$. 
Let $\Psym{k}$ denotes the orthogonal projector onto $k$-fold symmetrization of $\Hsym{k}\subset \H^{\ot k}$. We then have the following inequality
\begin{equation}\label{eq:projInequality2_app}
    \tr\left(X^{\ot2} \ot Y^{\ot2}\  \Psym{4}  \right)  \leq C \tr\left(X^{\ot2}\ \Psym{2}  \right)  \tr\left(Y^{\ot2}\ \Psym{2}  \right) \ ,\ \text{where}\ C=\frac{13}{6}   \ .
\end{equation}
\end{lem}

\begin{proof}
We begin by noting that, for Hermitian matrices $A$ and $B$ we have 
\begin{equation}
	\tr \left( (A^{\otimes 2} \otimes B^{\otimes 2}) (	\Psym{2} \otimes 	
	\Psym{2}) \right)= 
	\frac14 (\tr (A^2)  + (\tr (A))^2 )(\tr (B^2)  + (\tr (B))^2 ).
\end{equation}
We also have
\begin{equation}
	\begin{split}
		4! \ \tr \left( (A^{\otimes 2} B^{\otimes 2} ) \Psym{4} \right) &= 
		((\tr (A))^2 + \tr(A^2) ) ((\tr (B))^2 + \tr(B^2) ) \\
		&+
		4 \tr(A) \tr (B) \tr (AB) +  4 \tr(A)  \tr (A B^2) + 4 \tr (B) \tr (A^2 B)\\
		&+
		2 (\tr (AB))^2 + 2 \tr (A^2 B^2) + 2 \tr (ABAB).
	\end{split}
\end{equation}
Now we consider the following difference for, with arbitrary scalar parameter 
$c$
\begin{equation}\label{eqn:ineq}
	\begin{split}
		c & \ \tr \left( (A^{\otimes 2} \otimes B ^{\otimes 2}) (	\Psym{2} \otimes 	
		\Psym{2}) \right)
		-  
		\tr \left( (A^{\otimes 2}\otimes B^{\otimes 2}) \Psym{4} \right)
		= \\
		&=
		\frac{1}{4!} \Big(
		(6 c - 1 ) (\tr (A^2)  + (\tr (A))^2 )(\tr (B^2)  + (\tr (B))^2 ) \\
		&- 4 \tr(A) \tr (B) \tr (AB) - 4 \tr(A)  \tr (A B^2) - 4 \tr (B) \tr (A^2 B)\\
		&-
		2 (\tr (AB))^2 - 2 \tr (A^2 B^2) - 2 \tr (ABAB) 
		\Big).
	\end{split}
\end{equation}
Now we will bound the terms which occur above using standard inequalities, to get  
\begin{equation}
\begin{split}
	- 4 \tr (A) \tr (B) \tr (AB) & \geq - 4 |\tr (A)| |\tr (B)| 
	\sqrt{\tr(A^2)}\sqrt{\tr (B^2)}
	\geq -2 (|\tr (A) |^2 |\tr (B) |^2 + \tr (A^2) \tr (B^2) ) 
	;  \\
	- 4 \tr (A)  \tr (A B^2) & \geq - 4|\tr (A)| \sqrt{\tr (A^2)} \sqrt{\tr (B^4)} 
	\geq - 4|\tr (A)| \sqrt{\tr (A^2)} \tr (B^2) 
	\geq -2 (|\tr (A)|^2+ \tr (A^2))\tr (B^2)
	;   \\ 
	- 4 \tr (B) \tr (A^2 B) & \geq -4 |\tr (B)| \sqrt{\tr (B^2)} \sqrt{\tr (A^4)}
	\geq - 4|\tr (B)| \sqrt{\tr (B^2)} \tr (A^2) 
	\geq -2 (|\tr (B)|^2+ \tr (B^2))\tr (A^2) 
	;   \\
	- 2 (\tr (AB))^2 &\geq - 2 \tr (A^2) \tr (B^2) ; \\
	- 2 \tr (A^2 B^2)  &\geq - 2 \tr (A^2) \tr (B^2) ; \\
	- 2 \tr (ABAB)  &\geq - 2 \tr (A^2) \tr (B^2).  \\
\end{split}
\end{equation}
Combining above inequalities, we will determine the value of parameter $c$, for 
which~\eqref{eqn:ineq} is non-negative 
\begin{equation}
\begin{split}
	c & \ \tr (A^{\otimes 2} \otimes B^{\otimes 2}) (	\Psym{2} \otimes \Psym{2})
	-  
	\tr (A^{\otimes 2}\otimes B^{\otimes 2}) \Psym{4}
	\geq  \\
	&=
	\frac{1}{4!} \Big(
	(6 c - 1 ) (\tr A^2  + (\tr A)^2 )(\tr B^2  + (\tr B)^2 )  - 12 \tr A^2 \tr B^2\\
	&
	-2 |\tr A |^2 |\tr B |^2  -2 |\tr A|^2 \tr B^2  -2 |\tr B|^2 \tr A^2 
	\Big)\\
	&=
	\frac{1}{4!} \Big(
	(6c - 13) \tr A^2 \tr B^2
	+
	(6 c - 3 ) (\tr A^2 (\tr B)^2  + (\tr A)^2 \tr B^2  +(\tr A)^2 (\tr B)^2 ) 
	\Big).
\end{split}
\end{equation}
Note, that the above is larger than $0$ for $c \geq 13/6$.
\end{proof}

\section{Proofs of main theorems for $\delta$-approximate $4$-designs}\label{app:proofs_designs}

Here we outline the extension of proofs of Theorems~\ref{th:STATESav},~\ref{th:MEASav},~\ref{th:CHANNELSav} for approximate $4$-designs.

\subsection{Quantum states and measurements}

We will start with quantum states.
Let us consider $\delta$-approximate $4$-design $\nu$ (recall Section~\ref{sec:unitary_designs}), i.e., we have
\begin{equation}\label{eq:conjEQ}
    \left\| \T_{4,\nu} - \T_{4,\mu} \right\|_\diamond \leq \delta\ ,
\end{equation}
where $\mu$ is the Haar measure in $\mathrm{U}(\H_\dim)$ and  $\T_{4,\nu}$ is the quantum channel acting on $\H_\dim^{\ot 4}$ defined as $\mathcal{T}_{4,\nu}(A)=\int_{\mathrm{U}(\H_\dim)} d\nu(U) U^{\ot 4} A (U^{\dagger})^{\ot 4}$. For a measure $\nu=\lbrace{\nu_\alpha,U_\alpha\rbrace}$ on $\mathrm{U}(\H_\dim)$ let $\tilde{\nu}$ denote a measure supported on 'inverted gates' i.e. $\nu=\lbrace{\nu_\alpha,U^\dagger_\alpha\rbrace}$ (the generalization to non-discrete measures is straightforward). From the definition of the diamond norm and the identity $\mu=\tilde{\mu}$ it follows that \begin{equation}\label{eq:invarianceCONJUGATION}
    \left\| \T_{k,\nu} - \T_{k,\mu} \right\|_\diamond =  \left\| \T_{k,\tilde{\nu}} - \T_{k,\mu} \right\|_\diamond\ .
\end{equation}
Denote $X_{i,U} = \tr ( \ketbra{i}{i} U \Delta U^\dagger )$, where $\Delta=\rho-\sigma$ for two quantum states $\rho,\sigma\in \states(\H_{\dim})$ which we wish to compare. 
Using Berger's inequality (cf. Lemma~\ref{lem:curiousInequality1}) for every summand in the expression for the TV distance $\dtv(\p^{\rho,U},\p^{\sigma,U})$, where $U\sim\nu$ we get
\begin{equation}\label{eq:nonHAARfrac}
  \expect{U\sim\nu}    |X_{i,U}| \geq 
  \frac{\left( \expect{U\sim\nu} X_{i,U}^2 \right)^{3/2}}
  {\left( \expect{U\sim\nu} X_{i,U}^4  \right)^{1/2}}\ .
\end{equation}

Our goal is to compare the right-hand side of the above expression with its counterpart evaluated using the Haar measure $\mu$ i.e. : $\left( \expect{U\sim\mu} X_{i,U}^2 \right)^{\frac{3}{2}} \left( \expect{U\sim\mu} X_{i,U}^4  \right)^{-\frac{1}{2}} $. We begin with the lower bound for the numerator of \eqref{eq:nonHAARfrac}. 
\begin{equation}
\begin{split}
\expect{U\sim\nu} X_{i,U}^2 &\geq 
\expect{U\sim\mu} X_{i,U}^2 - \left| \tr (\T_{2,\mu} - \T_{2,\tilde{\nu}})[ \ketbra{i}{i}^{\ot 2}] \Delta^{\ot 2}\right| 
\\
&\geq
\frac{1}{d(d+1)} \tr(\Delta^{\ot 2}) - 
\left\| \tr (\T_{2,\mu} - \T_{2,\tilde{\nu}})[ \ketbra{i}{i}^{\ot 2}]\right\|_1
\| \Delta\|_{\infty}^2 \\
&\geq
\frac{1}{d(d+1)} \tr(\Delta^{\ot 2}) - \delta \|\Delta\|_{\infty}^2 
\\
&\geq
\frac{1}{d(d+1)} \tr(\Delta^{\ot 2}) (1 - d(d+1) \delta ).
\end{split}
\end{equation}
where we used standard inequalities $|\tr(AB)|\leq ||A||_1\ ||B||_{\infty}$, $||A||_1 \leq ||A||_{\diamond}$,  $||A||^2_{\infty}\leq ||A||^{2}_{\HS}$, the definition of the diamond norm and \eqref{eq:conjEQ}.

Next, we bound denominator from above using Lemma~\ref{lem:curiousInequality1} and reasoning analogous as before
\begin{equation}
\begin{split}
\expect{U\sim\nu} X_{i,U}^4 &\leq 
\expect{U\sim\mu} X_{i,U}^4 + \left| \tr (\T_{4,\mu} - \T_{4,\tilde{\nu}})[ \ketbra{i}{i}^{\ot 4}] \Delta^{\ot 4}\right|  \\
&\leq
C \left( \expect{U\sim\mu} X_{i,U}^2 \right)^2 + \delta \| \Delta\|_{\infty}^4 \\
&\leq
C \frac{\tr (\Delta^2)^2}{(d (d+1))^2} \left(1 + \frac{(d (d+1))^4  \delta}{C} \right) \ ,
\end{split}
\end{equation}
with $C=10.1$.

Combining above inequalities, we obtain that for $\delta$ approximate 4-desing, we have
\begin{align}
 \frac{\left( \expect{U\sim\nu} X_{i,U}^2 \right)^{3/2}}
  {\left( \expect{U\sim\nu} X_{i,U}^4  \right)^{1/2}}
\geq 
\tilde{\lstates}(\delta)
\frac{\left( \expect{U\sim\mu} X_{i,U}^2 \right)^{3/2}}
  {\left( \expect{U\sim\mu} X_{i,U}^4  \right)^{1/2}}
\end{align}
with
\begin{align}
    \tilde{\lstates}(\delta) = 
\frac{(1 - d(d+1) \delta)^{3/2}}{\left( 1 + \frac{\delta (d (d+1))^2}{C}\right)^{1/2}}
\geq \frac{(1 - 2d^2 \delta)^{3/2}}{\left( 1 + \frac{ 4d^4\delta}{C}\right)^{1/2}} \geq \frac{(1 - 2d^2 \delta)^{3/2}}{\left( 1 + 2d^4\delta\right)^{1/2}}
\end{align}
where we used the fact that $x(x+1)\leq 2x^2$ and $x^2(x+1)^2\leq4x^4$ for any $x\geq 1$, and $\frac{2}{C}=\frac{2}{10.1}<1$.

By setting $\delta = \frac{ \delta'}{2d^4}$, 
we obtain
\begin{align}\label{eq:states_lower_approx}
    \tilde{l}(\delta')\geq \sqrt{\frac{(1-\frac{\delta'}{d^2})^{3}}{1+\delta'}} \ \eqqcolon \lstates(\delta') \ .
\end{align}

Using analogous reasoning for bounding $\expect{U\sim\nu}    X^2_{i,U}$ from above, we obtain upper bound 
\begin{align}
    \expect{U\sim\nu}    |X_{i,U}| \leq  \tilde{\ustates}(\delta) \expect{U\sim\mu}  |X_{i,U}| \ ,
\end{align}
with 
\begin{equation}
    \tilde{\ustates}(\delta) = \left( 1 + d(d+1)\delta\right)^{1/2} \leq (1+2\delta d^2)^{1/2} = (1+\frac{\delta'}{d^2})^{1/2} \eqqcolon \ustates(\delta') ,
\end{equation}
where we used the fact that $x(x+1)\leq 2x^2$ for any $x\geq 1$.
This concludes the proof for quantum states.

For quantum measurements, we follow the analogous technique of proof.
For POVMs $\M$ and $\N$, each $\Delta_i=M_i-N_i$ will play a role of previous $\Delta$. 
The only difference will be that the second moment is equal to
\begin{equation}
  \expect{U\sim\nu} X_{i,U}^2  = \frac{1}{d(d+1)} (\tr (\Delta_i)^2 + \tr (\Delta_i^2)) ,
\end{equation}
because operators  $\Delta_i$ are generally not traceless.

\subsection{Quantum channels}

Let us now proceed to the proof of Theorem~\ref{th:CHANNELSav} for channels.
Denote $X_{i,V,U} = \tr ( \ketbra{i}{i} U \Delta(V\psi_0V^{\dagger}) U^\dagger )$ where $\Delta = \Lambda-\Gamma$ with two quantum channels $\Lambda,\Gamma \in \channels(\H_{\dim})$ that we are comparing. From the reasoning given in the preceding section (i.e. the proof of Theorem~\ref{th:STATESav} for approximate $4$-designs) we have that for $\delta=\delta'/(2d^4)$
 \begin{align}
  \lstates(\delta')\ \frac{\cstates}{2}   \| \Delta[\psi_V]  \|_\HS & \leq  \ \expect{U\sim\nu} \dtv(\p^{\Lambda,\psi_V,U},\p^{\Gamma,\psi_V,U})\ \leq  \ustates(\delta')\ \frac{\Cstates}{2}  \|\Delta [\psi_V]  \|_\HS\  \label{eq:chan_bounds_app} .
\end{align}
We will use the above bounds together with Jensen's  and Berger's inequality (applied for the function $Y_V\coloneq \| \Delta[\psi_V]  \|_\HS$ and $V\sim\nu$) to establish the desired result. 
We start by re-expressing the second and fourth moment of $Y_V$ in a convenient form :
\begin{equation}\label{eq:simplSECOND}
    \expect{V\sim\nu} Y_V^2  = \expect{V\sim \nu} \tr\left( \mathbb{S} \Delta^{\ot 2}[ \psi_V^{\ot 2}]\right)= 2 \tr\left( \Psym{2} \Delta^{\ot 2}[ \T_{2,\nu}(\psi_0^{\ot2})]\right)\ ,
\end{equation}
\begin{equation}\label{eq:simplfourth}
    \expect{V\sim\nu}Y_V^4  = \expect{V\sim \nu} \tr\left( \mathbb{S} \Delta^{\ot 2}[ \psi_V^{\ot 2}]\right)^2= 4 \tr\left( \Psym{2}\ot \Psym{2} \Delta^{\ot 4}[ \T_{4,\nu}(\psi_0^{\ot 4})]\right)\ ,
\end{equation}
where we have used the 'swap trick': $\tr(AB)=\tr(A\ot B \mathbb{S})$, and the fact that $\tr(\Delta[\psi_V])=0$. We start by using Eq.~\eqref{eq:simplSECOND} to derive an \emph{upper bound} on $\expect{V\sim\nu} Y_V^2$. 
From above for $\delta$-approximate $4$-design $\nu$
\begin{align}\label{eq:app:2nd_moment_channels_up}
      \expect{V\sim\nu} Y_V^2  &\leq \expect{V\sim\mu} Y_V^2 +2 \left|  \tr\left( \Psym{2} \Delta^{\ot 2}[ (\T_{2,\nu}-\T_{2,\mu})(\psi_0^{\ot2})]\right)\right| \\ 
  &= \expect{V\sim\mu} Y_V^2 +2 \left|  \tr\left( (\Delta^\dag)^{\ot 2}\left[\Psym{2}\right]  (\T_{2,\nu}-\T_{2,\mu})(\psi_0^{\ot2})\right)\right| \\
  &\leq  \expect{V\sim\mu} Y_V^2 +2   \left\| (\Delta^\dag)^{\ot 2}\left[\Psym{2}\right]   \right\|_\infty \delta\ , \label{eq:halfINEQ}
\end{align}
where we have used the definition of the dual of a super operator and utilized that $\delta$-approximate $4$-design $\nu$ is also $\delta$-approximate $2$-design. We proceed with bounding the operator norm of $(\Delta^\dag)^{\ot 2}\left[\Psym{2}\right] $ in terms of HS norm of the Jamio\l kowski-Choi state $\Choi_\Delta$:
\begin{align}
    \left\| (\Delta^\dag)^{\ot 2}\left[\Psym{2}\right]   \right\|_\infty &= \binom{d+1}{2}  \left\| (\Delta^\dag)^{\ot 2}\left[\expect{U\sim \mu}  \psi_U^{\ot 2} \right]   \right\|_\infty \leq \\
    &\leq \binom{d+1}{2}   \max_{\psi\in\pstates(\H)}\left\| (\Delta^\dag)^{\ot 2}\left[\psi^{\ot 2} \right]   \right\|_\infty= \binom{d+1}{2} \left(\max_{\psi,\phi\in\pstates(\H)} \tr(\phi\Delta(\psi))\right)^2  .
\end{align}
The result of double maximization can be upper bounded as follows:
\begin{equation}
    \max_{\psi,\phi\in\pstates(\H_\dim)} \tr(\phi\Delta(\psi))=  \max_{\psi,\phi\in\pstates(\H_\dim)} d \tr(\Choi_\Delta \phi \otimes \psi^T ) \leq d \|\Choi_\Delta\|_\HS\ . 
\end{equation}
Inserting this to \eqref{eq:halfINEQ} and recalling that $\expect{V\sim\mu} Y_V^2= (\Cchannels) ^2 (\|\Choi_\Delta\|_{\HS}^2+\tr(\Delta(\tau_{\dim})^2)$ with $\Cchannels = \frac{d}{d+1}$(cf. Eq. \eqref{eq:chan_bounds_app}) we get:
\begin{equation}
      \expect{V\sim\nu} Y_V^2 \leq  (\Cchannels)^2 (\|\Choi_\Delta\|_\HS^2+\tr(\Delta(\tau_\dim)^2) + d^5(d+1) \|\Choi_\Delta\|_\HS^2 \delta \leq \expect{V\sim\mu} Y_V^2\left(1+\frac{d^5(d+1)}{(\Cchannels)^2} \delta \right)\ . 
\end{equation}
Integrating both sides of the upper bound in Eq. \eqref{eq:chan_bounds_app} , using Jensen's inequality, and noting that $\frac{1}{2}\sqrt{\expect{V\sim\mu} Y_V^2}= \davC(\Lambda,\Gamma)$  yields
\begin{align}
\expect{V\sim\nu} \ \expect{U\sim\nu} \dtv(\p^{\Lambda,\psi_V,U},\p^{\Gamma,\psi_V,U}) & \leq 
\ustates(2d^4\ \delta)
\sqrt{1+\frac{d^5(d+1)}{(\Cchannels)^2} \delta}\ \frac{\Cchannels}{2}  \sqrt{\expect{V\sim\mu} Y_V^2} \ , \\
& = 
\ustates(2d^4\ \delta)\sqrt{1+d^3(d+1)^3\delta}
\ \Cchannels \davC(\Lambda,\Gamma)\ ,
 \\
& = 
\tilde{u}^{\ch}(\delta) \ 
\Cchannels \davC(\Lambda,\Gamma)\ ,
\end{align}
where we defined 
\begin{align}\label{eq:app_upper_approx_channels_first}
    \tilde{u}^{\ch}(\delta) \coloneqq \ustates(2d^4\ \delta)\ \sqrt{1+d^3(d+1)^3\delta} \leq \ustates(2d^4\ \delta)\ \sqrt{1+8d^6\delta} \ .
\end{align}
where in second step we used inequality $(x+a)\leq x(1+a)$ for $x,a\geq1$,

To get the lower bound, we will integrate LHS of Eq.~\eqref{eq:chan_bounds_app} and apply berger inequality. Proceeding analogously as before we obtain
\begin{equation}
    \expect{V\sim\nu} Y_V^2 \geq \expect{V\sim\mu} Y_V^2 -2   \left\| (\Delta^\dag)^{\ot 2}\left[\Psym{2}\right]   \right\|_\infty \delta\ \geq \  \expect{V\sim\mu} Y_V^2 (1 - d^3(d+1)^3 \delta)\ . 
\end{equation}
\begin{equation}\label{eq:app_foruth_mmoment_channels_bound}
    \expect{V\sim\nu} Y_V^4 \leq \expect{V\sim\mu} Y_V^4 +4   \left\| (\Delta^\dag)^{\ot 2}\left[\Psym{2}\right]   \right\|_\infty^2 \delta\ \leq \expect{V\sim\mu} Y_V^4  + d^4(d+1)^4 \|\Choi_\Delta\|_\HS^4 \delta \ 
\end{equation}

We now recall that fourth moment w.r.p. to Haar measure is bounded by (see Eq.~\eqref{eq:chanAUX2})
\begin{equation}\label{eq:app_additional_4th_moment_channels}
    \expect{V\sim\mu} Y_V^4 \leq  v \cdot \left(\expect{V\sim\mu} Y_V^2 \right)^2 =  v \cdot (\Cchannels)^4\ (2\davC(\Lambda,\Gamma))^4 \ ,
\end{equation}
with $v=\frac{\frac{13}{6}\binom{d+1}{2}^2}{\binom{d+3}{4}}$, where we used the fact that  $\expect{V\sim\mu} Y_V^2= (\Cchannels)^2\  (\|\Choi_\Delta\|_{\HS}^2+\tr(\Delta(\tau_{\dim})^2) = (\Cchannels)^2\ (2\davC(\Lambda,\Gamma))^2$.

We now note that $\|\Choi_\Delta\|_\HS^4 \leq \left(\|\Choi_\Delta\|_\HS^2+\tr\left(\Delta\left(\tau_{\dim}\right)^2\right)\right)^{2} = (2\davC(\Lambda,\Gamma))^4$, and combine it with Eq.~\eqref{eq:app_additional_4th_moment_channels} and Eq.~\eqref{eq:app_foruth_mmoment_channels_bound} to obtain
\begin{align}
    \expect{V\sim\nu} Y_V^4 &\leq 
    v \ (\Cchannels)^4\ (2\davC(\Lambda,\Gamma))^4 \ \left( 1 -\delta\ \frac{d^4(d+1)^4}{v \ (\Cchannels)^4}\right)\\
    &= v \ (\Cchannels)^4\ (2\davC(\Lambda,\Gamma))^4 \ \left( 1 +\delta\ \frac{(d+1)^7(d+2)(d+3)}{13d \ }\right)\  .
\end{align}
Similarly, we get 
\begin{align}
    \expect{V\sim\nu} Y_V^2 \geq (\Cchannels)^2\ (2\davC(\Lambda,\Gamma))^2\ \left(1-\delta\ d^3(d+1)^3\right)
\end{align}
Inserting the above to Berger's inequality yields
\begin{align}
    \frac{\left(\expect{V\sim\nu} Y_V^2\right)^{3/2}}{\left(\expect{V\sim\nu} Y_V^4\right)^{1/2}} \geq 2\ \tilde{b}_{\dim}\ \davC(\Lambda,\Gamma)\ \frac{\left(1-\delta\ d^3(d+1)^3\right)^{3/2}}{\left( 1 +\delta\ \frac{(d+1)^7(d+2)(d+3)}{13d \ , }\right)^{1/2}} 
\end{align}
where $\tilde{b}_{\dim} = \frac{d}{d+1}\ \sqrt{\frac{(d+2)(d+3)}{13d(d+1)}}$.

Now we integrate Eq.~\eqref{eq:chan_bounds_app} and combine with the above to obtain 
\begin{align}
    \expect{U\sim\nu} \dtv(\p^{\Lambda,\psi_V,U},\p^{\Gamma,\psi_V,U})\ \geq \cchannels\ \davC(\Lambda,\Gamma) \ \tilde{l}^{\ch}(\delta)  \ ,
\end{align}
where
\begin{align}
    \tilde{l}^{\ch}(\delta) &= \lstates(2d^4\ \delta)\ \frac{\left(1-\delta\ d^3(d+1)^3\right)^{3/2}}{\left( 1 +\delta\ \frac{(d+1)^7(d+2)(d+3)}{13d \ , }\right)^{1/2}} \\ 
    &\geq  \lstates(2d^4\ \delta)\ \frac{\left(1-\delta\ 8\ d^6\right)^{3/2}}{\left( 1 +\delta\ \frac{d^8\cdot 2^7\cdot 12}{13}\ \right)^{1/2}}\\
    &\geq   \lstates(2d^4\ \delta)\ \frac{\left(1-\delta\ 2^8\ d^6\right)^{3/2}}{\left( 1 +\delta\ (2d)^8\right)^{1/2}} \ , 
\end{align}
where we used inequality $(x+a)\leq x(1+a)$ for $x,a\geq 1$.

Now we set $ \delta \coloneqq \frac{\delta'}{(2d)^8}$ and get
\begin{align}
    \tilde{l}^{\ch}(\delta') \geq \lstates(\frac{\delta'}{2^7d^4})\ \frac{\left(1-\frac{\delta'}{d^2}\right)^{3/2}}{\left( 1 +\delta'\right)^{1/2}} \geq \lstates(\delta')\ \frac{\left(1-\frac{\delta'}{d^2}\right)^{3/2}}{\left( 1 +\delta'\right)^{1/2}} \geq \frac{\left(1-\frac{\delta'}{d^2}\right)^{3}}{ 1 +\delta'} \eqqcolon \lchannels(\delta') \ , 
\end{align}
where in second inequality we used the fact that $\lstates(x)$ is a decreasing function of $x\geq0$.

With set $\delta$, we also rewrite upper bound from Eq.~\eqref{eq:app_upper_approx_channels_first} as
\begin{align}
      \tilde{u}^{\ch}(\delta') \leq \ustates(\frac{\delta'}{2^7d^4})\sqrt{1+\frac{\delta'}{2^5d^2}} \leq \ustates(\delta)\sqrt{1+\frac{\delta'}{d^2}} \leq 1+\frac{\delta'}{d^2} \  \eqqcolon \uchannels(\delta')\ ,
\end{align}
where we used the fact that $\ustates(x)$ is an increasing function of $x\geq0$.

\end{document}